%% file: main.tex
\documentclass[11pt]{article}

\input{packages}
\input{commands}

\title{Independence-Number Parameterized Space Complexity for Directed Connectivity Certificate}

\author[(1)]{Ho-Lin Chen \thanks{holinchen\imageat{}ntu\imagedot{}edu\imagedot{}tw}}

\author[(2)]{Tsun-Ming Cheung \thanks{bencheung\imageat{}iis\imagedot{}sinica\imagedot{}edu\imagedot{}tw}}

\author[(2)]{Peng-Ting Lin \thanks{ptlin\imageat{}iis\imagedot{}sinica\imagedot{}edu\imagedot{}tw}}

\author[(2)]{Meng-Tsung Tsai \thanks{mttsai\imageat{}iis\imagedot{}sinica\imagedot{}edu\imagedot{}tw}}

\affil[(1)]{Department of Electrical Engineering, National Taiwan University}
\affil[(2)]{Institute of Information Science,
Academia Sinica}

\begin{document}

\maketitle

\begin{abstract}
We study the space complexity of computing a sparse subgraph of a directed graph that certifies connectivity in the streaming and distributed models. Formally, for a directed graph $G=(V,A)$ and $k\in \mathbb{N}$, a \emph{$k$-node (resp. $k$-arc) strong connectivity certificate} is a subgraph $H=(V,A')\subseteq G$ such that for every pair of distinct nodes $s,t\in V$, the number of pairwise internally node-disjoint (resp. pairwise arc distjoint) paths from $s$ to $t$ in $H$ is at least $k$ or the corresponding number in $G$.

Similar notions for undirected graphs have been studied extensively in various streaming models 
(Cheriyan et al. [SIAM J. Comput. 1993]; Sun and Woodruff [APPROX/RANDOM 2015]; Ahn et al. [SODA 2012]; Assadi and Shah [SOSA 2023]) with several efficient algorithms known. 
In contrast, streaming algorithms for directed connectivity are scarce due to the inherent hardness of directed connectivity problems: even $s$--$t$ reachability for a single pair $(s,t)$ in a digraph has a high space lower bound in standard streaming models. Consequently, any streaming algorithm for strong connectivity certificates would inherit these lower bounds for general graphs. 

In light of the inherent hardness of directed connectivity problems, several prior work focused on restricted graph classes, showing that several problems that are hard in general become efficiently solvable when the input graph is a tournament (i.e., a directed complete graph) (Chakrabarti et al. [SODA 2020]; Baweja, Jia, and Woddruff [ITCS 2022]), or close to a tournament in edit distance (Ghosh and Kuchlous [ESA 2024]).

Extending this line of work, our main result shows, at a qualitative level, that the streaming complexity of computing strong connectivity certificates and related directed connectivity problems is parameterized by the independence number of the input graph, demonstrating a continuum of hardness for directed graph connectivity problems. 
Moreover, our parameterization by independence number, an extrinsic graph parameter not inherent to connectivity problems, appears novel in the streaming context.

Quantitatively, for an $n$-node graph with independence number $\alpha$, we give $p$-pass randomized algorithms that compute a $k$-node or $k$-arc strong connectivity certificate of size $O(\alpha n)$ using $\tilde{O}(k^{1-1/p}\alpha n^{1+1/p})$ space in the insertion-only model, and $\tilde{O}(k^{1-O(1/\sqrt{p})}\alpha n^{1+O(1/\sqrt{p})})$ space in the turnstile model. In particular, the algorithm can be made deterministic for $k=1$ with the same pass and space usage in both models.

For the lower bound, we show that even when $k=1$, any $p$-pass streaming algorithm for a 1-node or 1-arc strong connectivity certificate in the insertion-only model requires $\Omega(\alpha n/p)$ space, which immediately implies the same bounds for all larger $k$ and in the turnstile model. 
To derive these lower bounds, we introduce the gadget-embedding tournament framework to construct direct-sum-type hard instances with a prescribed independence number, which is applicable to lower-bounding a wide range of directed graph problems.

Finally, in the distributed setting, we present protocols for computing directed strong connectivity certificates and strongly connected component (SCC) decompositions, with round complexity parameterized by the independence number.
\end{abstract}

\thispagestyle{empty}

\addtocontents{toc}{\protect\thispagestyle{empty}}
\tableofcontents
\thispagestyle{empty}
\setcounter{page}{0}

\newpage
\input{intro}
\input{notation}
\input{ConnCerts}

\input{algo}
\input{hardness}
\input{apps}
\input{distributed}

\bibliographystyle{alpha}
\bibliography{ref}

\appendix
\crefalias{section}{appendix}
\input{appendix}

\end{document}

%% file: packages.tex
\usepackage{amsmath,amssymb,amsthm}
\usepackage{thmtools}
\usepackage{mathtools}
\usepackage{graphicx} 
\usepackage[margin=1in]{geometry}
\usepackage[table]{xcolor}
\definecolor{darkgreen}{rgb}{0,0.5,0}
\usepackage[colorlinks=true, citecolor=darkgreen,backref=page]{hyperref}

\usepackage[textsize=footnotesize]{todonotes}
\usepackage{setspace}
\usepackage{microtype}
\usepackage{enumitem}
\usepackage{subcaption}
\usetikzlibrary{backgrounds,calc}
\usepackage{booktabs}
\usepackage{multirow} 
\usepackage{hhline}
\usepackage{authblk}
\usepackage[ruled,vlined]{algorithm2e}

\usepackage[capitalise]{cleveref}

%% file: commands.tex
\definecolor{purple}{HTML}{1A1AB3}
\colorlet{darkgreen}{green!50!black}

\DeclareEmphSequence{\color{purple}\itshape}

\AtBeginEnvironment{algorithm}{%
  \DeclareEmphSequence{\itshape}
}
\AtEndEnvironment{algorithm}{%
  \DeclareEmphSequence{\color{purple}\itshape}
}

\AtBeginEnvironment{thebibliography}{%
  \DeclareEmphSequence{\itshape}
}
\AtEndEnvironment{thebibliography}{%
  \DeclareEmphSequence{\color{purple}\itshape}
}

\theoremstyle{plain}
\newtheorem{theorem}{Theorem}[section]
\newtheorem{corollary}[theorem]{Corollary}
\newtheorem{proposition}[theorem]{Proposition}
\newtheorem{lemma}[theorem]{Lemma}
\newtheorem{fact}[theorem]{Fact}
\newtheorem{claim}[theorem]{Claim}
\newtheorem{observation}[theorem]{Observation}

\theoremstyle{definition}
\newtheorem{definition}[theorem]{Definition}

\crefname{claim}{Claim}{Claims}

\def\DISJ{\textup{\textsf{DISJ}}}
\def\AND{\textup{\textsf{AND}}}
\def\ind{\textup{\textsf{ind}}}
\newcommand{\set}[1]{\{#1\}}
\newcommand{\setB}[1]{\left\{#1\right\}}
\newcommand{\abs}[1]{|#1|}

\def\dbN{\mathbb{N}}
\def\Ex{\mathbb{E}}

\def\apx{\textup{\textsf{-apx}}}
\def\MaxW{\textup{\textsf{MaxW-}}}
\def\MinW{\textup{\textsf{MinW-}}}

\def\CYCLE{\textup{\textsf{CYCLE}}}

\def\SCC{\textup{\textsf{SCC}}}
\def\REACH{\textup{\textsf{REACH}}}
\def\TRIBES{\textup{\textsf{TRIBES}}}
\def\MSSS{\textup{\textsf{MSSS}}}
\def\HamCycle{\textup{\textsf{HamCycle}}}
\def\HamPath{\textup{\textsf{HamPath}}}
\def\PairReach{\textup{\textsf{PairReach}}}
\def\OutBranch{\textup{\textsf{OutBranch}}}
\def\InBranch{\textup{\textsf{InBranch}}}
\def\MaxColorTree{\textup{\textsf{MaxColourTree}}}
\def\ConnCert{\textup{\textsf{ConnCert}}}
\def\AConnCert{\textup{\textsf{AConnCert}}}
\def\NConnCert{\textup{\textsf{NConnCert}}}
\def\TopoSort{\textup{\textsf{TopoSort}}}
\def\MCC{\textup{\textsf{MCC}}}
\def\SAT{\textup{\textsf{2SAT}}}
\def\StrBridge{\textup{\textsf{StrongBridge}}}
\def\TranClos{\textup{\textsf{TranClosure}}}
\def\DisjBranch{\textup{\textsf{ArcDisjBranch}}}
\def\indBranch{\textup{\textsf{IndBranch}}}
\def\DistDom{\textup{\textsf{DistDom}}}

\def\calX{\mathcal{X}}
\def\calY{\mathcal{Y}}
\def\IC{\textup{\textsf{IC}}}

\def\Rcc{\textup{\textsf{R}}}
\def\arc{\curvearrowright}

\def\polylog{\textup{\textrm{polylog}}}

\newcommand{\plog}{\operatorname{polylog}}
\newcommand{\poly}{\operatorname{poly}}

\def\tdO{\tilde{O}}

\def\imageat{\includegraphics[scale=0.0234]{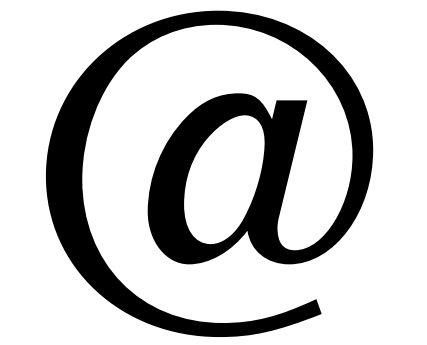}}
\def\imagedot{\includegraphics[scale=0.0234]{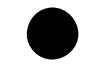}}

\def\reach{\rightsquigarrow}

%% file: intro.tex
\section{Introduction}\label{sec:intro}

We study sparse certificates for directed graph connectivity in streaming models. All graphs throughout this paper are simple unless stated otherwise. 
There are two notions of connectivity in graphs, namely \emph{node-connectivity} and \emph{arc-connectivity}, which in turn lead to two corresponding types of connectivity certificates.
For any digraph $G = (V, A)$ and a pair of distinct nodes $s,t \in V$, let $\kappa_{st}(G)$ (resp. $\lambda_{st}(G)$) be the maximum number of pairwise internally node-disjoint paths (resp. pairwise arc-disjoint paths) from $s$ to $t$. A \emph{$k$-node (resp. $k$-arc) strong connectivity certificate} of a digraph $G = (V, A)$ is defined to be a subgraph $H = (V, A')$ of $G$ with $A' \subseteq A$ such that for all pairs of distinct nodes $s, t \in V$ 
\[
    \kappa_{st}(H) \ge \min\{k, \kappa_{st}(G)\} \qquad\mbox{ (resp. } \lambda_{st}(H) \ge \min\{k, \lambda_{st}(G)\} \mbox{)}.
\]
When $k=1$, the notions of $1$-arc and $1$-node strong connectivity certificate coincide, however this equivalence does not hold for $k\ge 2$. The notion of strong connectivity certificates was first introduced by Nagamochi and Ibaraki \cite{NI92} for undirected graphs, and various notions of connectivity certificates for undirected graphs have been broadly studied in the annotated model ~\cite{GS24}, turnstile model~\cite{AhnGM12,GuhaMT15,AssadiS23} and insertion-only model~\cite{CheriyanKT93,EppsteinGIN97,SunW15}. 

Directed connectivity certificates support a broader range of applications since computing an undirected certificate is reducible to computing a directed one.  
The $k$-arc strong connectivity certificate enables efficient topological sorting, identification of strongly connected components, and applications to other computational problems, such as computing a minimum chain cover, finding an approximately minimum spanning strong subgraph, identifying $k$ arc-disjoint out-branchings, and detecting strong bridges.
The $k$-node strong connectivity certificate, on the other hand, allows us to test $k$-connectivity directly and to identify small node cuts and balanced separators. These applications will be discussed in detail in~\cref{sec:application_intro,sec:app}.

While directed certificates enable more applications than their undirected counterparts, obtaining them in the streaming model is significantly harder.
It is known that $s$--$t$ reachability problem in a digraph, where we are given distinct nodes $s,t\in V$ and must decide whether there is a directed path from $s$ to $t$, requires $\Omega(n^{2-o(1)})$ space\footnote{Following standard streaming conventions~(e.g. \cite{AY19,ChangFHT20,CC19}), we measure space lower bounds in bits and space upper bounds in words, with each word consisting of $O(\log n)$ bits.} for any $o(\sqrt{\log n})$-pass streaming algorithm~\cite{CKP+21}.
Since $G$ contains a directed path from $s$ to $t$ if and only if $\lambda_{st}(G)\ge 1$ or, equivalently, $\kappa_{st}(G)\ge 1$. It follows that, for any $k\ge 1$, computing a $k$-arc or a $k$-node strong connectivity certificate inherits the lower bound of the reachability problem. In contrast, for digraphs with independence number at most $\alpha$, Chen, Lin, and Tsai \cite{CLT25} showed that a $1$-arc strong connectivity certificate can be obtained via the $O(1/\varepsilon)$-pass $O(\alpha n^{1+\varepsilon})$-space algorithm for any $\varepsilon>0$ in the insertion-only streaming model. 
Crucially, this parameterized upper bound \cite{CLT25} aligns with the observation that in many previous hardness results \cite{CKP+21,GO16,BJW22,CGM+20}, the hard instances are layered graphs, which have linear independence number.

More broadly, there has been an emerging literature in studying graph streaming problems on restricted classes of digraphs.
One notable class of digraphs that has received considerable attention recently is the class of \emph{tournament} graphs, in which each pair of nodes shares exactly one arc, or equivalently digraphs with independence number 1.

Previous work on graph streaming problems for tournaments has been sporadic \cite{CFR10,BJW22}.
The first systematic study of the tractability of graph streaming problems on tournament graphs was initiated by Chakrabarti, Ghosh, McGregor, and Vorotnikova \cite{CGM+20}. They showed that several generally hard problems like topological sorting and acyclicity testing admit one-pass semi-streaming algorithms. 
Ghosh and Kuchlous \cite{GK24} built on this line of work and improved the streaming upper and lower bounds for tournaments, both in terms of bound optimality and the understanding of additional streaming problems. \cite{GK24} further showed that for digraphs $d$-close to tournaments, quantified by the number of arcs needed to be added or deleted to form a tournament, certain upper and lower bound results follow from the tournament case with an additive increase of $d$. 

Motivated by tractability results for digraphs with the minimal independence number (i.e., $1$), we establish upper and lower bounds in terms of the independence number for an extensive list of graph streaming problems, fully generalizing many prior results for tournaments. 
As a more natural measure of closeness to a tournament than the one proposed in \cite{GK24}, our parameterized bounds by independence number can be viewed as a complete interpolation result between tournaments and general graphs. 

The parameterized complexity framework to studying provably hard problems is not unfamiliar in the streaming literature. 
Examples of problems studied under the parameterized approach include vertex cover and maximum matching \cite{Chitnis2014,CC19,CCE+16}, treewidth \cite{CC19}, edge-dominating set \cite{Fafianie2014}, min-ones satisfiability \cite{ABB+19}, diameter and connectivity for undirected graphs \cite{OvL24}. 
Crucially, many of these studies follow the fixed-parameter tractability framework, in which the parameters under consideration are typically endemic to the problems, for example, the threshold for the parametrized vertex cover problem. Notable exceptions include the algorithms for approximate maximum matching based on arboricity in \cite{CCE+16} and topological sorting based on the independence number in \cite{CLT25}. 
Our work identifies an extrinsic parameter (independence number) for the streaming problems of interest, devises parameterized algorithms -- a less-studied direction, and derives parameterized streaming lower bounds -- an even more under-explored avenue. 

\subsection{Main Results}\label{sec:main-result}
Our main result begins with a graph-theoretic lemma. We show that a $k$-node strong connectivity certificate is also $k$-arc strong connectivity certificate:
\begin{lemma}\label{thm:node-vs-arc}
    For $k \in \dbN$ and a digraph $G$, if $H$ is a $k$-node strong connectivity certificate of $G$, then $H$ is also a $k$-arc strong connectivity certificate of $G$.
\end{lemma}
This result is not merely a trivial consequence of the implication from $k$-node-connectivity to $k$-arc-connectivity, since a strong connectivity certificate must correctly certify both the positive and negative cases of whether the connectivity exceeds the threshold $k$. We elaborate on this point further in \cref{sec:node-vs-arc}.

This implication is instrumental to our algorithm when a $k$-node strong connectivity certificate is easier to compute than a $k$-arc strong connectivity certificate, which is indeed the case here, as illustrated by our alternative streaming algorithm for $k$-arc strong connectivity certificates in \cref{sec:karc}.
With this lemma, we can restrict our attention to $k$-node strong connectivity certificates when designing algorithms, even for applications that require $k$-arc strong connectivity certificates. For this reason, from now on we may simple use a \emph{$k$-strong connectivity certificate} to refer to a $k$-node strong connectivity certificate for general-purpose applications.

For 1-connectivity, we extend the \emph{insertion-only} streaming algorithm of~\cite{CLT25} to $1$-strong connectivity certificate algorithms (\cref{algo:1conn}) in both the insertion-only and turnstile model. Moreover, we show that the resulting certificate has arboricity at most $\alpha+2$. 

\begin{theorem}\label{thm:main-1conn}
Let $G=(V,A)$ be an $n$-node digraph with independence number at most $\alpha$. For $p\in \dbN$, there exist $p$-pass deterministic algorithm for computing a $1$-strong connectivity certificate of $G$ with arboricity at most $\alpha+2$ using 
\begin{enumerate}[label=\textup{(\alph*)}]
    \item $O(\alpha n^{1+1/p})$ space in the insertion-only model;
    \item $O(\alpha n^{1+O(1/\sqrt{p}) })$ space in the turnstile model.
\end{enumerate}
\end{theorem}

Building on the $1$-strong connectivity certificate algorithm, our $k$-strong connectivity certificate algorithm (\cref{algo:kconn}) uses it as a building block in combination with the sampling technique of Guha, McGregor, and Tench~\cite{GuhaMT15} and the refinement by Assadi and Shah~\cite{AssadiS23} for undirected $k$-node connectivity certificates.
In more detail, the $k$-strong connectivity certificate algorithm samples sufficiently many node sets $V_1,\ldots,V_r \subseteq V(G)$ and computes $1$-strong connectivity certificates for the subgraph induced by each set. With a suitable sampling procedure, the union of these $r$ individual $1$-strong connectivity certificates forms a $k$-strong connectivity certificate. 
The $k$-strong connectivity certificate algorithm inherits the dependence on $\alpha$ in its space complexity, while the connectivity threshold $k$ appears in the space usage but not the number of passes.

\begin{theorem}\label{thm:main-kconn}
Let $G=(V,A)$ be an $n$-node digraph with independence number at most $\alpha$. For $p\in \dbN$ and $k\geq 2$, there exist $p$-pass randomized algorithm such that with probability $1 - 1/n^{\Omega(1)}$, a $k$-strong connectivity certificate of $G$ of size $O(k \alpha n\log n)$ is computed using:
\begin{enumerate}[label=\textup{(\alph*)}]
    \item $O(k^{1-1/p}\alpha n^{1+1/p}\log n)$ space in the insertion-only model; and
    \item $O(k^{1- O(1/\sqrt{p})}\alpha n^{1+ O(1/\sqrt{p}) }\log n)$ space in the turnstile model.
\end{enumerate}
\end{theorem}

Moreover, if the input digraph is guaranteed to be $k$-arc-strong, that is, $\lambda_{st}(G)\ge k$ for every distinct pair of nodes $s,t\in V(G)$, we obtain a deterministic algorithm that computes a $k$-arc connectivity certificate. Notably, the arboricity of the output certificate is linear in $k$, which is asymptotically tight among all the $k$-arc strong connectivity certificates of $G$.
\begin{theorem}\label{thm:main-kstrong}
    Let $k\in \dbN$, and $G=(V,A)$ be a \emph{$k$-arc-strong} $n$-node digraph with independence number at most $\alpha$. For $p\in \dbN$, there exist $kp$-pass deterministic algorithm for computing a $k$-strong connectivity certificate of $G$ with arboricity at most $k$ using 
    \begin{enumerate}[label=\textup{(\alph*)}]
        \item $O(k^2\alpha n^{1+1/p})$ space in the insertion-only model;
        \item $O(k^2\alpha n^{1+O(1/\sqrt{p}) })$ space in the turnstile model.
    \end{enumerate}
\end{theorem}

To complement our algorithmic results, we establish parameterized space lower bounds with respect to the independence number. In particular, any $p$-pass streaming algorithm for computing a $k$-node or $k$-arc strong connectivity certificate requires space usage with linear dependency in $\alpha$ and $n$ even when $k=1$.

To complement our algorithmic results, we establish parameterized space lower bounds with respect to the independence number. In particular, even for $k=1$, any streaming algorithm for computing a $k$-node or $k$-arc strong connectivity certificate must use space linear in both the independence number and the number of vertices.
\begin{theorem}\label{thm:ConnCert-intro}
    For the class of digraphs with $n$ nodes and independence number at most $\alpha$, 
    any $p$-pass streaming algorithm in the insertion-only model for computing a $1$-node or $1$-arc strong connectivity certificate requires $\Omega(\alpha n/p)$ space. 

    Consequently, in both insertion-only and turnstile models, any $p$-pass streaming algorithm for computing a $k$-node or $k$-arc strong connectivity certificate requires $\Omega(\alpha n/p)$ space.
\end{theorem}

To achieve the lower bounds stated in \cref{thm:ConnCert-intro}, we introduce a framework called \emph{gadget-embedding tournament}. 
This construction yields hard graph instances with a prescribed independence number while simultaneously embedding hard instances from the unconstrained setting.

We begin with an informal overview of the gadget-embedding tournament construction and present the formal definitions in \cref{subsec:GET}.
Given a list of $n/d$ size-$d$ digraphs $\vec{\Gamma} = (\Gamma_1,\ldots,\Gamma_{n/d})$, referred as gadgets,
we define the \emph{$\vec{\Gamma}$-embedding tournament of parameters $(n,d)$} as follows. The construction begins with a tournament on $n$ nodes $[n/d]\times [d]$ ordered lexicographically. For each $i\in [n/d]$, we delete all the arcs among the nodes in $\set{i}\times [d]$, and replace them with the arcs of $\Gamma_i$ under a suitable embedding. An immediate yet crucial observation is that the independence number of the resultant graph is at most $d$, since any cross-component node set must contain two neighbouring nodes, owing to the presence of the original tournament arcs. 

The parameterized lower bound for computing node or arc strong connectivity certificates (\cref{thm:ConnCert-intro}) now follows from a straightforward reduction to the $n$-pair reachability problem in the gadget-embedding tournament framework:
\begin{theorem}\label{thm:pairreach/tri-into}
    For the class of digraphs with $n$ nodes and independence number at most $\alpha$, any $p$-pass streaming algorithm for determining $n$-pair reachability or detecting directed triangles requires $\Omega(\alpha n/p)$ space.
\end{theorem}

The versatility of the gadget-embedding tournament framework allows many general streaming lower bounds to be lifted to the parameterized setting in an almost black-box manner. Moreover, its compatibility with embedding hard direct-sum problems makes it a powerful tool for proving strong parameterized bounds whenever suitable direct-sum theorems are available. Together, these properties yield new lower bounds parameterized by independence number for a range of well-known graph problems, extending many prior results that were specific to tournaments \cite{CGM+20,GK24}.
\begin{theorem}\label{thm:HamP/HamC-intro}
    For the class of digraphs with $n$ nodes and independence number at most $\alpha$, any $p$-pass streaming algorithm for finding a Hamiltonian path or cycle requires $\Omega\left(\frac{\alpha n}{p\log^2 \alpha}\right)$ space.
\end{theorem}
\begin{theorem}\label{thm:reach-intro}
    For the class of digraphs with $n$ nodes and independence number at most $\alpha$, any $p$-pass streaming algorithm for determining $s$--$t$ reachability requires 
    \begin{enumerate}[label=\textup{(\alph*)}]
        \item\label{thm:reach-intro-a} $\Omega(n/p)$ space when $\alpha=O(1)$;
        \item\label{thm:reach-intro-b} $\Omega(\alpha^{1/(2p)}n/p^{O(1)})$ space when $\alpha=\omega(1)$ and $p=O\left(\frac{\log \alpha}{\log\log \alpha}\right)$.
    \end{enumerate}
\end{theorem}


\def\marker{\textcolor{blue}{$\dagger$}}
\def\Imodel{\textcolor{darkgreen}{\textsf{\textbf{I}}}}
\def\Tmodel{\textcolor{orange}{\textsf{\textbf{T}}}}
\def\shading{\cellcolor{gray!30}}
\begin{table}
    \renewcommand{\arraystretch}{1.65}
    \centering
    {
    \crefname{theorem}{Thm.}{Thms.}
    \crefname{corollary}{Cor.}{Cors.}
    \crefname{lemma}{Lem.}{Lems.}
    \crefname{proposition}{Prop.}{Props.}
    
    \begin{tabular}{|*5{c|}}
        \hline
        \multirow{2}{*}{Problem} & \multicolumn{2}{c|}{$p$-pass Upper Bound} & \multicolumn{2}{c|}{$p$-pass Lower Bound} \\
        \hhline{~----}
        & Space & Reference & Space & Reference \\
        \hhline{*5{=}}
        \multirow{2}{*}{$1\text{-}\ConnCert^\alpha_n$} &  \Imodel: \shading$O\left(\alpha n^{1+1/p}\right)$ & \multirow{2}{*}{\cref{thm:main-1conn}} & \multirow{6}{*}{$\Omega(\alpha n/p)$} & \multirow{2}{*}{\cref{thm:ConnCert-intro}} \\
        & \shading\Tmodel: $O\left(\alpha n^{1+O(1/\sqrt{p})}\right)$ & & & \\
        \hhline{---~-}
        \multirow{2}{*}{$k\text{-}\ConnCert^\alpha_n$} & \Imodel: $\tdO\left(k^{1-1/p} \alpha n^{1+1/p)}\right)$ & \multirow{2}{*}{\cref{thm:main-kconn}} &  & \multirow{2}{*}{\cref{thm:ConnCert-intro}} \\
        & \Tmodel: $\tdO\left(k^{1-O(1/\sqrt{p})} \alpha n^{1+O(1/\sqrt{p})}\right)$ & & & \\
        \hhline{---~-}
        \multirow{2}{*}{$\TranClos^\alpha_n$} & \shading\Imodel: $O\left(\alpha n^{1+1/p}\right)$ & \multirow{2}{*}{\cref{thm:algo-tranclos}} &  & \multirow{2}{*}{\cref{cor:LB-tranclos}} \\
        & \shading\Tmodel: $O\left(\alpha n^{1+O(1/\sqrt{p})}\right)$& & & \\
        \hhline{-----}
        \multirow{2}{*}{$k\text{-}\DisjBranch_{n}^{\alpha}$} & \Imodel: $\tdO\left(k^{1-1/p}\alpha n^{1+1/p}\right)$ & \multirow{2}{*}{\cref{thm:algo-branch}} & \multirow{8}{*}{$\Omega\left(\frac{\alpha^{1/(2p)}n}{p^{O(1)}}\right)$~\marker} & \multirow{2}{*}{\cref{cor:LB-branching}} \\
        & \Tmodel: $\tdO\left(k^{1-O(1/\sqrt{p})}\alpha n^{1+O(1/\sqrt{p})}\right)$ & & & \\
        \hhline{---~-}
        {$\StrBridge_{n}^{\alpha}$} & \Imodel: $\tdO\left(\alpha n^{1+1/p}\right)$ & {\cref{thm:algo-bridge}} &  & {\cref{cor:LB-bridge}}  \\
        \hhline{-~-~-}
        $2\text{-}\indBranch^\alpha_n$ & \Tmodel: $\tdO\left(\alpha n^{1+O(1/\sqrt{p})}\right)$ & \cref{thm:algo-nodebranch} &  & \cref{cor:LB-nodebranching}\\
        \hhline{---~-}
        $\OutBranch_{n}^{\alpha}$ & \shading  & \cref{thm:algo-branch1} &  & \cref{cor:LB-branching} \\
        \hhline{-~-~-}
        $\InBranch_{n}^{\alpha}$ &\shading & \cref{thm:algo-branch1} &  & \cref{cor:LB-branching} \\
        \hhline{-~-~-}
        $\SCC_{n}^{\alpha}$ & \shading & \cref{thm:algo-topo-scc} &  & \cref{cor:LB-SCC} \\
        \hhline{-~-~-}
        $\SAT^\alpha_n$ & \shading\Imodel:~$O\left(\alpha n^{1+1/p}\right)$ & \cref{thm:algo-topo-scc} &  & \cref{cor:LB-2SAT} \\
        \hhline{-~---}
        $\MCC_{n}^{\alpha}$ &\shading\Tmodel:~$O\left(\alpha n^{1+O(1/\sqrt{p})}\right)$  & \cref{thm:algo-topo-scc} &  &  \\
        \hhline{-~-~~}
        $\TopoSort_{n}^{\alpha}$ &\shading  & \cref{thm:algo-topo-scc} &  &  \\
        \hhline{-~-~~}
        $2\apx$ $\MSSS_{n}^{\alpha}$ &\shading  & \cref{thm:algo-2apxMSSS} &  &  \\
        \hhline{-~-~~}
        $d$-$\DistDom^\alpha_n$ &\shading & \cref{thm:algo-dominat} & & \\
        \hhline{-----}
        $\PairReach_{n}^{\alpha}$ &  &  & $\Omega(\alpha n/p)$ & \cref{thm:pairreach/tri-into} \\
        \hhline{-~~--}
        $\HamPath_{n}^{\alpha}$ &  &  & \multirow{3}{*}{$\Omega\left(\frac{\alpha n}{p\log^2 \alpha}\right)$} & \cref{thm:HamP/HamC-intro} \\
        \hhline{-~~~-}
        $\HamCycle_{n}^{\alpha}$ &  &  &  & \cref{thm:HamP/HamC-intro} \\
        \hhline{-~~~-}
        $\MSSS_{n}^{\alpha}$ &  &  &  & \cref{cor:LB-msss} \\
        \hline
    \end{tabular}
    }
    \caption{
    Selected parameterized bounds in this work. \Imodel~indicates an insert-only model upper bound, and \Tmodel~indicates a turnstile model upper bound. The shaded cells indicate deterministic bounds, all other cells indicate randomized bounds.
    See \cref{apx:prob-def} for the definitions of the problems. \\  
    \marker~represents the lower bound following from \cref{thm:reach-intro}, see the theorem statement for the precise parameter restrictions.}
    \label{tab:bounds}
\end{table}

\subsubsection{Complexity-characterizing Parameters for Streaming Problems}
Adopting our gadget-embedding framework, we observe that many of our upper and lower bounds share similar asymptotic forms with dependency in $\alpha$. For example when $p=\plog(n)$, the upper and lower bounds for problems like $\TranClos^\alpha_n$ and $\SCC^\alpha_n$ take the same form of $\tdO(n)\cdot \poly(\alpha)$.

These lower bound results indicate that the independence number is a correct parameter for characterizing the space complexity of many graph streaming problems, particularly graph problems related to connectivity. 
Indeed, in many previous works establishing general lower bounds for the problems considered in this work \cite{CKP+21,GO16,BJW22,CGM+20}, the hard instances are layered graphs, which have linear independence number. This suggests that the hardness of these problems fundamentally arises from the presence of a large set of nodes with no adjacency information, further validating the independence number as a \emph{complexity-characterizing parameter} for a wide range of connectivity-related streaming problems. 


A similar phenomenon is observed in classical and parameterized algorithms, with treewidth serving as the characterizing parameter.
Many fundamental problems, such as independent set, dominating set, and max-cut, admit algorithms whose running time depends exponentially on the treewidth, and recent work by Lokshtanov, Marx, and Saurabh~\cite{lokshtanov2018known} shows that these algorithms are essentially optimal under the strong exponential time hypothesis. 
This provides a concrete example where a single structural parameter tightly captures the computational hardness of a broad class of graph problems.  

It is also worth noting that the independence number has also appeared as a key parameter in studying the reachability problem in non-streaming settings  \cite{nickelsen2002reachability, fradkin2015edge}. 
These results suggest a deeper connection between independence number and reachability, hinting that independence number might govern the intrinsic difficulty of many fundamental problems related to reachability in digraphs across different computation models.

Ultimately, we hope that this notion of a complexity-characterizing parameter may serve as a template for analyzing other streaming graph problems, where a single parameter precisely governs the boundary between tractability and hardness.

\subsection{Applications of Parameterized Results in Streaming and Distributed Models}
\label{sec:application_intro}
In~\cref{sec:app}, we present a list of applications of strong connectivity certificates. \cref{tab:bounds} summarizes the parameterized upper bounds and lower bounds of selected problems.

With a $k$-node strong connectivity certificate algorithm, the algorithms for a number of fundamental connectivity tasks in directed graphs become immediate. Testing whether the graph is $k$-strongly connected reduces to checking the property on the certificate itself. Furthermore, by Menger's theorem (\cref{thm:menger}), the certificate guarantees that, for any two distinct nodes $s$ and $t$, the existence of $k$ internally node-disjoint $s$--$t$ paths can be verified directly using the certificate. The certificate also preserves all node cuts of size less than $k$, thereby enabling efficient identification of small cuts and small balanced separators (of sizes up to $k-1$ nodes). 

The above list of graph problems has a long history. Menger's theorem \cite{Menger27} established the fundamental equivalence between cuts and disjoint paths, while minimum cuts and $k$-connectivity have been studied extensively since 1950s \cite{ford1956maximal,even1975network}. 
For balanced separators, the planar separator theorem of Lipton and Tarjan \cite{lipton1979separator} and the flow-based approximation framework of Leighton and Rao \cite{leighton1988approximate} represent the earliest results of this kind. Despite this extensive work, little is known in the streaming setting for directed graphs, largely due to the inherent hardness of $s$--$t$ reachability. 

A further obstacle is the difficulty of adapting algorithms from other models to the streaming setting. For instance, computing $k$ arc-disjoint spanning out-branchings can be formulated as a matroid intersection problem via Edmonds' branching theorem~\cite{Edmonds73}.
However, existing streaming algorithms for matroid intersection~\cite{CrouchS14,GargJS23,Ter25} cannot solve this problem even for $k=1$, since they only produce approximate solutions in the form of out-trees rather than a spanning out-branching. Our 
$k$-strong connectivity certificate algorithms circumvent many of these adaptation issues.

Beyond the streaming setting, our technique for constructing $k$-strong certificates by aggregating 1-strong certificates from specific subgraphs can be naturally extended to the CONGEST model in the distributed setting. Furthermore, we demonstrate that the independence number $\alpha$ serves as a critical parameter in the distributed complexity of SCC decomposition and topological sorting. We obtain efficient algorithms for these fundamental tasks in graphs where $\alpha$ is small.

\subsubsection{Comparison to Previous Works}
It is shown in~\cite{CLT25} that a $1$-strong connectivity certificate (a subgraph with the same transitive closure, in terms of~\cite{CLT25}) can be used to compute both a topological ordering and a minimum chain cover\footnote{Although not stated explicitly, Corollary 6 of~\cite{CLT25} implies a $O(1/\varepsilon)$-pass $O(\alpha n^{1+\varepsilon})$-space algorithm for any $\varepsilon>0$ that computes a minimum chain cover of a DAG.} for a directed acyclic graph. In addition,~\cite{CLT25} shows that a $1$-strong connectivity certificate suffices to obtain the SCC decomposition of a directed graph, which in turn enables solving 2-SAT via the SCC decomposition of the implication graph. The algorithm of~\cite{CLT25} for computing an $1$-strong connectivity certificate is designed for the insertion-only model, whereas our algorithms in~\cref{thm:main-1conn} extend this approach to the turnstile model. Moreover, the lower bounds presented in \cite{CLT25} are non-parameterized and thus not directly comparable to their algorithmic results. We complement their work by establishing parameterized lower bounds.
SCC decomposition of tournament graphs has been studied in~\cite{BJW22} and achieve a $p$-pass $O(n^{1+1/p})$-space algorithm, which can be considered a special case where $\alpha=1$ of~\cref{thm:main-1conn}.

\cite{GK24} has also studied strongly connected component (SCC) decomposition of tournament graphs and achieved a $1$-pass $\Theta (n\log{n})$-space algorithm. However, the algorithm of~\cite{GK24} relies on a property (\cite[Lemma 8]{GK24}) specifically for tournament graphs, which does not generalize to arbitrary directed graphs. \cite{CGM+20,CLT25} have studied topological sorting for random DAGs. The independence number of a dense random DAG is small (see~\cite{CLT25}), which falls into the space-efficient regime of~\cref{thm:main-kconn}.


%% file: notation.tex
\section{Notation and Preliminaries}\label{sec:notation}
For a positive integer $k$, we denote $[k]\coloneqq \set{1,\ldots, k}$. Throughout this work, we adopt the standard big-O notation in computer science. We use tilde notations to hide polylogarithmic factors, e.g. $\tilde{O}(f(n)) = O(f(n)\cdot \polylog n)$.
We use the following standard Chernoff bound:
\begin{proposition}[Chernoff bound] \label{prop:Chernoff}
    Suppose $X_1,\ldots,X_k$ are independent $0/1$-random variables. Let $X = \sum_{i=1}^k X_i$ and $\mu = \Ex[X]$. For any $\delta>0$,
    \[
        \Pr[X\leq (1-\delta)\mu] \leq \exp\left(-\frac{\delta^2 \mu}{2}\right) 
        \qquad\text{and}\qquad
        \Pr[X\geq (1+\delta)\mu] \leq \exp\left(-\frac{\delta^2 \mu}{2+\delta}\right).
    \]
\end{proposition}
We defer further notation and preliminaries to \cref{apx:notations}.

\subsection{Graph Theory}
We mention the frequently used or less conventional graph-theoretic notations here, and defer the rest to \cref{apx:graph}. For a digraph $G$, we use the notations $V(G)$ and $A(G)$ to denote its node set and arc set respectively. An arc from $u$ to $v$ is denoted by $(u,v)$, or $u\arc v$ in cases where the tuple notation becomes overly cumbersome. For two distinct nodes $u$ and $v$ in $G$, we say that $v$ is reachable from $u$ if there exists a $u$--$v$ path in $G$, which we sometimes use the shorthand $u\reach v$.

We use $\alpha(G)$ and $\ind(G)$ interchangeably to denote the independence number of $G$, and reserve the letter $\alpha$ for independence number.

We recall the notations for connectivity: for any two distinct nodes $s$ and $t$ of a digraph $G$,  $\kappa_{st}(G)$ and $\lambda_{st}(G)$ are the numbers of pairwise internal node-disjoint paths and pairwise arc-disjoint paths from $s$ to $t$, respectively.

Three classical graph-theoretic theorems central to our results on connectivity certificates are \emph{Edmonds’ branching theorem}, \emph{Menger’s theorem}, and \emph{Gallai--Milgram theorem}.
\begin{theorem}[Edmonds’ branching theorem~\cite{Edmonds73}] \label{thm:edmond}
    A digraph $G = (V, A)$ contains $k$ arc-disjoint spanning out-branchings rooted at a node $r \in V$ if and only if 
    \(
      |\delta^+_G(S)| \ge k
    \)
    for every cut $(S, V \setminus S)$ with $r \in S$.
\end{theorem}
\begin{theorem}[Menger's theorem \cite{Menger27}] \label{thm:menger}
    For a digraph $G$ and two distinct nodes $s,t$:
    \begin{itemize}
        \item $\kappa_{st}(G)$ equals the minimum size of an $s$--$t$ node cut, that is, a set of nodes $V'\subseteq V(G)\setminus\set{s,t}$ such that the graph $G\setminus V'$ contains no directed path from $s$ to $t$;
        \item $\lambda_{st}(G)$ equals the minimum size of an $s$--$t$ arc cut, that is, a set of arcs $A'\subseteq A(G)$ such that the graph $G \setminus A'$ contains no directed path from $s$ to $t$.
    \end{itemize}
\end{theorem}
\begin{theorem}[Gallai--Milgram theorem \cite{GM60}; see also \cite{Diestel12}]\label{thm:GallaiMilgram}
    Every digraph $G$ can be partitioned into $\alpha(G)$ node-independent directed paths.
\end{theorem}

%% file: ConnCerts.tex
\section{Node and Arc Strong Connectivity Certificates} \label{sec:node-vs-arc}
In this section, we consider the relationship between the notions of node and arc strong connectivity certificates.  
For any $s,t\in G$, a collection of pairwise internally node-disjoint paths is also pairwise arc-disjoint, so $\kappa_{st}(G) \leq \lambda_{st}(G)$ for any $s,t$. In particular, a $k$-node-connected graph is also $k$-arc-connected. However, the definition of strong connectivity certificates imposes stricter requirements: the certificate must preserve the connectivity value for pairs of nodes whose connectivity is below the threshold $k$. Thus, the inequality between $\kappa$ and $\lambda$ does not trivially guarantee that a $k$-node strong connectivity certificate is also a $k$-arc strong connectivity certificate.

As explained in \cref{sec:main-result}, the main reason it suffices to consider node-strong connectivity certificates is that, in our setting, a node-strong connectivity certificate can be computed more efficiently than an arc-strong connectivity certificate.
To the best of our knowledge, this lemma has not been stated explicitly in the literature, likely because the phenomenon of computational efficiency highlighted here has not been observed in other settings.

\begin{proof}[Proof of \cref{thm:node-vs-arc}]
Let $\mathcal{F}_H$ be the collection of subgraphs $H''$ such that $H\subseteq H''\subseteq G$ and $H''$ is a $k$-arc certificate of $G$. Note that $G\in \mathcal{F}_H$, therefore there exists a graph $H_* \in \mathcal{F}_H$ with the minimum number of arcs. The proof is completed if we prove that $H_*=H$.

Suppose the contrary that there exists an arc $(a,b)\in A(H_*)\setminus A(H)$. 
Note that $\kappa_{ab}(H_*) \geq \kappa_{ab}(H)+1$.
By definition $H$ is a $k$-node strong connectivity certificate of $G$, we have 
\[
    \min\set{\kappa_{ab}(G), k} \leq \kappa_{ab}(H)\leq \kappa_{ab}(H_*) - 1\leq \kappa_{ab}(G) -1,
\]
where the last inequality holds because $H'_*\subseteq G$. This forces that $\kappa_{ab}(H)\geq k$. 

Consider the graph $H_*' \coloneqq H_*\setminus\set{(a,b)}$.
Since $H_*$ is a supergraph of $H$, $H_*'$ contains at least $k$ arc-disjoint paths from $a$ to $b$ in $H$, none of which is the single arc $(a,b)$. Therefore, $\lambda_{ab}(H_*')\geq k$.

Suppose there exist two nodes $x,y$ such that $\ell\coloneqq \lambda_{xy}(H_*')$ satisfies
\(
    \ell<\min\set{k,\lambda_{xy}(G)}.
\)
Since $\ell<k$, by Menger's theorem \cref{thm:menger}, there exists an $(x,y)$-cut $(X,Y)$ of $H_*'$ with $x\in X$ and $y\in Y$,  such that the number of arcs from $X$ to $Y$ in $H_*'$ is exactly $\ell$. As $\lambda_{ab}(H_*')\geq k$, it cannot be the case that $a\in X$ and $b\in Y$. This implies that the number of arcs from $X$ to $Y$ in $H_*=H_*'\cup\set{(a,b)}$ is also exactly $\ell$. 

Applying Menger's theorem again, the cut $(X,Y)$ witnesses that
\[
    \lambda_{xy}(H_*)\leq \ell = \lambda_{xy}(H_*') < \min\set{k,\lambda_{xy}(G)},
\]
contradicting the assumption that $H_*$ is a $k$-arc strong connectivity certificate of $G$. This implies that $A(H_*)\setminus A(H) = \emptyset$ and thus $H_* = H$.
\end{proof}

%% file: algo.tex
\section{Streaming Algorithms for Strong Connectivity Certificates}\label{sec:algo}
In this section, we present streaming algorithms for computing sparse strong connectivity certificates in both the insertion-only and turnstile models. In \cref{sec:1arc}, we provide a deterministic algorithm for 1-strong connectivity certificates, and in \cref{sec:knode}, we reduce the computation of $k$-strong connectivity certificates to multiple independent instances of 1-strong connectivity certificates, yielding our randomized algorithm.

\subsection{$1$-Node Strong Connectivity Certificates} \label{sec:1arc} 
We begin with showing that if an $n$-node digraph has a small independence number, then one can apply a procedure we call \emph{transitive-closure-preserving pruning} to obtain a subgraph that preserves the transitive closure and whose arc set can be partitioned into a small number of directed forests (i.e. small arboricity). Our upper bound on arboricity improves a result of~\cite{CLT25} and will be used in~\cref{sec:tradeoff}.

\begin{lemma}[Transitive-closure-preserving pruning]\label{lem:bblock}
Let $G'$ and $G$ be digraphs on the same set of $n$ nodes with the same transitive closure $G^*$, where $G'$ has independence number $\alpha$ and $\ind(G)\geq \ind(G')$. 
Then $G$ contains a subgraph $H$ with arboricity at most $\alpha+2$ whose transitive closure equals $G^*$. 
\end{lemma}
\begin{proof}
By Gallai--Milgram theorem (\cref{thm:GallaiMilgram}), the nodes of $G'$ can be partitioned into at most $\alpha$ node-disjoint directed paths. Every directed path is a chain in the reachability relation, so a minimum chain cover of $G'$ consists of at most $\alpha$ chains. Since $G$ and $G'$ have the same transitive closure, they induce the same reachability relation, hence a minimum chain cover of $G$ also consists of at most $\alpha$ chains. Let $M$ be a minimum chain cover of $G$.

For each strongly connected component $C$ of $G$ that contains at least two nodes, choose an arbitrary node $r\in C$ and compute an in-branching $S^C_{in}$ and an out-branching $S^C_{out}$ both rooted at $r$ and spanning all nodes in $C$. Let $S=\bigcup_{C} (S^C_{in}\cup S^C_{out})$, and let $D$ be the acyclic subgraph of $G$ formed by removing all arcs within strongly connected components.

By construction, the transitive closure of $D \cup S$ is $G^*$. We now iteratively delete arcs from $D$ while preserving the transitive closure of $D \cup S$, until $D \cup S$ has arboricity $O(\alpha)$. Suppose $D$ contains a node $x$ with more than $\alpha$ out-neighbours. By the pigeonhole principle, at least two of the out-neighbours $u$ and $v$ lie on the same chain in $M$ (of length at least $2$). Assume without loss of generality that $u$ precedes $v$ on this chain.

As $G^*$ is the common transitive closure of $G$ and $D\cup S$, there exists a directed path $P$ from $u$ to $v$ in $D\cup S$. Note that $x\notin P$, otherwise $x$ and $u$ would be in the same strongly connected component of $G$, contradicting $(x,u)\in D$. Now, the arc $x\arc v$ can be removed from $D$ while preserving reachability as $x$ can reach $v$ by $x\arc u$ followed by $P$. 

Repeating this step ensures that every node in $D$ has at most $\alpha$ out-neighbours.
As $D$ is always acyclic, the final $D$ has arboricity at most $\alpha$. Since $S$ has arboricity at most $2$, setting $H$ as $D \cup S$ completes the proof.
\end{proof}

Transitive-closure-preserving pruning is a structured sparsification procedure in the sense that when applied to a graph $G$, the arboricity bound of the output graph $H$ automatically implies the bound $\abs{A(H)}\leq (\alpha(G)+2)\abs{V(G)}$.

\subsubsection{Insertion-Only Model}\label{sec:insertion}
We now present the deterministic 1-arc strong connectivity certificate algorithm in the insertion-only model based on \cref{lem:bblock}. 
Our algorithm follows \cite{CLT25} and uses a divide-and-conquer strategy. While it is presented recursively in \cref{algo:1conn} for ease of analysis, the recursion can be readily unravelled into an iterative implementation. The top-down perspective highlights the high-level structure of the algorithm: it first partitions the graph into progressively smaller components until a size-$\alpha$ chain cover can be computed by brute force using sufficiently small space; it then merges the resulting chains with additional inter-chain arcs and prunes the merged graph to control the graph size. 

\def\OneConnCertAlgo{\textup{\textsf{Recursive-1-ConnCert}}}
\begin{figure}[hbtp!]
\begin{algorithm}[H]
    \setlength{\leftskip}{0.75cm}
    \setlength{\parindent}{-0.75cm}
    \caption{\OneConnCertAlgo}\label{algo:1conn}
    \KwIn{A digraph $G=(V,A)$, size parameter $b$, recursion depth $p$}
    \If{$p=1$}{
        Store the whole graph in memory using one pass.
        
        Perform transitive-closure-preserving pruning (\cref{lem:bblock}) on $G$ offline and output the resultant graph.
    }
    \Else{        
    \underline{Offline Phase 1:} Equipartition $V$ into $b$ pairwise disjoint sets $V_1,\ldots,V_b$.
    
    \underline{Recursive Phase:} For each $i\in [b]$, run $\OneConnCertAlgo(G[V_i],b,p-1)$ in parallel. Denote $H_i$ the resultant graph obtained from $G[V_i]$.
    
    \underline{Offline Phase 2:} For each $i \in [b]$, compute a chain cover $M_i=\set{C^i_1,\ldots,C^i_{m(i)}}$ of $H_i$.
    
    \underline{Extra Streaming Phase:} Compute the following arc set $W$. For each $i\in [b]$, $j\in [m(i)]$, and $x\in V\setminus V_i$, find the first node $u^{i,j}_x$ in $C^i_j$ such that $x\arc u^{i,j}_x$, if one exists.\newline Let $W$ be the collection of all arcs $x\arc u^{i,j}_x$ for all $i,j,x$.

    \underline{Offline Phase 3:} Perform transitive-closure-preserving pruning (\cref{lem:bblock}) on $W \cup \bigcup_{i \in [b]} H_i$ to obtain the subgraph $G'$.
    }
\end{algorithm}
\caption{Deterministic algorithm for 1-strong connectivity certificate}
\end{figure}


\begin{lemma}\label{lem:alg-insert}
Let $G$ be an $n$-node digraph with independence number $\alpha$. Let $p\in \dbN$ and $b\coloneqq n^{1/p}$. \cref{algo:1conn} is a deterministic $p$-pass $O(\alpha n^{1+1/p})$-space streaming algorithm in the insertion-only model that outputs a $1$-node strong connectivity certificate of $G$ with arboricity at most $\alpha+2$. 
\end{lemma}
\begin{proof}    
    Each invocation of {\OneConnCertAlgo} use one additional pass. Including the single pass in the base case, the total number of passes is $p$ as claimed.

    For the space analysis, let $s(n,b,p)$ denote the space usage of a call of {\OneConnCertAlgo} with parameter $b$ and recursion depth $p$. For this analysis, we consider the more general setting in which the size parameter at the next recursion level $b'$ is not necessarily equal to $b$. Clearly $s(n,b,1) = O(n^2)$. 
    The function $s$ satisfies the recurrence
    \[
        s(n,b,p) = \max\setB{b\cdot s\left(\frac{n}{b}, b', p-1\right), \sum_{i=1}^b \abs{E(H_i)} + O(bn\cdot \max_i\set{m(i)})}.
    \]
    Gallai--Milgram theorem (\cref{thm:GallaiMilgram}) guarantees that $\max_i\set{m(i)} \leq \alpha$, and the arboricity upper bound proved in \cref{lem:bblock} implies that $\sum_{i=1}^b \abs{E(H_i)}\leq \alpha \sum_{i=1}^b \abs{E(V_i)} = \alpha n$. Therefore, the recurrence simplifies to 
    \begin{equation}
        s(n,b,p) = b \cdot \max\setB{s\left(\frac{n}{b}, b', p-1\right) , O(n\alpha)}. \label{eq:recurrence}
    \end{equation}
    It is easy to see from \cref{eq:recurrence} that it suffices to set $b'=b$ at all recursion levels to obtain the optimal choice of parameters, in which case the recurrence solves to
    \[
        s(n,b,p) = \max\setB{b^{p-1} \cdot s\left(\frac{n}{b^{p-1}}, b, 1\right) ,  O(bn\alpha)} = \max\setB{\frac{n^2}{b^{p-1}}, O(bn\alpha) }.
    \]
    Setting $b = n^{1/p}$ yields the claimed bound.

    Note that at each step of the algorithm, a subgraph of $G$ is maintained, therefore the independence number does not increase beyond $\ind(G) = \alpha$. By \cref{lem:bblock}, the arboricity of the output graph is at most $\alpha+2$. It remains to show that the output graph is indeed a 1-arc strong connectivity certificate of $G$ in order to conclude the correctness of the algorithm. 
    We claim that the transitive closures of $G$ and $H\coloneqq W \cup \bigcup_{i \in [b]} H_i$ are identical.

    Let $s,t\in V(G)$ and suppose there exists an $s$--$t$ path $P$ in $G$. We modify $P$ into a path in $H$ as follows. $P$ can be divided into contiguous blocks $B_1,\ldots,B_r$, where each block $B_k$ consists of the arcs from the same induced subgraph $G[V_{i(k)}]$ for some $i(k)\in [b]$. As each $H_{i}$ is recursively defined to be transitive-closure preserving with respect to $G[V_i]$, each subpath in the block $B_k$ can be replaced with a path in $H_{i(k)}$. 
    For $k=1,\ldots, r-1$, denote $x_k$ the last node in $B_k$ and $y_k$ the first node in $B_{k+1}$.
    Suppose $y_k$ lies in the chain $C^{i(k+1)}_\ell$ in $M_{i(k+1)}$. Then the arc $x_k\arc y_k$ in $P$ is replaced with the arc $x_k\arc u^{i(k+1),\ell}_{x_k}$ in $W$ and a path $u^{i(k+1),\ell}_{x_k}\reach y_k$ in $H_{i(k+1)}$. After this step, $P$ is replaced with a sequence of arcs in $H$. Removing all the self-loops gives a simple $s$--$t$ path in $H$. As $H$ is a subgraph of $G$, this justifies that $G$ and $H$ share the same transitive closure, and therefore $H$ is a 1-arc strong connectivity certificate of $G$.
\end{proof}

\subsubsection{Turnstile Model}\label{sec:turnstile}
By inspecting \cref{algo:1conn}, we see that to adapt our algorithm to the turnstile model, the only modification required is in the Extra Streaming Phase. Specifically, given a node $x$ and a chain, the turnstile algorithm requires a deletion-resilient method to find the out-neighbour of $x$ with the minimum rank according to the chain. 
This can be achieved by a multi-ary search method attributed to Munro and Paterson \cite{MP80}, who showed a tight $\Theta(n^{1/q})$ space bound for a $q$-pass deterministic turnstile algorithm for exact selection of the minimum element over an $n$-element universe.
In more detail, in each pass, the multi-ary search equipartitions the active search range into $n^{1/q}$ contiguous blocks, maintaining one counter per block suffices to identify the block containing a surviving minimum element, on which the search then recurses.

With the revised minimum element search subroutine, we modify \cref{algo:1conn} and update the corresponding pass and space usage analysis. Suppose the minimum element search is implemented with $q$ passes, and the recursion depth of the algorithm is $d$. The modified turnstile algorithm uses $dq+1$ passes, and the recurrence for the space usage function $s'(n,b,d)$ becomes
\[
    s'(n,b,d) = b \cdot \max\setB{s'\left(\frac{n}{b}, b, d-1\right) , O(n^{1+1/q}\alpha)}.
\]
with the same base case $s'(n,b,1) = O(n^2)$. Optimizing for the parameters $b,d,q$ yields a space upper bound of $O(\alpha n^{1+ O(1/\sqrt{p})})$\footnote{Note that for a prescribed number of passes $p$, the optimal choice of integral parameters $d$ and $q$ for $p=dq+1$ may not satisfy $d,q=\Theta(\sqrt{p})$. Nevertheless, one can choose a slightly smaller $p'$ to attain the exponent $1+O(1/\sqrt{p})$.} for a $p$-pass algorithm. Together with \cref{lem:alg-insert}, this completes the proof of \cref{thm:main-1conn}.

\subsection{$k$-Node Strong Connectivity Certificates}\label{sec:knode}
As alluded to in \cref{sec:intro}, the main idea is to reduce the task of computing a $k$-node strong connectivity certificate to multiple independent instances of computing $1$-node strong connectivity certificates, which can then be executed in parallel. Our proof adapts ideas developed for $k$-node connectivity certificates in undirected graphs~\cite{GuhaMT15,AssadiS23}. In contrast to the spanning trees in the undirected setting, we use $1$-strong connectivity certificates as the fundamental primitives for the directed case. Moreover, the key technical lemma \cref{lem:kvc-main} is needed in \cref{sec:congest} with a different sampling probability, so we present the argument in detail for completeness and clarity.

At a high level, the algorithm takes sufficiently many independent vertex samples and computes a $1$-strong connectivity certificate using \cref{algo:1conn} for each induced subgraph. With high probability, these certificates collectively form a cover with the following properties: for every arc $(u,v)$ with large $\kappa_{uv}(G)$, the cover preserves a directed path from $u$ to $v$ upon the removal of at most $k-1$ nodes, while for $(u,v)\in A(G)$ with small $\kappa_{uv}(G)$, at least one certificate retains the arc. Therefore, the individual $1$-strong connectivity certificates of the $r$ subgraphs are assembled into a $k$-strong connectivity certificate of $G$ with high probability.

We set up the notations in preparation for the proof.
For $\rho \in (0,1)$, let $V(\rho)$ be a random subset of $V(G)$ obtained by including each $v \in V(G)$ independently with probability $\rho$. For $r \in \dbN$, we let $V_1(\rho), V_2(\rho), \ldots, V_r(\rho)$ denote $r$ independent samples generated according to the distribution of $V(\rho)$. For each $i\in [r]$, let $Q_i$ be an arbitrary $1$-node strong connectivity certificate of the induced subgraph $G[V_i(\rho)]$.

In the following lemma, we prove that for $k\geq 2$, with suitable choices of the parameters $r$ and $\rho$, the union of 1-strong connectivity certificates of $r$ independent node samples forms a $k$-node strong connectivity certificate for $G$ with high probability. 
Specifically, for a fixed $k\geq 2$, we consider the range $\rho\in (0,1/k]$. For this range of $\rho$, $(1-\rho)^k\geq (1-1/k)^k\geq 1/4$. 
To simplify the calculations, we also assume that $n$ is sufficiently large so that constants and lower-order terms can be ignored in the bounds.
\begin{lemma}\label{lem:kvc-main}
Let $k\geq 2$, $\rho\in (0,1/k]$ and $n$ be sufficiently large. Write $r=\frac{192\lambda}{\rho^2}\log n$ for some $\lambda\geq 1$.
With probability $1 - n^{-\lambda}$, 
\(
Q \coloneqq \bigcup_{i \in [r]} Q_i
\)
forms a $k$-node strong connectivity certificate of $G$.
\end{lemma}

\cref{lem:kvc-main} follows from the claim about the relation between the connectivity of $G$ and $Q$. 
\begin{claim}\label{clm:kvc}
    Let $k\geq 2$, $\rho\in (0,1/k]$ and $n$ be sufficiently large. Write $r=\frac{192\lambda}{\rho^2}\log n$ for some $\lambda\geq 1$. With probability $1 - n^{-\lambda}$, every arc $(a,b)\in A(G)$ satisfies the following properties:
    \begin{enumerate}[label=\textup{(\alph*)}]
        \item\label{prop.a-kvc} If $\kappa_{ab}(G) \ge 2k$, then $\kappa_{ab}(Q) \ge k$;
        \item\label{prop.b-kvc} If $\kappa_{ab}(G) < 2k$, then $(a, b) \in A(Q)$.
    \end{enumerate}
\end{claim}
\begin{proof}[Proof of \cref{lem:kvc-main} assuming \cref{clm:kvc}]
    We show that $Q$ is a $k$-node connectivity certificate if every arc $(a,b)\in A(G)$ satisfies Properties \ref{prop.a-kvc} and \ref{prop.b-kvc}. Suppose for a contradiction that there exists a pair of nodes $s,t\in V(G)$ such that $\kappa_{st}(Q)<\min\set{k,\kappa_{st}(G)}$. By the contrapositive of Property \ref{prop.a-kvc}, $\kappa_{st}(G)<2k$ and thus $(s,t)\in A(Q)$ by Property \ref{prop.b-kvc}. As $Q$ is a subgraph of $G$, $(s,t)\in A(G)$.

    Let $Q' = Q \setminus \set{(s,t)}$ and $G' = G \setminus \set{(s,t)}$. Then $\kappa_{st}(Q') = \kappa_{st}(Q) - 1$ and $\kappa_{st}(G') = \kappa_{st}(G) - 1$, implying $\kappa_{st}(Q') < \min\set{k, \kappa_{st}(G')}$. By Menger’s theorem, there exists a set $X \subseteq V(Q') \setminus \set{s,t}$ of size $\kappa_{st}(Q')$ that forms an $(s,t)$-node cut in $Q'$. $V(Q')$ can be partitioned into $S \cup X \cup T$ with $s \in S$ and $t \in T$ such that there is no directed path from $S$ to $T$.  
    
    On the other hand, since $\kappa_{st}(Q') < \kappa_{st}(G')$, the size-$\kappa_{st}(Q')$ node set $X$ does not form an $(S,T)$-cut in $G'$. Hence, there exists an arc $(u,v) \in A(G') \setminus A(Q')$ with $u \in S$ and $v \in T$. Since $(u,v)\notin A(Q)$, by the contrapositive of Property~\ref{prop.b-kvc}, we have $\kappa_{uv}(G) \ge \theta$, and thus $\kappa_{uv}(Q) \ge k$ by Property~\ref{prop.a-kvc}. This contradicts the fact that $X$ is a node cut in $Q'$ separating $S$ and $T$.
\end{proof}

Now, we analyze Properties \ref{prop.a-kvc} and \ref{prop.b-kvc} separately. 
For two distinct nodes $a,b\in V(G)$, a subset of nodes $X\subseteq V(G)\setminus\set{a,b}$, and a subset of indices $R\subseteq [r]$, a directed path $P$ of $G$ without containing the arc $(a,b)$ is said to \emph{survive with respect to $a,b,X,R$} if for every arc $(u,v)$ in $P$, there exists an index $i\in R$ with $u,v\in V_i(\rho)$ and $V_i(\rho)\cap X=\emptyset$, in addition $a\notin V_i(\rho)$ if $a\notin\set{u,v}$, and $b\notin V_i(\rho)$ if $b\notin\set{u,v}$.
\begin{claim}\label{clm:survival}
    Let $P$ be a path from $s$ to $t$.
    If $P$ survives with respect to $a,b,X,R$, then $Q\setminus X$ contains an $s$--$t$ path without the arc $(a,b)$.
\end{claim}
\begin{proof}
    For any arc $(u,v)$ in the path $P$, there exists an index $i\in R$ such that $u,v\in V_i(\rho)$, in particular $(u,v)\in G[V_i(\rho)]$. Since $Q_i$ is a $1$-strong connectivity certificate of $G[V_i(\rho)]$, $Q_i$ contains a $u$--$v$ path $P_{uv}$. Moreover, at most one of $a$ and $b$ is contained in $V(P_{uv})$, so $(a,b)\notin A(P_{uv})$. Concatenating all paths $P_{uv}$'s and removing loops yields an $s$--$t$ path without the arc $(a,b)$.
\end{proof}


Fix an arc $(a,b)\in A(G)$ with $\kappa_{ab}(G)\geq 2k$, and a subset of at most $k-1$ nodes $X\subseteq V(G)\setminus \set{a,b}$. Excluding the length-1 $a-b$ path (i.e. $(a,b)$), by the definition of $\kappa_{ab}(G)$, $G\setminus X$ contains at least $k$ pairwise internally node-disjoint paths $P_1,\ldots,P_k$ from $a$ to $b$, each of length at least 2. We decompose each $P_j$ into the form $a\arc \pi_j\arc b$, where $\pi_j$ is a subpath with starting node $s_j$ and ending node $t_j$. Define 
\[
    \Gamma_{\textup{start}} \coloneqq \set{(a,s_j):j\in [k]} ;
    \qquad
    \Gamma_{\textup{mid}} \coloneqq \set{\pi_j:j\in [k]} ;
    \qquad
    \Gamma_{\textup{end}} \coloneqq \set{(t_j,b):j\in [k]} .
\]
For a collection $\Gamma$ of $k$ paths in $G$ and $R'\subseteq [r]$, we say $\Gamma$ is \emph{nice with respect to $a,b,X,R'$} if at least $2k/3+1$ paths in $\Gamma$ survive with respect to $a,b,X,R'$, and \emph{poor} otherwise.

Let $R_{\textup{start}}, R_{\textup{mid}}, R_{\textup{end}}\subseteq [r]$ be defined as follows:
\begin{align*}
    R_{\textup{start}} &\coloneqq \set{i\in [r]: a\in V_i(p) \text{ and } V_i(p)\cap (X\cup\set{b})=\emptyset}\\
    R_{\textup{mid}} &\coloneqq \set{i\in [r]: V_i(p)\cap (X\cup\set{a,b})=\emptyset}\\
    R_{\textup{start}} &\coloneqq \set{i\in [r]: b\in V_i(p) \text{ and } V_i(p)\cap (X\cup\set{a})=\emptyset}
\end{align*}

We first upper bound the probability that $\Gamma_{\textup{start}}$ is poor. By a symmetric argument, $\Gamma_{\textup{end}}$ is poor with respect to $a,b,X,R_{\textup{end}}$ with the same probability bound.
\begin{claim}\label{clm:R-start}
    For fixed choices of $a,b,X$, with probability $2n^{-6k\lambda}$, $\Gamma_{\textup{start}}$ is poor with respect to $a,b,X,R_{\textup{start}}$.
\end{claim}
\begin{proof}
    Note that 
    \(
        \Ex[ \abs{R_{\textup{start}}} ] = r\rho(1-\rho)^{\abs{X}+1}\geq r\rho/4.
    \) 
    By Chernoff Bound (\cref{prop:Chernoff}),
    \[
        \Pr\left[\abs{R_{\textup{start}}} \leq \frac{r\rho}{8} \right] 
        \leq \Pr\left[\abs{R_{\textup{start}}} \leq \frac{1}{2}\Ex[\abs{R_{\textup{start}}}]\right]
        \leq \exp(-r\rho/32) = n^{-6\lambda/\rho} \leq n^{-6k\lambda}.
    \]
    Now consider when $R_{\textup{start}}$ contains at least $r\rho/8$ elements. For each $j\in [k]$, denote $E_j$ the event that $s_j\notin V_i(\rho) \text{ for all } i\in R_{\textup{start}}$. Then
    \[
        \Pr[E_j] \leq (1-\rho)^{r\rho/8} \leq \exp(-r\rho^2/8) = n^{-24\lambda}.
    \]
    Let $q\coloneqq n^{-24\lambda}$.
    Note that all the $s_j$'s are distinct, so the events $\set{E_j}_{j=1}^k$ are mutually independent, by Chernoff bound we obtain
    \begin{align*}
        \Pr\left[\Gamma_{\textup{start}}\text{ is poor}~\middle|~\abs{R_{\textup{start}}}\geq  \frac{r\rho}{8}\right] 
        &\leq \sum_{t=k/3}^{k} \binom{k}{t} q^t (1-q)^{k-t}
        \leq q^{k/3}\cdot 2^k\leq n^{-8k\lambda} \cdot 2^k.
    \end{align*}
    Therefore, the probability that $\Gamma_{\textup{start}}$ is poor is at most 
    \(
        n^{-6k\lambda} + n^{-8k\lambda} \cdot 2^k \leq  2n^{-6k\lambda}. 
    \)
\end{proof}

Next, we turn to bound the probability that $R_{\textup{middle}}$ is poor.

\begin{claim}\label{clm:R-middle}
    For fixed choices of $a,b,X$, with probability at most $n^{-2k\lambda}$, $\Gamma_{\textup{middle}}$ is poor with respect to $a,b,X,R_{\textup{middle}}$.
\end{claim}
\begin{proof}
    We proceed as in \cref{clm:R-start}.
    Note that $\Ex[\abs{R_{\textup{middle}}}] = r(1-\rho)^{\abs{X}+2} \geq r(1-\rho)^{k+1} \geq r/8$. By Chernoff bound, we obtain
    \[
        \Pr\left[ \abs{ R_{\textup{middle}} } \leq \frac{r}{16} \right] 
        \leq \Pr\left[ \abs{ R_{\textup{middle}} } \leq \frac{1}{2}\Ex[\abs{ R_{\textup{middle}} }] \right]
        \leq \exp(-r/64) \leq n^{-{3\lambda}/\rho^2} \leq n^{-3k^2\lambda}.
    \]
    Now consider when $R_{\textup{middle}}$ contains at least $r/16$ elements. For each arc $(u,v)\in A(G)$ with $u,v\notin \set{a,b}$, and each index $i\in R_{\textup{middle}}$, the probability that $\set{u,v}\subseteq V_i(\rho)$ is $\rho^2$, and therefore
    \[
        \Pr[\set{u,v} \not\subseteq V_i(\rho) \text{ for all $i\in R_{\textup{middle}}$}] = (1-\rho^2)^{\abs{ R_{\textup{middle}} }} \leq \exp\left(-\frac{r\rho^2}{16}\right)\leq n^{-12\lambda}.
    \]
    For each $j\in [k]$, denote $E'_j$ the event that the path $\pi_j$ does not survive with respect to $a,b,X,R_{\textup{middle}}$.   
    By a union bound over all arcs in $\pi_j$,
    \[
        \Pr[E'_j]\leq \abs{\pi_j}\cdot n^{-12\lambda}\leq n\cdot n^{-12\lambda}\leq n^{-9\lambda}.
    \]

    Let $q' \coloneqq n^{-9\lambda}$.
    Recall that the paths $P_1,\ldots,P_k$ are pairwise internal node-disjoint, the node sets $V(\pi_1),\ldots,V(\pi_k)$ are pairwise disjoint and thus the events $\set{E'_j}_{j=1}^k$ are mutually independent. By Chernoff bound, we obtain
    \begin{align*}
        \Pr\left[\Gamma_{\textup{middle}}\text{ is poor}~\middle|~\abs{R_{\textup{middle}}}\geq  \frac{r}{16}\right] 
        &\leq \sum_{t=k/3}^{k} \binom{k}{t} (q')^t (1-q')^{k-t}
        \leq (q')^{k/3}\cdot 2^k = n^{-3k\lambda}\cdot 2^k.
    \end{align*}
    Therefore, the probability that $\Gamma_{\textup{middle}}$ is poor is at most $n^{-3k^2\lambda} + n^{-3k\lambda}\cdot 2^k \leq n^{-5k\lambda/2}$.
\end{proof}

\begin{proof}[Proof of \cref{lem:kvc-main}]
    We first prove that Property \ref{prop.a-kvc} is violated with probability at most $n^{-k\lambda}$. For each arc $(a,b)\in A(G)$ with $\kappa_{ab}(G)\geq 2k$, and a subset of at most $k-1$ nodes $X\subseteq V(G)\setminus \set{a,b}$, if all of $\Gamma_{\textup{start}}, \Gamma_{\textup{middle}}, \Gamma_{\textup{end}}$ are nice, at least one path $P_j$ survives with respect to $a,b,X,[r]$, by \cref{clm:survival} such a set $X$ certifies that $\kappa_{ab}(Q)\geq k$. By \cref{clm:R-start,clm:R-middle} and a union bound on $X,a,b$, Property \ref{prop.a-kvc} is violated with probability at most
    \begin{equation}        
        n^{k-1} \cdot n^2 \cdot (2\cdot 2n^{-6k\lambda} + n^{-5k\lambda/2}) \leq n^{-k\lambda}. \label{eq:prop.a-violation}
    \end{equation}

    Next, we prove that with high probability, $(a,b)\in Q$ whenever $\kappa_{ab}(G)<2k$. For such an arc $(a,b)$, by Menger's theorem, $G\setminus\set{(a,b)}$ contains an $a$--$b$ cut $W\subseteq V(G)\setminus \set{a,b}$ with $\abs{W}\leq 2k-2$. Here, we subtract the contribution of the arc $(a,b)$ to $\kappa_{ab}(G)$. $V(G)$ can be partitioned into $S\cup W\cup T$ with $a\in S$ and $b\in T$ such that there is no directed path from $S$ to $T$.

    For each $i\in[r]$, denote $E''_i$ the event that $V_i(\rho)$ contains both $a,b$ but no nodes from $X$. Then as $(1-\rho)^k\geq 1/4$ for $\rho\in(0,1/k]$,
    \[
        \Pr[E''_i] = \rho^2 (1-\rho)^{\abs{X}} \geq \rho^2 (1-\rho)^{2k} \geq \frac{\rho^2}{16}.
    \]
    Whenever $E_i$ occurs, the arc $(a,b)$ is the unique arc crossing from $S$ to $W\cup T$ in $G[V_i(\rho)]$, hence the certificate $Q_i$ must include $(a,b)$. Therefore
    \[
        \Pr[(a,b)\notin Q] \leq \left(1 - \frac{\rho^2}{16} \right)^r \leq \exp\left(- \frac{r\rho^2}{16} \right) =n^{-12\lambda},
    \]
    hence Property \ref{prop.b-kvc} is violated with probability at most $n^2 \cdot n^{-12\lambda} \leq n^{-10\lambda}$. Together with \cref{eq:prop.a-violation} concludes that both properties are satisfied with probability at least $1-n^{-\lambda}$.
\end{proof}

Up to this point, we have shown that for a sampling probability parameter $\rho\in(0,1/k]$, it suffices to take $O(\rho^{-2}\log n)$ independent node samples so that the $1$-strong connectivity certificates of the induced subgraphs combined form a $k$-strong connectivity certificate of $G$. For concreteness, our $k$-strong connectivity algorithm is described in \cref{algo:kconn}.

\begin{figure}[hbtp!]
\begin{algorithm}[H]
    \setlength{\leftskip}{0.75cm}
    \setlength{\parindent}{-0.75cm}
    \caption{$k$-strong Connectivity Certificate Algorithm}\label{algo:kconn}
    \KwIn{A digraph $G=(V,A)$, sampling probability $\rho$, depth parameter $p$}
    Set $r=O(\rho^{-2}\log n)$\;
    Sample $V_1(\rho),\ldots,V_r(\rho)$ independently\;
    \For{$i=1,\ldots,r$}{
        Apply {\OneConnCertAlgo} (\cref{algo:1conn}) of depth $p$ on $G[V_i(\rho)]$ to obtain a $1$-strong connectivity certificate $Q_i$\; 
    }
    \KwRet{$Q = \bigcup_{i=1}^r Q_i$}\;
\end{algorithm}
\caption{Randomized algorithm for $k$-node strong connectivity certificate}
\end{figure}

By Chernoff bound, we ensure that every node set $V_i(\rho)$ has $O(\rho n)$ nodes with high probability.
In such a case, as proven in \cref{thm:main-1conn}, {\OneConnCertAlgo} outputs an $O(\alpha\rho n)$-arc 1-strong connectivity certificate of $G[V_i(\rho)]$ in $p$ passes and $O(\alpha(\rho n)^{1+1/p})$ space for each $i\in [r]$. 
Observe that the $r$ computations of {\OneConnCertAlgo} on $G[V_i(\rho)]$'s can be executed in parallel. Therefore, the algorithm uses $p$ passes in total, outputs a $k$-strong connectivity certificate of size $r\cdot O(\alpha \rho n) = O(\alpha \rho^{-1} n\log n)$, and the overall space usage is 
\[
    r\cdot O(\alpha (\rho n)^{1+1/p}) = O(\alpha \rho^{-1+1/p} n^{1+1/p}\log n).
\]

The space is minimized by setting $\rho=1/k$, which yields a $p$-pass $O(k^{1-1/p}\alpha n^{1+1/p}\log n)$-space algorithm in the insertion-only model. By the same augmenting procedure used to adapt {\OneConnCertAlgo} to the turnstile model, we can modify \cref{algo:kconn} accordingly to obtain a $p$-pass $O(k^{1-O(1/\sqrt{p}) }\alpha n^{1+ O(1/\sqrt{p}) }\log n)$-space algorithm in the turnstile model.

\subsection{Deterministic Algorithm for $k$-Arc-Strong Digraphs}\label{sec:tradeoff}
For the case where a digraph $G$ is \emph{$k$-arc-strong}, we obtain a deterministic algorithm for computing a $k$-arc strong connectivity certificate of $G$, described below in \cref{algo:kstrong}. 

\begin{figure}[hbtp!]
\begin{algorithm}[H]
    \setlength{\leftskip}{0.75cm}
    \setlength{\parindent}{-0.75cm}
    \caption{$k$-arc Strong Connectivity Certificate Algorithm for $k$-arc Strong Digraphs }\label{algo:kstrong}
    \KwIn{A $k$-arc strong digraph $G=(V,A)$, depth parameter $p$}
    Set $R_0,R'_0\coloneqq G$ and $U_0,U'_0\coloneqq \emptyset$\;
    Choose an arbitrary $x\in V(G)$\;
    \For{$t=1,\ldots,k$}{
        Apply {\OneConnCertAlgo} (\cref{algo:1conn}) of depth $p$ on $R_{t-1}$ to obtain a $1$-strong connectivity certificate $H_t$\;
        Apply {\OneConnCertAlgo} (\cref{algo:1conn}) of depth $p$ on $R'_{t-1}$ to obtain a $1$-strong connectivity certificate $H'_t$\;
        Set $U_t\coloneqq U_{t-1} \cup H_t$ and $U'_t\coloneqq U'_{t-1} \cup H'_t$, which are $t$-arc strong since $G$ is $k$-arc-strong\;
        Extract $t$ pairwise arc-disjoint spanning out-branchings $\mathcal{B}_t$ rooted at $x$ from $U_t$\;
        Extract $t$ pairwise arc-disjoint spanning in-branchings $\mathcal{B}'_t$ rooted at $x$ from $U'_t$\;
        Update $R_t \coloneqq G\setminus \bigcup_{B\in \mathcal{B}_t} B$, $R'_t \coloneqq G\setminus \bigcup_{B'\in \mathcal{B}'_t} B'$\;
    }
    \KwRet{$H\coloneqq (\bigcup_{B\in \mathcal{B}_k} B) \cup (\bigcup_{B'\in \mathcal{B}'_k} B')$}\;
\end{algorithm}
\caption{Deterministic algorithm for $k$-arc strong connectivity certificate for $k$-arc strong digraphs}
\end{figure}

By Edmonds’ branching theorem (\cref{thm:edmond}), if a digraph $G$ is $k$-arc-strong, then for every node $x\in V(G)$, there exist $k$ pairwise arc-disjoint spanning out-branchings rooted $\set{B_i}^k_{i=1}$ at $x$ and $k$ pairwise arc-disjoint spanning in-branchings $\set{B'_i}^k_{i=1}$ rooted at $x$. Note that the out-branchings need not be arc-disjoint from the in-branchings; nevertheless, the union of these $2k$ spanning branchings forms a $k$-arc-strong subgraph and thus a $k$-arc strong connectivity for $G$.
\begin{claim} \label{clm:Hx}
    For any $x\in V(G)$, $H_x \coloneqq (\bigcup_{i=1}^k B_i) \cup (\bigcup_{i=1}^k B'_i)$ is $k$-arc-strong.
\end{claim}
\begin{proof}
    Let $(S,T)$ be any non-trivial cut of $V(G)$, where we assume $x\in S$ without loss of generality. For any $v\in T$, each $B_i$ contains a unique $x$--$v$ path that crosses from $S$ to $T$. Since $B_1,\ldots,B_k$ are pairwise arc-disjoint, the number of arcs crossing from $S$ to $T$ is at least $k$. Likewise by considering the in-branchings, the number of arcs crossing from $T$ to $S$ is at least $k$. This shows that every directed cut of $H_x$ has size at least $k$ in each direction, implying that $H_x$ is $k$-arc-strong.
\end{proof}

The correctness of \cref{algo:kstrong} now follows from \cref{clm:Hx} and the fact that $H_t\subseteq R_{t-1}$ is disjoint from $U_{t-1}$ (and similarly for $H_t',U_{t-1}'$) at each timestamp $t\in [k]$. Also, as $H$ is a union of $2k$ spanning branchings, the arboricity of $H$ is at most $2k$.

For the space usage, we remind the readers that the efficiency of {\OneConnCertAlgo} is parameterized by the independence number of the input graph. The following lemma shows that removing a low-degeneracy subgraph from a graph does not substantially increase the independence number of the remaining graph.

\begin{lemma}\label{lem:smallincrease}
Let $G$ be a digraph with independence number $\alpha$. Let $H$ be a subgraph of $G$ with degeneracy $\ell$. Then $\ind(G \setminus H)\leq (\ell+1)\alpha$.
\end{lemma}
\begin{proof}
Let $S$ be a maximum independent set in $G\setminus H$. Since $H$ is $\ell$-degenerate, every induced subgraph of $H$ is $\ell$-degenerate. In particular, $H[S]$ is $\ell$-degenerate and thus $(\ell+1)$-colourable by a greedy colouring. Each colour class $C_1,\ldots,C_{\ell+1}\subseteq S$ forms an independent set in $H[S]$, hence also an independent set in $G$. Since $G$ has independence number $\alpha$, it follows that
\[
    \alpha \ge \max_i\set{\abs{C_i}}\geq  |S|/(\ell+1),
\]
which completes the proof.
\end{proof}

Observe that for each $t\in [k]$, a union of $t$ branchings has degeneracy $O(t)$. By \cref{lem:smallincrease}, the residual graphs $R_t$ and $R'_t$ have independence number $O(t\alpha)$. Therefore, this iterative peeling scheme uses
\[
\sum_{t=1}^{k} O\left(t\alpha n^{1+1/p}\right)
= O\left(k^{2}\alpha n^{1+1/p}\right)
\]
space in the insertion-only model. Notice that the iterative process of \cref{algo:kstrong} has to be implemented sequentially, so the pass usage is $kp$.

The discussion for the turnstile model follows from \cref{sec:turnstile}, and it is easy to see that the space usage of \cref{algo:kstrong} in the turnstile model is $O\left(k^{2}\alpha n^{1+O(1/\sqrt{p})}\right)$. This completes the proof of \cref{thm:main-kstrong}.

%% file: hardness.tex
\section{Parameterized Lower Bounds}\label{sec:hardness}
Our lower bound for digraph streaming problems follows the standard reduction-based paradigm from communication complexity. As noted in \cref{sec:intro}, our hard instance construction is based on the framework of gadget-embedding tournament.

\subsection{Gadget-embedding Tournament}\label{subsec:GET}
For $n\in \dbN$ and $d\in [n]$, and $n/d$ size-$d$ gadgets $\vec{\Gamma} = (\Gamma_1,\ldots,\Gamma_{n/d})$, a \emph{$\vec{\Gamma}$-embedding tournament of parameters $(n,d)$} is a digraph $G^{n,d}_{\vec{\Gamma}}$ constructed as follows:
\begin{itemize}
    \item The node set $V$ is $[n/d]\times [d]$, which we denote $V^{(i)} \coloneqq \set{(i,j):j\in [d]}$ the $i$-th component;
    \item For each $i\in [n/d]$, embed $\Gamma_i$ into the node part $V^{(i)}$, typically under the canonical embedding $\mathcal{E}_i:V(\Gamma_i)\to V^{(i)}$ defined by $\mathcal{E}_i(v_j) = (i,j)$;
    \item For each $1\leq i<i'\leq n/d$, create the arc $(i,j)\arc (i',j')$ for each $j,j'\in [d]$.
\end{itemize}
We remind the readers that the resultant graph is generally not a tournament anymore.
Due to the presence of inter-component arcs, the following important observation is immediate.
\begin{observation}
    For any list of size-$d$ gadgets $\vec{\Gamma}$, the independence number of $G^{n,d}_{\vec{\Gamma}}$ is at most $d$.
\end{observation}
Another crucial property of the gadget-embedding tournament is the following observation:
\begin{observation} \label{obs:ordering}
    In a gadget-embedding tournament, if $u$ and $v$ belong to the $i$-th component and the $i'$-th component respectively for some $i>i'$, then $u\not\arc v$.
\end{observation}
This fact reflects that each inter-component edge carries little information, and each gadget component remains minimally interactive. This naturally leads us to consider decomposable functions (\cref{def:decomposable}) for the communication complexity lower bound of the new instance.
These two observations allow us to design a hard, structured instance with constrained independence number based on existing hard instance constructions.

\subsection{Parameterized Lower Bounds through Disjointness}
As a starting point, we present several parameterized lower bounds through reductions from the well-known hard communication problem of disjointness.

Recall that the two-party $m$-bit boolean disjointness problem $\DISJ_m$ is defined by $\DISJ_m(x,y)=0$ iff there exists $i\in [m]$ such that $x_i=y_i=1$. Following the randomized communication complexity lower bound, we have the following streaming lower bound.
\begin{fact}[\cite{KN96,Rou16}] \label{fact:DISJ}
    $\Rcc(\DISJ_m) =\Omega(m)$.
    Consequently, for any streaming problem reduced from $\DISJ_m$, any $p$-pass streaming algorithm for the problem requires $\Omega(m/p)$ space.
\end{fact}
With slight modifications to the now-standard disjointness reduction for undirected triangle counting in \cite{BKS02,BOV13}, we build ``disjointness gadgets'' to prove hardness for various graph streaming problems. For a parameter $m\in \dbN$, each disjointness gadget constructs a digraph on $3m$ nodes from two strings $x,y\in\set{0,1}^{m\times m}$. The set of nodes is $V=\set{a_i:i\in [m]}\cup \set{b_i:i\in [m]}\cup \set{c_i:i\in [m]}$, and the arcs are constructed as follows:
\def\gplain{\Gamma_{\textup{\textsf{plain}}}}
\def\gtri{\Gamma_{\triangle}}
\begin{itemize}
    \item For the plain gadget $\gplain^{x,y}$, for each $i,j\in [m]$, add $a_i\arc b_j$ iff $x_{ij}=1$, and $b_j\arc c_i$ iff $y_{ij}=1$;
    \item For the triangle gadget $\gtri^{x,y}$, for each $i,j\in [m]$, add $a_i\arc b_j$ if $x_{ij}=1$ and $b_j\arc a_i$ otherwise, and likewise add $b_j\arc c_i$ if $y_{ij}=1$ and $c_i\arc b_j$ otherwise; lastly, add $c_i\arc a_i$ for each $i\in [m]$.
\end{itemize}

\def\lwidth{3pt}
\begin{figure}
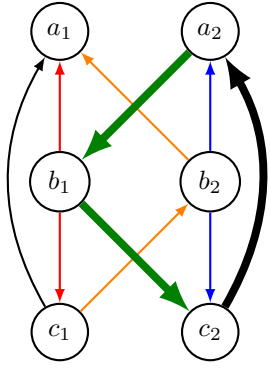
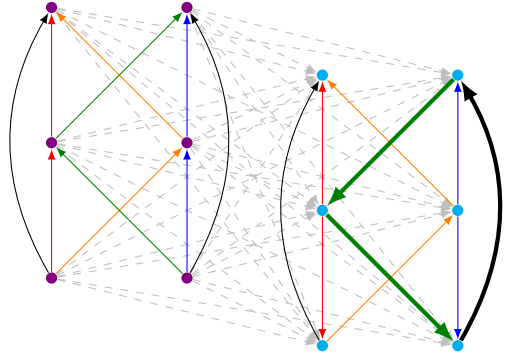

    \centering
    \begin{subfigure}[t]{0.45\textwidth}

    \centering
        \tikz[scale=2, every path/.style={thick}, every node/.style={circle,draw=black,fill=white,minimum width=7.5mm,font=\small}]{
            \foreach \ii in {1,2}{
                \node (a\ii) at (\ii,0) {$a_{\ii}$};
                \node (b\ii) at (\ii,-1) {$b_{\ii}$};
                \node (c\ii) at (\ii,-2) {$c_{\ii}$};
            }        

            \draw[-latex,red] (b1)--(a1); 
            \draw[-latex,orange] (b2)--(a1); 
            \draw[-latex,darkgreen,line width=\lwidth] (a2)--(b1); 
            \draw[-latex,blue] (b2)--(a2); 
            
            \draw[-latex,red] (b1)--(c1); 
            \draw[-latex,orange] (c1)--(b2); 
            \draw[-latex,darkgreen,line width=\lwidth] (b1)--(c2); 
            \draw[-latex,blue] (b2)--(c2); 

            \draw[-latex,black] (c1) to[bend left] (a1);
            \draw[-latex,black,line width=\lwidth] (c2) to[bend right] (a2);
            
        }
        \caption{$\gtri^{x,y}$ for $m=2$, $x=\begin{bmatrix}{\textcolor{red}0}&{\textcolor{orange}0}\\{\textcolor{darkgreen}1}&{\textcolor{blue}0}\end{bmatrix}$, 
        $y=\begin{bmatrix}{\textcolor{red}1}&{\textcolor{orange}0}\\{\textcolor{darkgreen}1}&{\textcolor{blue}1}\end{bmatrix}$. In this case, $\DISJ_4(x,y) = 0$ and a directed triangle is indicated with bold lines.}
    \end{subfigure}
    \hfill
    \begin{subfigure}[t]{0.45\textwidth}

    \centering
        \tikz[scale=1.8,every node/.style={circle,inner sep=1.5pt}]{
            \foreach \tt/\xx/\yy/\cc in {1/0/0/violet,2/2/-0.5/cyan}{
                \begin{scope}[shift={(\xx,\yy)}]
                \foreach \ii in {1,2}{
                    \node[fill=\cc] (a\ii-\tt) at (\ii,0) {};
                    \node[fill=\cc] (b\ii-\tt) at (\ii,-1) {};
                    \node[fill=\cc] (c\ii-\tt) at (\ii,-2) {};
                }
                \end{scope}   
            }

            \foreach \aa/\bb in {1/2}
                \foreach \xch in {a,b,c}
                    \foreach \iii in {1,2}
                        \foreach \ych in {a,b,c}
                            \foreach \jjj in {1,2}
                                \draw[lightgray,dashed,-latex] (\xch\iii-\aa) -- (\ych\jjj-\bb);

            \foreach \tt in {2}{
                \draw[-latex,red] (b1-\tt)--(a1-\tt); 
                \draw[-latex,orange] (b2-\tt)--(a1-\tt); 
                \draw[-latex,darkgreen,ultra thick] (a2-\tt)--(b1-\tt); 
                \draw[-latex,blue] (b2-\tt)--(a2-\tt); 
                
                \draw[-latex,red] (b1-\tt)--(c1-\tt); 
                \draw[-latex,orange] (c1-\tt)--(b2-\tt); 
                \draw[-latex,darkgreen,ultra thick] (b1-\tt)--(c2-\tt); 
                \draw[-latex,blue] (b2-\tt)--(c2-\tt); 
                
                \draw[-latex,black] (c1-\tt) to[bend left] (a1-\tt);
                \draw[-latex,black,ultra thick] (c2-\tt) to[bend right] (a2-\tt);
            }
            
            \foreach \tt in {1}{
                \draw[-latex,red] (b1-\tt)--(a1-\tt);
                \draw[-latex,orange] (b2-\tt)--(a1-\tt);
                \draw[-latex,darkgreen] (b1-\tt)--(a2-\tt);
                \draw[-latex,blue] (b2-\tt)--(a2-\tt);
                
                \draw[-latex,red] (c1-\tt)--(b1-\tt);
                \draw[-latex,orange] (c1-\tt)--(b2-\tt);
                \draw[-latex,darkgreen] (c2-\tt)--(b1-\tt);
                \draw[-latex,blue] (c2-\tt)--(b2-\tt);
                
                \draw[-latex,black] (c1-\tt) to[bend left] (a1-\tt);
                \draw[-latex,black] (c2-\tt) to[bend right] (a2-\tt);
            }
        } 
        \caption{$G^{12,6}_{\vec{\Gamma}}$ built from triangle gadgets on $\vec{x}=(x_1,x_2)$ and $\vec{y}=(y_1,y_2)$, where  $\textcolor{violet}{x_1}=\textcolor{violet}{y_1}=\begin{bmatrix}{\textcolor{red}0}&{\textcolor{orange}0}\\{\textcolor{darkgreen}0}&{\textcolor{blue}0}\end{bmatrix}$, 
        $\textcolor{cyan}{x_2}=\begin{bmatrix}{\textcolor{red}0}&{\textcolor{orange}0}\\{\textcolor{darkgreen}1}&{\textcolor{blue}0}\end{bmatrix}$, 
        $\textcolor{cyan}{y_2}=\begin{bmatrix}{\textcolor{red}1}&{\textcolor{orange}0}\\{\textcolor{darkgreen}1}&{\textcolor{blue}1}\end{bmatrix}$. The construction corresponds to solving $\DISJ_8(\vec{x},\vec{y})$.}
    \end{subfigure}
    \caption{$\gtri$ and a gadget-embedding tournament with triangle gadgets}
    \label{fig:Gamma-triangle}
\end{figure}
See \cref{fig:Gamma-triangle} for an illustration.
Each disjointness gadget is designed in such a way to encode a $\DISJ_{m^2}$ instance.
We state the following simple results without proof, which illustrate how certain natural graph problems detect the disjointness-embedding structural properties.
\begin{proposition}
    In $\gplain^{x,y}$, for each $i$, $c_i$ is reachable from $a_i$ iff $x_{ij}=y_{ij}=1$ for some $j$. Consequently, any $p$-pass streaming algorithm for $\PairReach_m$ requires $\Omega(m^2/p)$ space.
\end{proposition}
\begin{proposition} \label{prop:3-cycle}
    Any directed triangle in $\gtri^{x,y}$ takes the form of $a_i\arc b_j\arc c_j\arc a_i$ for some $i,j$, which occurs only when $x_{ij}=y_{ij}=1$. Consequently, any $p$-pass streaming algorithm for $3\text{-}\CYCLE_m$ requires $\Omega(m^2/p)$ space.
\end{proposition}
The two disjointness gadgets admit independence number upper bounds more precise than the trivial bound given by the number of nodes: for the parameter $m$, we have $\alpha(\gplain^{x,y})\leq 2m$ and $\alpha(\gtri^{x,y})\leq m$. Indeed, in $\gplain^{x,y}$, any independent set contains at most two of $\set{a_i,b_i,c_i}$ for each $i\in [m]$, while in $\gtri^{x,y}$ it contains at most one for each $i\in [m]$. This allows us to obtain parameterized lower bounds for $\PairReach^\alpha_n$ and $3\text{-}\CYCLE^\alpha_n$ for the full range of $\alpha$.

\begin{theorem}[Restatement of \cref{thm:pairreach/tri-into}]\label{thm:pair-reach}
    Any $p$-pass streaming algorithm for $\PairReach^\alpha_n$ or $3\text{-}\CYCLE^\alpha_n$ requires $\Omega(\alpha n/p)$ space.
\end{theorem}
\begin{proof}
    We first present the proof for $3\text{-}\CYCLE^\alpha$ for all $\alpha$. Consider two strings $(X,Y)\in \set{0,1}^{\alpha n}$. We divide $X$ into $n/\alpha$ blocks of length-$\alpha^2$ substrings $X_1,\ldots,X_{n/\alpha}$, and likewise for $Y$. Next, we create $G^{3n,3\alpha}_{\vec{\Gamma}}$, where $\vec{\Gamma}=(\Gamma_1,\ldots,\Gamma_{n/\alpha})$ is defined by
    \[
        \Gamma_k = \gtri^{X_k,Y_k} \qquad\text{ for each } k\in [n/\alpha].
    \]
    Note that each $\Gamma_k$ is a graph of size $3\alpha$, and $\ind(G^{3n,3\alpha}_{\vec{\Gamma}}) = \max_k \ind(\Gamma_k)\leq \alpha$.
    
    From \cref{obs:ordering}, we note that any possible directed triangle must situate in some induced subgraph $G^{3n,3\alpha}_{\vec{\Gamma}}[V^{(k)}] \cong \Gamma_k=\gtri^{X_k,Y_k}$ for some $k\in [n/\alpha]$. Therefore
    \begin{align}
        3\text{-}\CYCLE^\alpha \left(G^{3n,3\alpha}_{\vec{\Gamma}}\right) 
        &=\neg \bigwedge^{n/\alpha}_{i=1} \DISJ_{\alpha^2}(X_i,Y_i) = \neg\DISJ_{\alpha n}(X,Y). \label{eq:and-disj}
    \end{align}
    The desired streaming lower bound of $\Omega(\alpha n/p)$ follows from \cref{fact:DISJ}.

    An analogous argument applies to $\PairReach^\alpha$ with each triangle gadget replaced by a plain gadget. Because of the bound for $\ind(\gplain^{x,y})$, the analogous argument is applicable for even $\alpha$, which extends to $\alpha\geq 2$. It remains to handle the remaining case $\alpha=1$. 
    
    For $\alpha=1$, we show a reduction from $\DISJ_n$. Consider $G^{3n,3}_{\vec{\Gamma}}$, where each $\Gamma_i$ is a triangle gadget with node set $\set{a_i,b_i,c_i}$ encoding the bit pair $(X_i,Y_i)$ for each $i\in [n]$. In this case, $c_i$ is reachable from $a_i$ iff $X_i=Y_i=1$ for each $i$ (this reachability property is not true for a triangle gadget of parameters $m\geq 2$). It is easy check that $\ind(G^{3n,3}_{\vec{\Gamma}})=1$, so a $\PairReach^1$ streaming algorithm on the input $G^{3n,3}_{\vec{\Gamma}}$ and $\set{(a_i,c_i)}_{i=1}^n$ would solve $\DISJ_n(X,Y)$. \cref{fact:DISJ} concludes the proof.
\end{proof}

An immediate reduction to $\PairReach$ allows us to lower bound the streaming complexity of computing a strong connectivity certificate, complementing our algorithmic results in \cref{thm:main-1conn,thm:main-kconn}.
\begin{theorem}[Restatement of \cref{thm:ConnCert-intro}]\label{thm:cert-LB}
    For every $k$, any $p$-pass algorithm for $k\text{-}\AConnCert^\alpha$ or $k\text{-}\NConnCert^\alpha$ requires $\Omega(\alpha n/p)$ space.
\end{theorem}
\begin{proof}
    It suffices to prove the case $k=1$, where a $1$-arc connectivity certificate is equivalent to a $1$-node strong connectivity certificate. Notice that for any two nodes $u$ and $v$ in $G$, a $1$-arc strong connectivity certificate of $G$ contains a $u$--$v$ path iff $v$ is reachable from $u$ in $G$. Therefore, a $1$-arc strong connectivity certificate maintains all-pair reachability, which certainly solves $\PairReach^\alpha$. The lower bound follows from~\cref{thm:pair-reach}.
\end{proof}
We remark that \cref{thm:cert-LB} holds independent of the certificate size. Note that $\PairReach$ maps an $\Theta(n^2)$-bit input digraph to an $\Theta(n)$-bit output, and a $1$-arc connectivity certificate may contain as few as $n$ arcs (when it is a Hamiltonian cycle of $G$), hence the space lower bound of $\Omega(\alpha n/p)$ is indeed nontrivial. 
It is an interesting open question whether one can prove a stronger workspace lower bound under the restriction that the output certificate be nontrivial, for example, by forbidding the original graph as the certificate output (when $G$ is dense enough, specifically, containing $\Omega(\alpha n)$ arcs).

\subsection{Parameterized Lower Bounds through $\TRIBES$}
The parameterized hardness results for $\PairReach$ and $3\text{-}\CYCLE$ crucially rely on the fact that the randomized communication complexity for $\AND_r\circ^r \DISJ_s = \DISJ_{rs}$ is $\Omega(rs)$. A related communication problem called \emph{tribes} is also known to achieve the same complexity lower bound.
\begin{theorem}[\cite{jayram2003two}] \label{thm:TRIBES}
    Define $\TRIBES_{r,s} = \AND_r\circ^r (\neg \DISJ_s)$.\footnote{Notice that the definition of disjointness in this paper is the negation of the version adopted in \cite{jayram2003two}.} Then $\Rcc(\TRIBES_{r,s}) = \Omega(r(s-1))$\footnote{When $s=1$, $\TRIBES_{r,1}$ reduces to determining whether each of the two parties holds an all-1 strings, which only requires $O(1)$ bits of communication.}, and consequently any $p$-pass algorithm for a streaming problem reduced from $\TRIBES_{r,s}$ requires $\Omega(r(s-1)/p)$ space.
\end{theorem}
In the rest of this section, we apply the gadget-embedding framework using $\TRIBES$ to the Hamiltonian path and cycle problems. For a list of $n/\alpha$ size-$\alpha$ gadgets $\vec{\Gamma}$ to be specified later, we build a gadget-embedding tournament $G^{n,\alpha}_{\vec{\Gamma}}$. Next, we add a source $s_*$ and a sink $t_*$, we create the arcs $s_*\arc v$ and $v\arc t_*$ for every $v\in V(G^{n,\alpha}_{\vec{\Gamma}})$. Denote the resultant graph $G^{\star}_{\vec{\Gamma}}$. 

Notice that the new graph $G^{\star}_{\vec{\Gamma}}$ preserves the same ``global ordering'' structure: every arc is directed from $V^{(i)}$ to $V^{(i')}$ for some $i\leq i'$, which for convenience we set $V^{(0)}=\set{s_*}$ and $V^{(n/d+1)}=\set{t_*}$. It is also obvious that $\ind(G^\star_{\vec{\Gamma}}) \leq \alpha$.

The key to our construction is that $G^{\star}_{\vec{\Gamma}}$ contains a Hamiltonian path iff each $\Gamma_i$ in $\vec{\Gamma}$ contains a Hamiltonian path. 
\begin{lemma}\label{lem:ham-tribes}
    $\HamPath_{n+2}^\alpha\left(G^{\star}_{\vec{\Gamma}}\right) = 
    \bigwedge^{n/\alpha}_{i=1} \HamPath_\alpha(\Gamma_i)$.
\end{lemma}
\begin{proof}
    Suppose each $\Gamma_i$ contains a Hamiltonian path $P_i$ starting at $s_i$ and ending at $t_i$. Using the inter-component arcs $t_i\arc s_{i+1}$ for each $i\in [n/\alpha-1]$, together with $s_*\arc s_1$ and $t_{n/\alpha}\arc t_*$, we join all $P_i$'s with $s_*$ and $t_*$ to form a Hamiltonian path for $G^{\star}_{\vec{\Gamma}}$. This proves one direction.

    For the other direction, \cref{obs:ordering} implies that any Hamiltonian path $P$ in $G^{\star}_{\vec{\Gamma}}$ must start at $s_*$ and end at $t_*$. Moreover, every node from the $i$-th component must precede all nodes from the $i'$-th component whenever $i<i'$. Thus, after removing the first and last arcs together with the inter-component arcs, $P$ is decomposed into $n/\alpha$ paths $P'_1,\ldots,P'_{n/\alpha}$, where each $P'_i$ lies entirely within $\Gamma_i$ and visits all its nodes. Therefore, each $P'_i$ is a Hamiltonian path in $\Gamma_i$.
\end{proof}

Using a reduction from disjointness to $\HamPath$ in \cite{bacrach2019hardness} in a blackbox way, we obtain a reduction from $\TRIBES$ to our $\alpha$-bounded instance for Hamiltonian Path problem.
\begin{lemma}[{\cite[Theorem 2.2]{bacrach2019hardness}}] \label{lem:disj-ham}
    There is a reduction from $\DISJ_{\Theta(m^2/\log^{2}{m})}$ to $\HamPath_m$ which outputs an $m$-node graph $G_{x,y}$ on the input $(x,y)$ of $\DISJ$, such that $\DISJ(x,y)=0$ iff $G_{x,y}$ contains a Hamiltonian path.
\end{lemma}
\begin{theorem}[Restatement of \cref{thm:HamP/HamC-intro}]\label{thm:HamP/HamC-LB}
    Any $p$-pass streaming algorithm for $\HamPath^\alpha_n$ or $\HamCycle^\alpha_n$ requires $\Omega\left(\frac{\alpha n}{p\log^2 \alpha}\right)$ space.
\end{theorem}
\begin{proof}
    We first prove the lower bound for $\HamPath^\alpha_n$ for $\alpha$ exceeding a suitable constant, which is a necessary condition for the construction in \cref{lem:disj-ham} to exist. Write $\beta = \Theta(\alpha^2/\log^2\alpha)$. For two collections of $n/\alpha$ length-$\beta$ boolean strings $\set{x_1,\ldots,x_{n/\alpha}}$ and $\set{y_1,\ldots,y_{n/\alpha}}$, set $\Gamma_i = G_{x_i,y_i}$ to be the size-$\alpha$ digraph as given in \cref{lem:disj-ham}.    
    \cref{lem:disj-ham,lem:ham-tribes} together imply that 
    \[
        \HamPath^\alpha_{n+2}\left(G^\star_{\vec{\Gamma}}\right) 
        = \bigwedge_{i=1}^{n/\alpha} (\neg \DISJ_{\beta}(\Gamma_i))
        = \TRIBES_{n/\alpha,\beta}(\vec{x},\vec{y}),
    \]
    \cref{thm:TRIBES} yields the space lower bound of $\Omega\left(\frac{\beta n}{\alpha p}\right) = \Omega\left(\frac{\alpha n}{p\log^2 \alpha}\right)$.

    The remaining range of $\alpha$ satisfies $\alpha=O(1)$, and we give a simpler construction for this case. For the inputs $x,y\in\set{0,1}$ and a given $\alpha$, the gadget $\Gamma^{x,y}_{\triangle,\alpha}$ is constructed as follows. Denoting $V_\triangle \coloneqq\set{s,u,t}$, and $V_P\coloneqq \set{p_1,\ldots,p_{2\alpha-3}}$ for $\alpha\geq 2$, the node set of $\Gamma^{x,y}_{\triangle,\alpha}$ is $V_\triangle$ for $\alpha=1$ and $V_\triangle\cup V_P$ for $\alpha\geq 2$. The input-dependent arcs are: $\set{s\arc u,u\arc t}$ if $x=0$ and $\set{u\arc s,t\arc u}$ if $x=1$, and $\set{s\arc t}$ if $x=0$ and $\set{t\arc s}$ if $y=1$. Next, construct the arcs $t\arc p_1$, $p_{2\alpha-3}\arc u$ (when $\alpha=2$) and $p_i\arc p_{i+1}$ for $i\in [2\alpha-4]$ (when $\alpha\geq 3$). See \cref{fig:Gamma-triangle-alpha} for an illustration.

    The gadget consists of $O(\alpha)$ nodes.
    Clearly, $t$ is reachable from $s$ iff $x\wedge y=0$, so $\HamPath(\Gamma^{x,y}_{\triangle,\alpha}) = \DISJ(x,y)$. Moreover, $\ind(\Gamma^{x,y}_{\triangle,\alpha}) \leq \ind(\Gamma^{x,y}_{\triangle,\alpha}[V_\triangle]) + \ind(\Gamma^{x,y}_{\triangle,\alpha}[V_P]) = 1+(\alpha-1) =\alpha$.

    For $X,Y\in\set{0,1}^{n/\Omega(\alpha)}$, set $\vec{\Gamma}=(\Gamma_1,\ldots,\Gamma_{\Omega(\alpha)})$ by $\Gamma_i = (\Gamma^{X_i,Y_i}_{\triangle,\alpha})$.
    By \cref{lem:ham-tribes}, we obtain a reduction from $\bigwedge^{n/\Omega(\alpha)}_{i=1} \DISJ(X_i,Y_i) = \DISJ_{\Omega(n/\alpha)}$ to $\HamPath^\alpha$. Applying \cref{fact:DISJ} yields the claimed $\Omega(n/p)$ space lower bound.
\end{proof}

\begin{figure}[t]
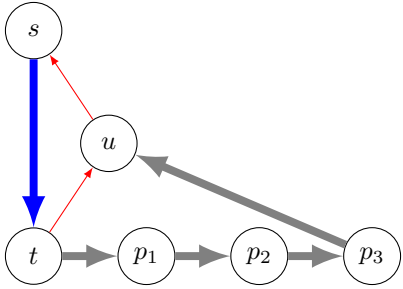

\centering
\tikz[every node/.style={circle,draw=black,fill=white,minimum width=7.5mm,font=\small}]{
    \node (s) at (0,1) {$s$};
    \node (u) at (1,-0.5) {$u$};
    \node (t) at (0,-2) {$t$};
    \foreach \ii in {1,...,3}
        \node (p\ii) at (\ii*1.5,-2) {$p_{\ii}$};
            
    \begin{scope}[-latex]
        \draw[blue,line width=\lwidth] (s)--(t);
        \draw[red] (t)--(u);
        \draw[red] (u)--(s);
        \draw[gray,line width=\lwidth] (t)--(p1);
        \draw[gray,line width=\lwidth] (p1)--(p2);   
        \draw[gray,line width=\lwidth] (p2)--(p3);   
        \draw[gray,line width=\lwidth] (p3)--(u);   
    \end{scope}
}
\caption{$\Gamma^{x,y}_{\triangle,\alpha}$ for $\alpha=3$, $x=0$ and $y=1$. A Hamiltonian path is indicated with bold lines.}
\label{fig:Gamma-triangle-alpha}
\end{figure}

\subsection{Parameterized Lower Bounds through Reachability}
In this section, we illustrate the applicability of the gadget-embedding framework using reachability.

Suppose $\vec{\Gamma}$ consists of $n/\alpha$ reachability instances $\Gamma_i$ with source $s_i$ and sink $t_i$.
We instantiate a $\REACH^\alpha_n$ instance by embedding $\vec{\Gamma}$ with the following modifications:
\begin{itemize}
    \item The embedding for $G^{n,\alpha}_{\vec{\Gamma}}$ ensures that $\mathcal{E}_i(s_i) = (i,1)$ and $\mathcal{E}_i(t_i) = (i,\alpha)$ for each $i\in [n/\alpha]$;
    \item Reverse the direction of the arc $(i,1)\arc(i+1,\alpha)$ for each $i\in [n/\alpha-1]$, which we call the reversed arcs $(i+1,\alpha)\arc (i,1)$ \emph{back-edges};
    \item Denote the resultant digraph by $G^{\#}_{\vec{\Gamma}}$, and set $s_*\coloneqq(n/\alpha,1)$ and $t_*\coloneqq (1,\alpha)$ to be the source and sink of the $\REACH^\alpha_n$ instance respectively.
\end{itemize}

The following lemma shows that the reachability of the constructed gadget-embedding tournament is determined by the conjunction of reachability instances of the individual gadgets independent of the exact choice of gadgets.
\begin{lemma} \label{lem:reach-tour}
    $\REACH\left(G^{\#}_{\vec{\Gamma}}; s_*,t_*\right) = 1$ iff $\REACH\left(G^{n,\alpha}_{\vec{\Gamma}}[V^{(i)}]; (i,1), (i,\alpha)\right) = 1$ for each $i\in [n/\alpha]$.
\end{lemma}
\begin{proof}
    By the construction of the back-edges, the backwards direction is obviously true. For the forward direction, notice that $s_*\in V^{(n/\alpha)}$ and $t_*\in V^{(1)}$, the only order-defying arcs of $G^{\#}_{\vec{\Gamma}}$ are the $\left(\frac{n}{\alpha}-1\right)$ back-edges. Hence, any $s_*$--$t_*$ path must include each back-edge exactly once in the order of $\left(\frac{n}{\alpha},\alpha\right) , \left(\frac{n}{\alpha}-1,1\right), 
    \left(\frac{n}{\alpha}-1,\alpha\right) , \left(\frac{n}{\alpha}-2,1\right), \ldots,
    (2,\alpha), (1,1)$. For such a path to exist, the end-point of a back-edge must reach the starting point of the next back-edge in the above ordering, $\left(\frac{n}{\alpha},\alpha\right)$ is reachable from $s_*$, and $t_*$ is reachable from $(1,1)$. Since none of the back-edges can be used more than once, each of the reachability conditions is enforced on a separate induced subgraph $G^{n,\alpha}_{\vec{\Gamma}}[V^{(i)}]$, which completes the proof.
\end{proof}

\cref{lem:reach-tour} suggests that a direct-sum-type theorem may apply to proving parameterized lower bounds for reachability. To set up for the proof, we first present the necessary background from information theory.
We refer readers to \cite{CT05,Wei15} for the standard information-theoretic notations and for a more comprehensive introduction to information theory and its connection to streaming algorithms. Recall that the notation $I(X,Y;Z)$ denotes the mutual information between $X$ and $Y$ conditioned on $Z$. For a communication protocol $\Pi$, we denote $\Pi(X,Y)$ the transcript of $\Pi$ on input $(X,Y)$.
For a distribution $\mu$ on the inputs and $\delta\in (0,1)$, the \emph{information complexity} of a function $F$ with respect to $\mu$ is defined by
\[
    \IC_{\mu,\delta}(F) = \min_\Pi I(X,Y;\Pi(X,Y)),
\]
where $(X,Y)\sim \mu$ and $\Pi$ ranges over all protocols with error $\delta$ with respect to $\mu$.

The following lemmas are sufficient for our purpose of establishing a parameterized lower bound for reachability.
\begin{lemma}[{\cite[Proposition 4.3]{BJKS02}}] \label{lem:Rcc-IC}
    For any $\delta\in (0,1)$, any communication problem $F$ and any distribution $\eta$ on inputs, $\Rcc_\delta(F) \geq \IC_{\eta,\delta}(F)$.
\end{lemma}
\begin{lemma}[{\cite[Theorem 2]{GO16}}]\label{lem:reach-comm-reduction}
    A $p$-pass $s$-space streaming algorithm for $\REACH_m$ yields an $O(p^2 s)$-cost two-party communication protocol for a problem $\Psi_{m,p}$.
\end{lemma}
\begin{lemma}[{\cite[Lemmas 3 and 10]{GO16}}] \label{lem:I-LB}
    Let $\Pi$ be a randomized protocol that computes $\Psi_{m,p}$ with error $\Omega(m^{-1})$. For sufficiently large $m$, there exists a product distribution $\mu$ over the domain of $\Psi_{m,p}$ such that for $(X,Y)\sim \mu$, it holds that
    \[
        I(X,Y;\Pi(X,Y)) =\Omega\left(m^{1+1/(2p)}/p^{O(1)}\right).
    \]
\end{lemma}
\begin{lemma}[Follows from {\cite[Lemmas 5.1 and 5.6]{BJKS02}}] \label{lem:I-direct-sum}
    Suppose $f:(\calX\times \calY)^t\to\set{0,1}$ takes the form
    \[
        f(\vec{x},\vec{y}) = g(h(x_1,y_1),\ldots,h(x_t,y_t)),
    \]
    for some $g:\set{0,1}^t\to \set{0,1}$ and $h:\calX\times \calY\to \set{0,1}$, which we also write $f = g\circ h^t$. Let $\mu$ be a product distribution on $\calX\times \calY$, and $(\vec{X},\vec{Y})\sim \mu^t$. Then 
    \[
        \IC_{\mu^t,\delta}(f) \geq t\cdot \IC_{\mu,\delta}(h).
    \]
\end{lemma}

We are now in a position to prove a parameterized space lower bound for reachability through the general lower bound from \cite{GO16}.
\begin{theorem}[Restatement of \cref{thm:reach-intro} (b)]\label{thm:s-t-LB}
    For $\alpha=\omega(1)$ and $p=O\left(\frac{\log \alpha}{\log\log \alpha}\right)$, any $p$-pass streaming algorithm for $\REACH_n^\alpha$ requires $\Omega(\alpha^{1/(2p)}n/p^{O(1)})$ space.
\end{theorem}
\begin{proof}
    From \cref{lem:reach-tour,lem:reach-comm-reduction}, any $p$-pass streaming algorithm for $\REACH_n^\alpha$ solves a conjunction of $n/\alpha$ instances of $\REACH_\alpha$, thus $\REACH_n^\alpha$ is reduced from the communication problem $\Phi \coloneqq \AND_{n/\alpha}\circ (\Psi_{\alpha,p})^{n/\alpha}$. 

    Let $\mu$ be the product distribution given by \cref{lem:I-LB}. By \cref{lem:I-direct-sum,lem:I-LB},
    \[
        \IC_{\mu^{n/\alpha},\Omega(\alpha^{-1})}(\Phi) \geq \frac{n}{\alpha}\cdot \IC_{\mu,\Omega(\alpha^{-1})}(\Psi_{\alpha,p}) = \Omega\left(\alpha^{1/(2p)}n /p^{O(1)}\right).
    \]
    This yields a lower bound for $\Rcc_{\Omega(\alpha^{-1})}(\Phi)$ (\cref{lem:Rcc-IC}). By standard error amplification argument and the fact that, we obtain $\Rcc_{1/3}(\Phi)=\tilde{\Omega}(\alpha^{1/(2p)}n/p^{O(1)})$, and the desired space lower bound follows by considering $\REACH(G^{\#}_{\vec{\Gamma}}; s_*,t_*)$ with $\vec{\Gamma}$ being generated by $\Psi_{\alpha,p}$.
\end{proof}

For the constant $\alpha$ regime, we present a reduction of reachability from the $\TRIBES$ function that supplies a lower bound with an asymptotic form matching the general case.
\begin{theorem}[Restatement of \cref{thm:reach-intro} (a)]\label{thm:s-t-LB-const-alpha}
    For $\alpha=O(1)$, any $p$-pass streaming algorithm for $\REACH_n^\alpha$ requires $\Omega(n/p)$ space.
\end{theorem}
\begin{proof}
    The reduction is identical to that used in the constant-$\alpha$ regime in the proof of \cref{thm:HamP/HamC-LB}. As noted in the proof, for $x,y\in\set{0,1}$, $\REACH(\Gamma^{x,y}_{\triangle,\alpha};s,t) = \DISJ(x,y)$. The remaining steps follow verbatim from the previous proof.
\end{proof}
We remark that, in the unconstrained setting, \cite{CKP+21} established a nearly optimal space lower bound of $\Omega(n^{2-o(1)})$ for $o(\sqrt{\log n})$-pass algorithms, which is substantially stronger than the $\Omega(n^{1+1/(2p)}/p^{O(1)})$ bound implied by \cref{lem:reach-comm-reduction,lem:I-LB} proven by Guruswami and Onak \cite{GO16}. 
However, the lower bound result of \cite{CKP+21} is based on an indistinguishability result for low-pass sub-quadratic-space algorithms over certain hard distributions. 
There appears to be no black-box direct-sum-type argument that extends indistinguishability to the setting where the space usage is scaled by the number of embedded reachability instances.
For this reason, we resort to the weaker bound in \cite{GO16}, where the information complexity lower bound permits a direct-sum approach.
It remains an interesting direction for future work to determine whether the stronger lower bound results of \cite{CKP+21} can be applied within our direct-sum paradigm for parameterized lower bounds.

%% file: apps.tex
\section{Applications of Strong Connectivity Certificates}\label{sec:app}

In this section, we present some applications of strong connectivity certificates, complemented with lower bounds derived from the results in \cref{sec:hardness}. Remind that in~\cref{sec:node-vs-arc} we show that a $k$-node strong connectivity certificate is also a $k$-arc strong connectivity certificate.

\subsection{Arc-Disjoint Spanning Out-Branchings}\label{sec:branch}
An application of the $k$-arc strong connectivity certificate algorithm is the computation of $k$ arc-disjoint spanning out-branchings. See \cref{def:spanning-branching} for the definition of a spanning out-branching. A valid output of this problem consists of $k$ spanning out-branchings originating from a mutual designated root, and the arc sets of the out-branchings are pairwise disjoint. The existence of $k$ arc-disjoint spanning out-branchings is characterized by Edmond's branching theorem.

Let $G$ be a digraph and $H$ be a $k$-arc strong connectivity certificate of $G$. For any given root $r$, if $G$ admits $k$ arc-disjoint spanning out-branchings rooted at $r$, then $H$ does as well. Therefore, a $k$-arc strong connectivity certificate suffices to determine $k$ arc-disjoint spanning out-branchings.

\begin{theorem}\label{thm:algo-branch}
    Let $G$ be an $n$-node digraph with independence number $\alpha$. For $p\in \dbN$ and $k \ge 2$, $k$ arc-disjoint spanning out-branchings of $G$ can be computed with probability $1 - 1/n^{\Omega(1)}$ by a $p$-pass randomized algorithm using
\begin{enumerate}[label=\textup{(\alph*)}]
    \item $O(k^{1-1/p}\alpha n^{1+1/p}\log n)$ space in the insertion-only model; and
    \item $O(k^{1- O(1/\sqrt{p})}\alpha n^{1+ O(1/\sqrt{p}) }\log n)$ space in the turnstile model.
\end{enumerate}
\end{theorem}
\begin{proof}
    Let $H$ be a $k$-strong connectivity certificate of $G$. Let $(S, V \setminus S)$ be any cut with $r \in S$. If $G$ admits $k$ arc-disjoint spanning out-branchings rooted at $r$, then $|\delta^+_G(S)| \ge k$ since each spanning out-branching must cross the cut at least once. As $H$ is a $k$-arc strong connectivity certificate (\cref{thm:node-vs-arc}), $|\delta^+_H(S)| \ge k$. Therefore by \cref{thm:edmond}, $H$ admits $k$ arc-disjoint spanning out-branchings rooted at $r$. We apply the algorithm in \cref{thm:main-kconn} to obtain a $k$-arc strong connectivity certificate of $G$, and then apply an offline algorithm on $H$ to obtain $k$ arc-disjoint out-branchings. The pass and space usage follows from \cref{thm:main-kconn}.
\end{proof}

For $\InBranch$ and $\OutBranch$, we can improve with the deterministic 1-strong connectivity certificate algorithm (\cref{thm:main-1conn}).
\begin{theorem}\label{thm:algo-branch1}
    For $p\in \dbN$, $\InBranch_n^\alpha$ and $\OutBranch_n^\alpha$ can be computed by a $p$-pass deterministic algorithm using
\begin{enumerate}[label=\textup{(\alph*)}]
    \item $O(\alpha n^{1+1/p})$ space in the insertion-only model; and
    \item $O(\alpha n^{1+ O(1/\sqrt{p}) })$ space in the turnstile model.
\end{enumerate}
\end{theorem}

Another approach to computing $k$ arc-disjoint spanning out-branchings based on \cref{thm:edmond} is matroid intersection~\cite{Edmonds73}. The goal of a matroid intersection problem is to find a maximal common independent set of two matroids. The problem of computing $k$ arc-disjoint spanning out-branchings of a graph $G=(V,E)$ can be expressed as a matroid intersection problem of two matroids $M_1$ and $M_2$, where $M_1$ is the union of $k$ graphic matroids and $M_2$ is the union of $k$ partition matroids, each take the form $\{e\subseteq E: \delta^{+}_{e}(v)\leq 1 \text{ for all }v\in V\}$. However, existing streaming algorithms for matroid intersection~\cite{CrouchS14,GargJS23,Ter25} are unable to solve this problem even for the case $k=1$, the reason being that these algorithms only return approximate solutions in the form of out-trees rather than a spanning out-branching.

For the lower bound, the parameterized bound of the spanning branching problems follows from a direct reduction from reachability.
\begin{lemma}\label{lem:reach-spanningbranching}
    $\REACH_{n}^{\alpha}$ reduces to $\OutBranch_{n}^{\alpha}$ and $\InBranch_{n}^{\alpha}$.
\end{lemma}
\begin{proof}
    It suffices to show the reduction to $\OutBranch_{n}^{\alpha}$, and the in-branching reduction are the same except the arc directions are reversed. For an instance $(G;s,t)$ to $\REACH^\alpha_n$, construct a new digraph $G'$ by adding the arcs $\set{(t,v):v\in V(G)\setminus \set{t}}$. Obviously, $\ind(G')\leq \ind(G)\leq \alpha$. It is clear that $G'$ has a spanning out-branching from $s$ whenever $t$ is reachable from $s$; conversely, if $s$ cannot reach $t$, any out-branching originating from $s$ in $G'$ does not contain $t$. This shows that $\OutBranch^\alpha_n(G';s) = \REACH^\alpha_n(G;s,t)$.
\end{proof}
\begin{corollary}\label{cor:LB-branching}
    Any $p$-pass algorithm for $\OutBranch^\alpha_n$ or $\InBranch^\alpha_n$ requires $\Omega(\alpha^{1/(2p)} n/p^{O(1)})$ space if $\alpha=O(1)$ or $\alpha=\omega(1)$, $p=o(\sqrt{\log n})$.
\end{corollary}

Our lower bound results extend to weighted digraph settings as well. A specific problem is to find a $\rho$-approximate minimum-weight spanning out-branching rooted at a designated node, i.e., a spanning out-branching whose total weight is within a multiplicative factor $\rho$ of the optimum. We show that this problem and related variants are as hard as the unweighted reachability problem.
\begin{lemma}\label{thm:LBst-w-outbr}
    For $M>1$ and $\delta\in (0,1)$, $\REACH_{n}^{\alpha}$ reduces to $(1-\delta)M\apx~\MaxW\OutBranch_{\Theta(n/\delta)}^{\alpha,[M]}$. A similar reduction to the $\MaxW\InBranch$, $\MinW\OutBranch$ and $\MinW\InBranch$ variants (with the same parameters) holds.
\end{lemma}
\begin{proof}
    We present the proof for maximum-weight spanning out-branching. The other variants follow a similar proof. Towards a reduction, for a $\REACH^\alpha_n$ instance with source $s$ and sink $t$, we build a weighted digraph $G'=(V',A')$ in $\mathcal{G}_{n,\alpha,M}$ as follows. 
    The node set is $V'=V\cup\set{a}\cup\set{p_i:i\in [cn]}$, which has size $(1+c)n+1$. Each arc in $A'$ is assigned a weight of 1 or $M$. The weight-1 arcs are all the original arcs $A$ and the new arcs
    \[
        \set{(p_i,p_{i+1}):i\in [cn-1]} 
        \cup \set{(p_i,p_j):i-j\geq 2}
        \cup \set{(a,v):v\in V'\setminus \set{s,a,p_{cn}}}
        \cup \set{(s,p_1),(p_{cn},a)}.
    \]
    The arcs $\set{(t,p_i):i\in [cn]}$ are each assigned a weight of $M$.     
    In other words, $G'$ consists of the original graph $G$ and a new clique with a planted path $p_1\arc p_2\arc\ldots\arc p_{cn}\arc a$, so that $s$ is connected to $p_1$, $a$ is connected to all other nodes in $G'$ except $s$ and $p_{cn}$, and $t$ is connected to all path nodes except $a$ with heavier arcs. See \cref{fig:LBst-wob} for an illustration. 
    Notice that $\ind(G') \leq \ind(G'[V]) +\ind(G'[V'\setminus V])\leq \alpha+1$.

    We set $s$ to be the designated root of the out-branching, and we show that the maximum-weight out-branching from $s$ in $G'$ depends on the $s$--$t$ reachability in the original graph $G$. Notice that in $G'$, $a$ is always reachable from $s$ via the path, so there is always an out-branching of weight $(1+c)n$. On the other hand, an out-branching rooted at $s$ can contain weight-$M$ arcs only if $t$ is reachable from $s$. In such a case, the maximum-weight out-branching contains all $cn$ weight-$M$ arcs and has a total weight of $(Mc+1)n$. The ratio of the total weights in these two cases is $\frac{Mc+1}{c+1} = M - \frac{M-1}{c+1}$. Setting $c=\Theta(1/\delta)$, the above ratio exceeds $(1-\delta)M$ and hence a $(1-\delta)M$-approximation algorithm can tell the two cases of reachability apart.
\end{proof}
\begin{corollary}
    For $\delta\in (0,1)$, Any $p$-pass algorithm for $(1-\delta)M\apx~\MaxW\OutBranch_{n}^{\alpha,[M]}$ requires $\Omega(\alpha^{1/(2p)} \delta n/p^{O(1)})$ space if $2\leq \alpha=O(1)$ or $\alpha=\omega(1)$, $p=o(\sqrt{\log n})$. The same bounds hold for the $\MaxW\InBranch$, $\MinW\OutBranch$ and $\MinW\InBranch$ variants.
\end{corollary}

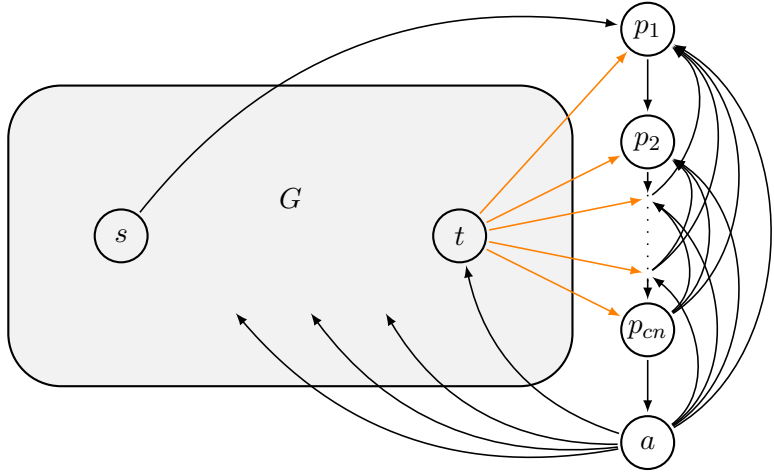
\begin{figure}[h]
    \centering
    \begin{tikzpicture}[thick,
          vertex/.style={draw, circle, thick, inner sep=1pt, minimum size=20pt},
          edge/.style={black, line width=0.6pt,>=stealth,shorten >=1pt,shorten <=1pt}
        ]
        
    \draw[rounded corners=20pt,fill=gray!10] (-1.5,0) rectangle (6,4);
    
    \node[vertex] (S) at (0,2){$s$};
    \node[vertex] (T) at (4.5,2){$t$};
    \node[vertex] (p1) at (7,4.75){$p_1$};
    \node[vertex] (p2) at (7,3.25){$p_2$};
    \node[vertex] (pn) at (7,0.75){$p_{cn}$};
    \node[vertex] (A) at (7,-0.75){$a$};
    \node[] at (2.25,2.5) {$G$};
    \node[inner sep=0pt] (p4) at (7,1.5){};
    \node[inner sep=0pt] (p3) at (7,2.5){};
    
    \draw[edge,-latex,orange] (T) to (p1);
    \draw[edge,-latex,orange] (T) to (p2);
    \draw[edge,-latex,orange] (T) to (pn);
    \draw[edge,-latex,orange] (T) to (p3);
    \draw[edge,-latex,orange] (T) to (p4);
    \draw[edge,-latex] (A) to[bend left] (T);
    \draw[edge,-latex] (S) to[bend left] (p1);
    \draw[edge,-latex] (p1) to (p2);
    \draw[edge,-latex] (p2) to (p3);
    \draw[edge,-latex] (p4) to (pn);
    \draw[edge,-latex] (pn) to (A);
    \draw[edge,-latex] (A) to[bend left] (3.5,1);
    \draw[edge,-latex] (A) to[bend left] (2.5,1);
    \draw[edge,-latex] (A) to[bend left] (1.5,1);
    \draw[edge,loosely dotted,shorten >=5pt,shorten <=5pt] (p2) to (pn);
    \draw[edge,-latex,out=37,in=-37] (A) to (p4);
    \draw[edge,-latex,out=37,in=-37] (pn) to (p3);
    \draw[edge,-latex,out=37,in=-37] (p4) to (p2);
    \draw[edge,-latex,out=37,in=-37] (p3) to (p1);
    \draw[edge,-latex,out=35,in=-35] (A) to (p3);
    \draw[edge,-latex,out=35,in=-35] (pn) to (p2);
    \draw[edge,-latex,out=35,in=-35] (p4) to (p1);
    \draw[edge,-latex,out=33,in=-33] (A) to (p2);
    \draw[edge,-latex,out=33,in=-33] (pn) to (p1);
    \draw[edge,-latex,out=30,in=-30] (A) to (p1);
    
\end{tikzpicture}
\caption{Construction of $G'$ in the proof of~\cref{thm:LBst-w-outbr}. The orange arcs are each assigned a weight of $M$, and each of the other arcs has weight 1.}
\label{fig:LBst-wob}
\end{figure}

\subsection{Independent Spanning Out-Branchings}
Let $T_1, T_2$ be two spanning out-branchings of a digraph $G=(V,A)$, where $T_1, T_2$ have a mutual root $r$. $T_1$ and $T_2$ are \textit{$r$-independent} if for every $v\in V\setminus \{r\}$, the path from $r$ to $v$ in $T_1$ and the path from $r$ to $v$ in $T_2$ are internally node-disjoint, i.e. they share only the endpoints $r$ and $v$. It is known that if a digraph $G$ is $k$-node-strong for $k\in \{1,2\}$, then $G$ contains $k$ $r$-independent spanning out-branchings for any $r\in V$~\cite{whitty1987}. Notably, this result does not extend to $k \geq 3$ \cite{huck1995}.

Let $G$ be a digraph and $H$ be a $2$-node strong connectivity certificate of $G$. For any given root $r$, if $G$ contains a pair of $r$-independent spanning out-branchings, then $H$ does as well. Therefore, a $2$-node strong connectivity certificate suffices to determine two $r$-independent spanning out-branchings if exists.
\begin{theorem}\label{thm:algo-nodebranch}
    For $p\in \dbN$, $2$-$\indBranch^\alpha_n$ can be solved with probability $1 - 1/n^{\Omega(1)}$ by a $p$-pass randomized algorithm using
\begin{enumerate}[label=\textup{(\alph*)}]
    \item $O(\alpha n^{1+1/p}\log n)$ space in the insertion-only model; and
    \item $O(\alpha n^{1+ O(1/\sqrt{p}) }\log n)$ space in the turnstile model.
\end{enumerate}
\end{theorem}

Since computing a pair of $r$-independent spanning out-branchings includes computing one spanning out-branching, it inherits the lower bound of~\cref{cor:LB-branching}.
\begin{corollary}\label{cor:LB-nodebranching}
    Any $p$-pass algorithm that computes a pair of $r$-independent spanning out-branchings for any given root $r$ requires $\Omega(\alpha^{1/(2p)} n/p^{O(1)})$ space if $\alpha=O(1)$ or $\alpha=\omega(1)$, $p=o(\sqrt{\log n})$.
\end{corollary}

\subsection{Minimum Spanning Strong Subgraph}
A \emph{minimum spanning strong subgraph} of a digraph $G$ is a spanning subgraph $H$ with the minimum number of arcs such that $H$ is strongly connected. 
Evidently, $G$ must be strongly connected to admit a minimum spanning strong subgraph. 
The problem of computing a minimum spanning strong subgraph ($\MSSS$) is known to be NP-hard~\cite{garey2002computers}, and approximation algorithms for $\MSSS$ have been studied in the RAM model \cite{FJ81,vetta2001approximating}.
The next theorem shows that our $1$-arc strong connectivity certificate algorithm can be used to obtain a 2-approximate MSSS.
\begin{theorem}\label{thm:algo-2apxMSSS}
For $p\in \dbN$, a $2$-approximate solution to $\MSSS^{\alpha}_{n}$ can be computed by a $p$-pass deterministic algorithm using:
\begin{enumerate}[label=\textup{(\alph*)}]
    \item $O(\alpha n^{1+1/p})$ space in the insertion-only model; and
    \item $O(\alpha n^{1+ O(1/\sqrt{p}) })$ space in the turnstile model.
\end{enumerate}

\end{theorem}
\begin{proof}
    Suppose $G$ is the input to the $\MSSS^{\alpha}_{n}$ instance with the designated root $r$. Let $H$ be a $1$-strong connectivity certificate of $G$. Notice that when $G$ is not strongly connected, then by definition $H$ is not strongly connected, thus the algorithm correctly concludes that no MSSS exists. 
    
    If $G$ is strongly connected, then $H$ is guaranteed to contain a $2$-approximate solution to $\MSSS$. Indeed, by definition of a $1$-arc strong connectivity subgraph, $H$ admits both a spanning in-branching $T^{+}$ and a spanning out-branching $T^{-}$ rooted at any node. Clearly, $T^{+}\cup T^{-}$ is a spanning strong subgraph with at most $2n-2$ arcs. Since any spanning strong subgraph must contain at least $n$ arcs, $T^{+}\cup T^{-}$ is a $2$-approximate minimum spanning strong subgraph. By extracting $T^+$ and $T^-$ offline from a $1$-arc strong connectivity certificate, the pass and space usage follows from \cref{thm:main-1conn}.
\end{proof}

As for the lower bound for $\MSSS$, we give a reduction to the Hamiltonian cycle problem.
\begin{lemma}\label{lem:red-ham-c}
    $\MSSS^{\alpha}_{n}$ reduces to $\HamCycle^\alpha_n$.
\end{lemma}
\begin{proof}
    Notice that in any strongly connected subgraph $H$ of a digraph $G$, each node has in-degree and out-degree at least one. Equality is attained only when $H$ is a Hamiltonian cycle of $G$. Hence, on a $\HamCycle^\alpha_n$ instance $G$, a correct $\MSSS^\alpha_n$ solver on $G$ would return a Hamiltonian cycle of $G$ if one exists, and otherwise output a non-Hamiltonian cycle graph or declare that no MSSS otherwise. This completes the proof.
\end{proof}
\begin{corollary}\label{cor:LB-msss}
    Any $p$-pass algorithm for $\MSSS^{\alpha}_{n}$ requires $\Omega(\alpha n/(p\log^2{\alpha}))$ space.
\end{corollary}

\subsection{Strong Bridges}
A \emph{strong bridge} of a digraph $G$ is an arc $e$ whose removal increases the number of strongly connected components of $G$.  Italiano, Laura, and Santaroni~\cite{italiano2012finding} devised a linear-time algorithm for computing \emph{all} strong bridges of a digraph in the RAM model.
By employing a characterization of strong bridges given in their work, we show that all strong bridges can be retrieved from a $2$-arc strong connected certificate.

\begin{lemma}[{\cite[Lemma 2.2]{italiano2012finding}}] \label{lem:strong-bridge}
    For a strongly connected digraph $G=(V,A)$, an arc $e\in A$ is a strong bridge iff there exist two nodes $x,y\in V$ such that every path from $x$ to $y$ contains $e$. 
\end{lemma}

\begin{lemma}\label{lem:2-arc-bridge}
    Let $G$ be a digraph and $H$ be a $2$-arc strong connectivity certificate of $G$. An arc $e$ is a strong bridge of $G$ iff $e$ is a strong bridge of $H$.
\end{lemma}
\begin{proof}
    Suppose $S_1,S_2,\ldots, S_t$ is a strongly connected component decomposition for $G$, which also forms a strongly connected component decomposition for $H$. Any potential bridge of $G$ or $H$ must be an arc within a strongly connected component.

    First, we prove the backward direction. Let $e$ be a strong bridge of $H[S_i]$ for some $i\in [t]$. By \cref{lem:strong-bridge}, there exist $x,y\in S_i$ such that every path from $x$ to $y$ in $H[S_i]$, thus in $H$, contains $e$. In particular, $\lambda_{xy}(H)=1$. Since $H$ is a subgraph of $G$, $\lambda_{xy}(G)\geq 1$, and as $H$ is a 2-arc strongly connectivity certificate, this forces $\lambda_{xy}(G) = 1$. Moreover, this $x$--$y$ path of $G$ is the same path in $H$, so removing $e$ separates $x$ and $y$ into different strongly connected components. Therefore, $e$ is a strong bridge of $G$.

    For the forward direction, suppose $e$ is a strong bridge of $G[S_i]$ for some $i\in [t]$. By \cref{lem:strong-bridge}, there exist two nodes $x,y\in S_i$ such that every path from $x$ to $y$ in $G[S_i]$, thus in $G$, contains $e$. In particular, $\lambda_{xy}(G) = 1$. As $H$ is a 2-arc strongly connectivity certificate, it follows that $\lambda_{xy}(H)=1$. Therefore, $e$ is contained in $H$ and hence a strong bridge of $H$.
\end{proof}
\cref{lem:2-arc-bridge} immediately suggests an algorithm for computing all strong bridges of a digraph.
\begin{theorem}\label{thm:algo-bridge}
    Let $G$ be an $n$-node digraph with independence number $\alpha$. For $p\in \dbN$, all strong bridges of $G$ can be found with probability $1 - 1/n^{\Omega(1)}$ by a $p$-pass randomized algorithm using:
\begin{enumerate}[label=\textup{(\alph*)}]
    \item $O( \alpha n^{1+1/p} \log n)$ space in the insertion-only model; and
    \item $O( \alpha n^{1+O(1/\sqrt{p})} \log n)$ space in the turnstile model.
\end{enumerate}

\end{theorem}
\begin{proof}
    By \cref{lem:2-arc-bridge}, it suffices to retrieve a $2$-arc strong connectivity certificate $H$ of $G$ and apply an offline algorithm to determine all strong bridges of $H$. The pass and space usage follows from \cref{thm:main-kconn}. 
\end{proof}

We remark that one cannot replace the $2$-arc strong connectivity certificate with a $1$-arc certificate in \cref{lem:2-arc-bridge}, so the computation of strong bridges truly requires a $2$-arc strong connectivity certificate. A counterexample is given in \cref{fig:strong-bridge}: every arc of the $1$-arc strong connectivity certificate $H$ in \cref{fig:strong-bridge} is a strong bridge of $H$, but none of them is a strong bridge of the original graph $G$. Indeed, by replacing with a $1$-arc strong connectivity certificate $H$, it can be the case that $\lambda_{xy}(H)=1$ but $\lambda_{xy}(G)>1$ for some $x,y\in G$.
\begin{figure}
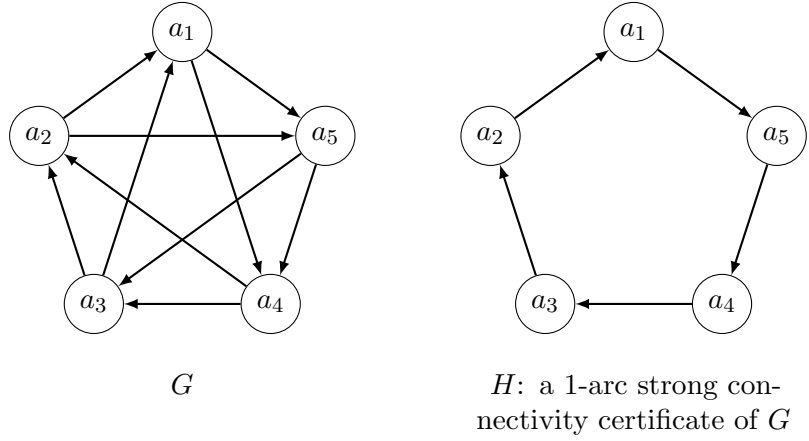

    \centering
    \tikz{
        \begin{scope}[local bounding box = Box1]  
        \foreach \ii in {1,...,5}{
            \node[minimum width=7.5mm,draw,circle] (a\ii) at (18+\ii*72:2) {$a_{\ii}$};
        }
        \begin{scope}[every path/.style={-latex,thick}]
            \draw (a1)--(a5);
            \draw (a2)--(a1);
            \draw (a3)--(a2);
            \draw (a4)--(a3);
            \draw (a5)--(a4);

            \draw (a1)--(a4);
            \draw (a2)--(a5);
            \draw (a3)--(a1);
            \draw (a4)--(a2);
            \draw (a5)--(a3);
        \end{scope}
        \node[below,yshift=-4mm] at (Box1.south) {$G$};
        \end{scope}

        \begin{scope}[local bounding box = Box2,shift={(6,0)}]  
        \foreach \ii in {1,...,5}{
            \node[minimum width=7.5mm,draw,circle] (a\ii) at (18+\ii*72:2) {$a_{\ii}$};
        }
        \begin{scope}[every path/.style={-latex,thick}]
            \draw (a1)--(a5);
            \draw (a2)--(a1);
            \draw (a3)--(a2);
            \draw (a4)--(a3);
            \draw (a5)--(a4);
        \end{scope}
        \node[below,yshift=-4mm,text width=6cm,text centered] at (Box2.south) {$H$: a 1-arc strong connectivity certificate of $G$};
        \end{scope}
    }
    \caption{Counterexample of replacing a $2$-arc strong certificate of \cref{lem:2-arc-bridge} and \cref{thm:algo-bridge} with a $1$-arc strong certificate.}
    \label{fig:strong-bridge}
\end{figure}

\begin{lemma}\label{lem:red-st-bridge}
    $\REACH_{n}^{\alpha}$ reduces to $\StrBridge_{n}^{\alpha}$.
\end{lemma}
\begin{proof}
Let $(G; s, t)$ be an instance of $\REACH^\alpha_n$, and assume without loss of generality that $s$ has no incoming edges. Construct $G'$ by adding the arc $(t, s)$ to $G$. Note that the independence number of $G'$ is at most that of $G$, and hence at most $\alpha$. If there is a path from $s$ to $t$ in $G$, then $s$ and $t$ lie in the same strongly connected component of $G'$, making $(t, s)$ a strong bridge. Otherwise, $s$ and $t$ are in different components, so $(t, s)$ cannot be a strong bridge. Thus, $\REACH^\alpha_n(G; s, t)$ can be decided by checking whether $(t, s)$ is a strong bridge in $G'$.
\end{proof}

\begin{corollary}\label{cor:LB-bridge}
    Any $p$-pass algorithm for $\StrBridge^{\alpha}_{n}$ requires $\Omega(\alpha^{1/(2p)} n/p^{O(1)})$ space if $\alpha=O(1)$ or $\alpha=\omega(1)$, $p=o(\sqrt{\log n})$.
\end{corollary}

\subsection{Topological Sort, SCC decomposition, 2-SAT, Minimum Chain Cover}
In~\cite{CLT25}, Chen, Lin, and Tsai studied parameterized streaming algorithms for $\SCC$ and $\SAT$. By contrast, the lower bounds presented there are non-parameterized and thus not directly comparable to the algorithmic results. We complement their work by establishing parameterized lower bounds derived from the hardness results in \cref{sec:hardness}.

\begin{lemma}
    $\REACH_{n}^{\alpha}$ reduces to $\SCC_{n}^{\alpha}$.
\end{lemma}
\begin{proof}  
Given an instance $(G; s, t)$ of $\REACH_{n}^{\alpha}$, construct $G' = G \cup \{(t, s)\}$ as an instance of $\SCC_{n}^{\alpha}$. This construction is valid since the independence number of $G'$ is at most that of $G$, and hence at most $\alpha$. Observe that $s$ and $t$ belong to the same strongly connected component in $G'$ if and only if there exists an $s$--$t$ path in $G$. Therefore, by examining the strongly connected components in the solution to $\SCC_{n}^{\alpha}$ on the instance $G'$, one can decide $\REACH_{n}^{\alpha}(G; s, t)$.
\end{proof}

\begin{lemma}
    $\REACH_{n}^{\alpha}$ reduces to $\SAT_n^{2\alpha}$.
\end{lemma}
\begin{proof}
We adopt the reduction from reachability to $\SAT$ presented in~\cite{jones1976new}. 
    Given an instance $(G; s, t)$ of $\REACH_{n}^{\alpha}$, construct an $\SAT$ instance $F$ as
    \[
        F = \left(\bigwedge_{(u,v)\in A}(u \lor \neg v)\right) \land (s \lor t) \land (\neg s \lor \neg t) \land (\neg s \lor t).
    \]
    Here, by abuse of notation, we identify each node $v \in V$ with a Boolean variable of $F$. 
    The last three clauses enforce $s = 1$ and $t = 0$ for $F$ to be satisfiable. 
    Thus, as shown in~\cite{jones1976new}, $G$ contains an $s$--$t$ path if and only if $F$ is not satisfiable. 

    It remains to show that the independence number of $G'$ is at most $2\alpha$. 
    Let $S^+$ and $S^-$ denote the sets of nodes corresponding to positive literals and negative literals, respectively.
    Notice that each arc $(u,v)$ in $G$ introduces two arcs $(v,u)$ and $(\neg u,\neg v)$ in $G'$. 
    Excluding the arcs from the last three clauses of $F$, we have that $G'[S^-]$ is isomorphic to $G$, while $G'[S^+]$ is identical to $G$ with all arcs reversed. 
    Therefore, $\ind(G')\leq \ind(G'[S])+\ind(G'[V'\setminus S])\leq 2\ind(G)\leq 2\alpha$.
\end{proof}

The parameterized lower bounds of these problems follow from the above reductions.
\begin{corollary}\label{cor:LB-SCC}
    Any $p$-pass algorithm for $\SCC^{\alpha}_{n}$ requires $\Omega(\alpha^{1/(2p)} n/p^{O(1)})$ space if $\alpha=O(1)$ or $\alpha=\omega(1)$, $p=o(\sqrt{\log n})$.
\end{corollary}
\begin{corollary}\label{cor:LB-2SAT}
    For $\alpha\geq 2$, any $p$-pass algorithm for $\SAT^{\alpha}_{n}$ requires $\Omega(\alpha^{1/(2p)} n/p^{O(1)})$ space if $2\leq \alpha=O(1)$ or $\alpha=\omega(1)$, $p=o(\sqrt{\log n})$.
\end{corollary}

In addition to $\SCC$ and $\SAT$, \cite{CLT25} also presented algorithms for $\TopoSort$ and $\MCC$, by using a $1$-strong connectivity certificate. Our parameterized lower bounds do not extend to the latter two problems, since both problems require the input digraph to be acyclic. Nevertheless, we extend all four parameterized upper bound results from the insertion-only model \cite{CLT25} to the turnstile model by \cref{thm:main-1conn}.
\begin{theorem}\label{thm:algo-topo-scc}
For $p\in \dbN$, $\TopoSort_{n}^{\alpha}$, $\SCC_{n}^{\alpha}$, $\MCC_{n}^{\alpha}$ or $\SAT_{n}^{\alpha}$ can be solved by a $p$-pass deterministic algorithm using:
\begin{enumerate}[label=\textup{(\alph*)}]
    \item $O( \alpha n^{1+1/p} )$ space in the insertion-only model; and
    \item $O(\alpha n^{1+O(1/\sqrt{p}} )$ space in the turnstile model.
\end{enumerate}
\end{theorem}

\subsection{Distance-$d$ Dominating Set}
For a digraph $G=(V,A)$ and parameter $d$, a distance-$d$ dominating set is a subset $D\subseteq V$ such that for every $v\in V$, there is a path of at most $d$ arcs from some node $u\in D$ to $v$. Distance-$d$ dominating set has been studied in~\cite{wang2003,joeshi2021,meir1975}. In this section, we consider computing an $O(n/d)$-node distance-$d$ dominating set for strongly connected graphs. 

Let $T$ be a spanning out-branching of $G$ rooted at an arbitrary node $r$. Since $G$ is strongly connected, $T$ is guaranteed to exist by~\cref{thm:edmond}. If an $O(n/d)$-node set $S$ is a distance-$d$ dominating set of $T$, it is also a distance-$d$ dominating set of $G$. We now describe a procedure to find such a set $S$ in $T$. We initialize $S$ as an empty set. In each iteration, each leaf $v$ of $T$ marks its farthest ancestor in $T$ that can reach $v$ through a path of at most $d$ arcs. Among the marked nodes, we retain only those that are not ancestors of any other marked nodes in $T$. These remaining marked nodes are added to $S$, and all nodes reachable from any node in $S$ are removed from $T$. Repeat this procedure until $T$ is empty. $S$ is clearly a distance-$d$ dominating set of $T$. Since adding one node to $S$ results in removing at least $d$ nodes from $T$ in the procedure, $\lvert S\rvert\leq n/d$.
The procedure for finding a distance-$d$ dominating set can be performed in memory if a spanning out-branching is given. Therefore, by~\cref{thm:algo-branch1} we have the following theorem.

\begin{theorem}\label{thm:algo-dominat}
    Let $G$ be an $n$-node strongly connected digraph with independence number $\alpha$. For $p\in \dbN$ and $d\in [n]$, an $O(n/d)$-node distance-$d$ dominating set of $G$ can be computed by a $p$-pass deterministic algorithm using
\begin{enumerate}[label=\textup{(\alph*)}]
    \item $O(\alpha n^{1+1/p})$ space in the insertion-only model; and
    \item $O(\alpha n^{1+ O(1/\sqrt{p}) })$ space in the turnstile model.
\end{enumerate}
\end{theorem}

\subsection{Transitive closure}
The transitive closure of a digraph $G$ is a graph $G'$ on the same node set, where each arc in $G'$ is defined according to the reachability of $G$. Note that $G'$ may contain substantially more arcs than $G$: for example, the transitive closure of a directed path on $n$ nodes is a size-$n$ tournament with $\Theta(n^2)$ arcs.

The task of computing the transitive closure of a digraph $G$ becomes straightforward if a $1$-strong certificate of $G$ is given, thus we have the following upper bound by \cref{thm:main-1conn}.
\begin{theorem}\label{thm:algo-tranclos}
For $p\in \dbN$, $\TranClos_{n}^{\alpha}$ can be computed by a $p$-pass deterministic algorithm using:
\begin{enumerate}[label=\textup{(\alph*)}]
    \item $O(\alpha n^{1+1/p})$ space in the insertion-only model; and
    \item $O(\alpha n^{1+O(1/\sqrt{p})})$ space in the turnstile model.
\end{enumerate}

\end{theorem}
\begin{proof}
    By definition, a $1$-node strong connectivity certificate of $G$ preserves all-pair reachability of $G$. Therefore, the transitive closure of $G$ can be computed offline once a $1$-node strong connectivity certificate is obtained. The pass and space usage follows from \cref{thm:main-1conn}.
\end{proof}
The above argument also implies that the transitive closure problem inherits the hardness of the all-pair reachability, and thus $n$-pair reachability problem.
\begin{lemma}
    $\PairReach_{n}^{\alpha}$ reduces to $\TranClos_{n}^{\alpha}$.
\end{lemma}
\begin{corollary}\label{cor:LB-tranclos}
    Any $p$-pass algorithm for $\TranClos^\alpha_n$ requires $\Omega(\alpha n/p)$ space.
\end{corollary}

\subsection{Maximum-colour Spanning Tree}
To close this section, we present an application of our digraph hardness results to an undirected graph problem.
The problem of finding a maximum-colour spanning tree of an undirected graph ($\MaxColorTree$) can be formulated as a matroid intersection problem~\cite{DBLP:journals/dmgt/BroersmaL97}.
Our hardness result implies a parameterized lower bound for this problem.
\begin{theorem}\label{thm:max-color}
    $\OutBranch^\alpha_n$ reduces to $\MaxColorTree_{n}^{\alpha}$.
\end{theorem}
\begin{proof}
    Let $G$ be an instance of $\OutBranch^\alpha_n$, where we assume $V(G)=[n]$ and the designated root is $n$. We further assume without loss of generality that the root $n$ has no incoming arcs. We obtain the undirected graph $G'$ by disregarding the arc directions. For an arc $(u,v)\in G$, assign $v$ to be the edge colour of the corresponding edge $\set{u,v}$ in $G'$.

    Notice that the available edge colours are $[n-1]$. Also, for any subgraph $T'\subseteq G'$, the corresponding subgraph $T\subseteq G$ can be reconstructed together with the edge colours in $T$. We claim that the maximum-colour spanning tree $T'$ of $G'$ uses $n-1$ colours iff $G$ admits a spanning out-branching.

    The backward direction follows immediately from the definition of the edge colouring. For the forward direction, note that $T'$ is an $(n-1)$-edge subgraph uses all $n-1$ colours, this implies $|\delta^+_T(i)| = 1$ for each $i\in [n-1]$, thus $T$ is a spanning out-branching of $G$.
\end{proof}
\begin{corollary}\label{cor:LB-maxcolor}
    Any $p$-pass algorithm for $\MaxColorTree^\alpha_n$ requires $\Omega(\alpha^{1/(2p)} n/p^{O(1)})$ space if $\alpha=O(1)$ or $\alpha=\omega(1)$, $p=o(\sqrt{\log n})$. 
\end{corollary}

%% file: distributed.tex
\section{Algorithms in the CONGEST model}\label{sec:congest}
In this section, we consider directed graph problems under distributed setting, namely the CONGEST model. Several directed graph problems have been studied in the CONGEST model, including single-source reachability \cite{ghaffari2015,jambulapati2019}, single-source shortest path \cite{forster2018faster,cao2021brief}, and SCC decomposition for planar graphs \cite{parter_planar}. 

\subsection{Model and Problem Definitions}

In the CONGEST model, the input graph $G=(V,E)$ represents a communication network, where each node $v\in V$ represents a processor with unlimited computation power, and each edge $e\in E$ represents a communication link with limited bandwidth. For directed graph problems where each arc $e\in E$ is directed, we follow the standard assumption that the communication link is bidirectional. Initially, each node $v$ knows nothing about $G$ except for its distinct ID and its adjacent arcs $\delta^{+}(v)$ and $\delta^{-}(v)$. The nodes of $G$ communicate in synchronous rounds through the communication links. In each round, each communication link can transmit an $O(\log{n})$-bit message in both directions. The performance of an algorithm is measured by the number of rounds being used in the worst case. We use $D$ to denote the diameter of the communication network. For standard protocols such as leader election, please refer to~\cite{peleg2000distributed}.

In this work, we investigate the problems of finding $k$-node strong connectivity certificates, performing SCC decomposition, and computing the topological sorting of SCCs in the CONGEST model. For any integer $k \geq 1$, the task of computing an $O(t)$-arc $k$-node strong certificate for a digraph $G=(V,E)$ requires each node $v$ to identify a subset of its incident edges in $\delta^{+}(v) \cup \delta^{-}(v)$ such that the union $Q$ of all marked edges forms a $k$-node strong connectivity certificate with $|Q| = O(t)$. The SCC decomposition problem requires each node to determine the unique identifier of the SCC to which it belongs. In the CONGEST model, SCC decomposition for planar graphs was previously studied in \cite{parter_planar}, which provides a near-optimal $\tilde{O}(D)$-round algorithm. Finally, the topological sorting of SCCs involves assigning each node $v$ a rank $t_v \in [n]$ such that for every edge $(a, b) \in E$, $t_a = t_b$ if $a$ and $b$ belong to the same SCC, and $t_a < t_b$ otherwise. Notably, if $G$ is a DAG, this problem reduces to the standard topological sorting problem.

\subsection{$k$-node Strong Connectivity Certificate}\label{sec:congest_knode}

By \cref{lem:kvc-main}, the task of computing a $k$-node strong connectivity certificate for $G$ can be reduced to computing $1$-node strong connectivity certificates for $r=O(\rho^{-2}\log n)$ subgraphs, each induced by a sampled node set $V_i (\rho)$, $i\in [r]$, where $V_i(\rho)$ is obtained by sampling each node with probability $\rho$. In this section, we show that the above result also yields efficient algorithms that compute $k$-node strong connectivity certificates in the CONGEST model. 

\begin{theorem}\label{thm:congest_knode}
    Given an algorithm $A$ in the CONGEST model that computes an $s(n)$-arc $1$-node strong certificate using $t(n)$ rounds. For any $k\in \dbN$, $\rho\in(0,1/k]$, there is an algorithm that computes an $O(r \cdot s(\rho n))$-arc $k$-node strong connectivity certificate using $O(r\rho+r\rho^2\cdot t(\rho n))$ rounds with high probability, where $r=O(\rho^{-2}\log n)$ is the number defined in~\cref{lem:kvc-main}.
\end{theorem}
\begin{proof}
    The proposed algorithm runs algorithm $A$ on all subgraphs $V_i(\rho)$, $i\in [r]$. Let each node sample itself with probability $\rho$ for $r$ times independently to decide whether it belongs to each of the sampled sets $V_i(\rho)$. Using standard Chernoff bounds, we know that each node is contained in $O(r\rho)$ sampled sets and each arc is contained in at most $O(r \rho^2)$ subgraphs with high probability. 
    
    After sampling, each node then sends messages to all of its neighbours announcing every sampled set it belongs to. Since each node is contained in $O(r\rho)$ sampled sets, all announcements will be completed within $O(r\rho)$ rounds. This allows each induced subgraph $G[V_i(\rho)]$ to execute algorithm $A$ without sending messages through arcs that are not contained in $G[V_i(\rho)]$.

    Next, since each arc is contained in at most $O(r \rho^2)$ subgraphs, $O(r \rho^2)$ rounds are enough to let each $G[V_i(\rho)]$ complete one round of its own computation. Hence, after $r\rho^2\cdot t(\rho n)$ rounds, every $G[V_i(\rho)]$ completes its computation of an $1$-node strong connectivity certificate.    
\end{proof}

Given a graph $G$ with independence number $\alpha$, each induced subgraph $G[V_i]$ also has independence number at most $\alpha$. Therefore, if we just take the naive algorithm of broadcasting all edges in each subgraph, we can have an $O(\lvert E\lvert)$-round algorithm that computes an $O(\alpha \lvert V\lvert)$-arc $1$-node strong connectivity certificate. Choosing $s(n) = O(\alpha n)$ and $t(n)=O(n^2)$ in \cref{thm:congest_knode} yields the following.

\begin{corollary}\label{cor:congest_knode}
    For any $k\in \dbN$, $\rho\in(0,1/k]$, there is an $O(r\rho (1+\rho^3 n^2))$-round algorithm that computes an $O(r\rho \alpha n)$-arc $k$-node strong connectivity certificate for an $n$-node digraph whose independence number is $\alpha$ with high probability, where $r=O(\rho^{-2}\log n)$ is the number defined in~\cref{lem:kvc-main}.
\end{corollary}
Taking $\rho = n^{-2/3}$, \cref{cor:congest_knode} gives an $O(n^{2/3}\log n)$-round algorithm that computes a $k$-node strong connectivity certificate.

\subsection{SCC Decomposition}

We follow the approach of \cite{schudy2008} to find SCC decomposition in $O(\log^2 n)$ phases, each consisting of several multi-source reachability queries executed in parallel across disjoint subgraphs.\footnote{It is worth noting that the definition of multi-source reachability in \cite{schudy2008} which we adapt differs from that of \cite{parter_planar}: while the latter requires computing reachability for every pair in $S \times V$, the former only requires identifying the subset of nodes in $V$ reachable from at least one node in $S$.} We show that when the independence number $\alpha$ is small, these queries can be solved efficiently in the CONGEST model using established single-source reachability algorithms \cite{ghaffari2015, jambulapati2019}, thereby obtaining an efficient algorithm for SCC decomposition.

For completeness, we briefly describe the SCC decomposition algorithm of \cite{schudy2008}. Let $R^{+}(S)$ denote the set of nodes reachable from at least one node in $S$, and let $R^{-}(S)$ denote the set of nodes from which at least one node in $S$ is reachable. The algorithm employs a divide-and-conquer strategy: for each subproblem instance $G=(V,E)$, it generates a random permutation $\sigma$ of $V$ and identifies the smallest index $s$ such that the sum of nodes and edges in the subgraph induced by $R^{+}(\{\sigma(1), \dots, \sigma(s)\})$ is at least $(\lvert V\rvert + \lvert E\rvert)/2$. This index $s$ is determined via binary search over $\sigma$, using a multi-source reachability query for each step. 

Once $s$ is identified, the algorithm performs multi-source reachability queries to determine three sets: $A=R^{+}(\{\sigma(1), \dots, \sigma(s-1)\})$, $B=R^{+}(\sigma(s))$, and $C=R^{-}(\sigma(s)) \cap R^{+}(\sigma(s))$, where $R^{-}(\sigma(s))$ is computed by reversing the direction of all edges. As shown in \cite{schudy2008}, the SCC decomposition of $G$ can then be reduced to the decomposition of four disjoint sets: $V \setminus (A \cup B)$, $A \setminus B$, $B \setminus (A \cup C)$, and $(A \cap B) \setminus C$. The set $C$ is itself an SCC containing $\sigma(s)$. The recursion depth is $O(\log n)$ with high probability, resulting in a total of $O(\log^2 n)$ phases, where each phase consists of multi-source reachability queries executed in parallel across disjoint subgraphs.

The remainder of this section describes the implementation of the aforementioned algorithm in the CONGEST model. First, to obtain a random permutation, each node $v \in V$ independently selects a rank uniformly at random from the range $\{1, 2, \dots, n^3\}$. Under this scheme, all ranks are distinct with probability at least $1 - 1/n$. We then perform a binary search over these ranks to identify the minimum rank $t$ such that the total number of nodes and edges in the subgraph induced by $R^{+}(S)$ is at least $(\lvert V\rvert + \lvert E\rvert)/2$, where $S = \{v \in V \mid \text{rank}(v) \le t\}$.

To perform the binary search on $G$, we first execute a standard leader election protocol to select a leader $v \in V$. This leader coordinates the binary search by broadcasting rank thresholds and aggregating the resulting subgraph sizes. Notice that, without loss of generality, the communication network of each subproblem instance $G=(V,E)$ is connected. If a subproblem consists of multiple connected components, they cannot share nodes in the same SCC and can thus be processed independently.  

We construct a BFS tree $T$ rooted at $v$ via flooding. After computing multi-source reachability for a set of sources $S$, each node $u \in V$ broadcasts its membership status in $R^{+}(S)$ to its neighbors. This allows each node to determine the number of edges it contributes to the induced subgraph $G[R^{+}(S)]$. The total count of nodes and edges is then aggregated at $v$ using a convergecast over $T$. Given that the independence number of the original graph is $\alpha$, the diameter $D$ of any connected component of an induced subgraph is bounded by $O(\alpha)$. Consequently, each of these coordination operations—leader election, BFS tree construction, broadcasting, and convergecast—can be completed in $O(D) = O(\alpha)$ rounds.

We now demonstrate that the multi-source reachability problem can be reduced to single-source reachability in the CONGEST model \cite{ghaffari2015, jambulapati2019}. The core idea is to introduce a virtual super-source $\hat{s}$ that connects to all sources in the set $W$. Formally, given $G=(V,E)$ and $W \subseteq V$, let $\hat{G}$ be the augmented graph obtained by adding a new node $\hat{s}$ and directed arcs $(\hat{s}, w)$ for all $w \in W$.

Clearly, the set of nodes reachable from $W$ in $G$, denoted $R^{+}_{G}(W)$, is identical to $R^{+}_{\hat{G}}(\hat{s})$. Since $\hat{s}$ does not exist in the physical communication network, we must ensure it does not act as a relay for messages in the single-source algorithms, as simulating such a node could incur a $O(D)$ round overhead. Both \cite{ghaffari2015} and \cite{jambulapati2019} follow a framework that can be summarized as follows:

\begin{enumerate}[label={Step \arabic*.}, leftmargin=*]
    \item Sample each node with probability $p$. Let $S$ be the set of sampled nodes, and let $S' = S \cup \{\hat{s}\}$.
    \item Construct a skeleton graph $G'=(S', E_{S'})$, where an edge $(u, v)$ exists if there is a path of at most $h = \tilde{O}(1/p)$ edges from $u$ to $v$ in $\hat{G}$. This is computed by simulating $h$ BFS steps from each node in $S'$.
    \item Determine $R^{+}_{G'}(\hat{s})$ within the skeleton graph $G'$.\footnote{If $G'$ is sparse, the algorithm in \cite{ghaffari2015} broadcasts the edges of the skeleton graph to all nodes rather than explicitly simulating BFS on $G'$.}
    \item Broadcast the identities of the nodes in $R^{+}_{G'}(\hat{s})$ to the entire network $G$.
\end{enumerate}

Aside from Step 2, $\hat{s}$ is never used as a relay station. In Step 2, $\hat{s}$ is involved in the BFS simulations and could theoretically be reached by BFS steps originating from other sampled nodes. However, in our construction of $\hat{G}$, $\hat{s}$ has no incoming edges. Consequently, no node in $V$ can reach $\hat{s}$, ensuring it never serves as an intermediate relay for any messages during the simulation.

In summary, the algorithm of \cite{schudy2008} can be effectively simulated in the CONGEST model using $O(\log^2 n)$ phases. Each phase consists of multiple multi-source reachability queries (executed as single-source queries via our reduction) running in parallel across disjoint subgraphs. By combining the single-source reachability bounds from \cite{ghaffari2015, jambulapati2019} with our implementation, we obtain the following complexity result.

\begin{corollary}\label{cor:congest_scc}
    Let $G$ be a digraph whose independence number is $\alpha$. There is an $\tilde O(\sqrt{n}+n^{1/3+o(1)}\alpha^{2/3})$-round algorithm that computes an SCC decomposition for $G$ w.h.p. in the CONGEST model.
\end{corollary}

\subsection{Topological Sorting}
The SCC decomposition algorithm presented in~\cite{schudy2008} also produces a topological ordering of the SCCs. In this section, we show how to implement their method in the CONGEST model, by modifying the SCC decomposition algorithm presented in the previous section. 

For a given digraph $G$, let $T(G)$ be a topological ordering of $G$. Then, using the same definition of $A$, $B$, and $C$ in the previous section,~\cite{schudy2008} shows a crucial property for topological ordering:

\begin{lemma}[{\cite[Lemma 7]{schudy2008}}]\label{lem:schudy_topo}
    The concatenation of $T(G[V \setminus (A \cup B)])$, $T(G[A \setminus B])$, $C$, $T(G[B\setminus (A\cup C)])$, $T(G[(A\cap B)\setminus C])$ is a topological ordering of the SCCs of $G$. 
\end{lemma}

We now describe how to modify the SCC decomposition algorithm to perform topological sorting of the SCCs. Each node $v$ initializes a counter $c_v \coloneqq 1$, which it maintains and updates throughout the recursive process. For each subproblem instance $G=(V,E)$, before splitting into subproblems, the sizes of the five sets—$V \setminus (A \cup B)$, $A \setminus B$, $C$, $B \setminus (A \cup C)$, and $(A \cap B) \setminus C$—are aggregated at the leader of $G$ via convergecasts over $T$. The leader then broadcasts these five values to all nodes in $V$. Each node increases its counter by the total number of nodes in all subproblems that precede its own in the topological ordering:

\begin{itemize}
    \item If $v \in V \setminus (A \cup B)$, $c_v \leftarrow c_v$.
    \item If $v \in A \setminus B$, $c_v \leftarrow c_v + \lvert V \setminus (A \cup B) \rvert$.
    \item If $v \in C$, $c_v \leftarrow c_v + \lvert V \setminus (A \cup B) \rvert + \lvert A \setminus B \rvert$.
    \item If $v \in B \setminus (A \cup C)$, $c_v \leftarrow c_v + \lvert V \setminus (A \cup B) \rvert + \lvert A \setminus B \rvert + \lvert C \rvert$.
    \item If $v \in (A \cap B) \setminus C$, $c_v \leftarrow c_v + \lvert V \setminus (A \cup B) \rvert + \lvert A \setminus B \rvert + \lvert C \rvert + \lvert B \setminus (A \cup C) \rvert$.
\end{itemize}

This recursive division continues until every subproblem consists of a single SCC. Upon termination, the counter $c_v$ for each node $v$ keeps the total number of nodes belonging to all preceding SCCs. Consequently, setting the rank $t_v = c_v$ yields a valid topological sorting of the SCCs. Integrating this ranking scheme with the divide-and-conquer algorithm in the previous section, we obtain the following result:


\begin{corollary}\label{cor:congest_topo}
    Let $G$ be a digraph with independence number $\alpha$. There is an $\tilde O(\sqrt{n}+n^{1/3+o(1)}\alpha^{2/3})$-round algorithm that computes a topological ordering for $G$ w.h.p. in the CONGEST model.
\end{corollary}

%% file: appendix.tex
\section{Additional Notation and Preliminaries}
\label{apx:notations}
\subsection{Streaming Models}
Our model of computation follows the \emph{graph streaming model}~\cite{Muthu05,McGregor14}. The arcs of the input graph $G = (V, A)$ are presented in an order determined by an adversary who has access to the algorithm but not its random seed. The algorithm processes the arcs sequentially, making up to $p$ passes over the input while using at most $s$ space. During each pass, the algorithm can only read the stream forward and cannot backtrack. We refer to such an algorithm as a $p$-pass $s$-space streaming algorithm. If the arcs are static, meaning only arc insertions occur in the stream, the model is called the \emph{insertion-only model}. In contrast, if both arc insertions and deletions are allowed, where deletions can only remove previously inserted arcs, the model is referred to as the \emph{turnstile model}. 

\subsection{Graph Theory}\label{apx:graph}
For a digraph $G=(V,A)$ and a subset of nodes $S\subseteq V$, we denote $G[S]$ the subgraph of $G$ induced by $S$.
For a subset of nodes $V'\subseteq V$, we denote $G\setminus V'$ the subgraph obtained by deleting the nodes in $V'$ and all their incident arcs, i.e. $G\setminus V' = G[V\setminus V']$. For a subset of arcs $A'\subseteq A$, we denote $G\setminus A'$ the subgraph obtained by deleting the arcs in $A'$, i.e. $G\setminus A' = (V,A\setminus A')$. For two digraphs $G=(V,A)$ and $H=(V,A')$ on the same node set $V$, we write
$G\cup H \coloneqq (V,\, A\cup A')$, with duplicate arcs removed.

When we use notions standard for undirected graphs, such as \emph{independence number}, \emph{arboricity}, and \emph{degeneracy}, we mean these quantities for the underlying undirected graph of $G$.     

For a set of nodes $S$, we denote $\delta^{+}_{G}(S)$ (resp. $\delta^{-}_{G}(S)$) the set of incoming (resp. outgoing) edges of $S$ in $G$, the subscript $G$ is omitted when it is clear from context. We use $\delta^{+}_{G}(v)$ in place of $\delta^{+}_{G}(\{v\})$ for a singleton.

For a digraph, a subgraph is spanning if it includes all the nodes in the original digraph.
\begin{definition}[Spanning out-branching/in-branching] \label{def:spanning-branching}
    For a designated node $r\in V$, a \emph{spanning out-branching rooted at $r$} is a subgraph $T=(V,A')$ such that $|\delta^{+}_{T}(v)|=1$ for all $v\in V\setminus\set{r}$. We implicitly assume that the root satisfies $\delta^+_G(r) = \emptyset$. A spanning in-branching is defined similarly with $\delta^+_G(\cdot)$ replaced by $\delta^-_G(\cdot)$.
\end{definition}

A digraph is \emph{strongly connected} if for any two nodes $u$ and $v$, $u$ is reachable from $v$ and $v$ is reachable from $u$. 
For a strongly connected digraph, a \emph{minimum spanning strong subgraph} is a strongly connected spanning subgraph with the minimum number of edges. A \emph{strong bridge} of a digraph $G$ is an arc $e$ whose removal increases the number of strongly connected components of $G$.

The \emph{transitive closure} of a digraph $G=(V,A)$ is a digraph $G^{\textup{tc}}=(V,A^{\textup{tc}})$ such that for every $u,v\in V$, $(u,v)\in A^{\textup{tc}}$ iff $u\reach v$. A \emph{chain} of $G$ is a sequence of nodes which forms a directed path in $G^{\textup{tc}}$, and a \emph{chain cover} is a collection of chains that partitions the node set $V$.
A \emph{topological ordering} of a directed acyclic graph (DAG) $G$ is an ordering $(v_1, v_2, \ldots, v_n)$ of the nodes in $V$ such that $v_j\not\reach v_i$ for any $j > i$.

\subsection{Communication Complexity}
As in much of the streaming lower-bound literature, our hardness results are obtained via reductions from communication problems. We summarize the essentials of communication complexity needed for this work and refer the reader to \cite{KN96,Rou16} for further details.

We focus exclusively on the two-player communication model. For a function $F:\mathcal{X}\times \mathcal{Y}\to \mathcal{Z}$, Alice receives $x\in \mathcal{X}$ and Bob receives $y\in \mathcal{Y}$; their goal is to compute $F(x,y)$ while exchanging as few bits as possible. The randomized communication complexity with error $\delta$, denoted $\Rcc_\delta(F)$, is the minimum number of bits exchanged by any protocol that outputs $F(x,y)$ correctly with probability at least $1-\delta$ on every input. We omit $\delta$ when using the canonical value $\delta=1/3$. For convenience in information-complexity arguments, we work with private randomness; switching to public randomness only incurs a polylogarithmic overhead by Newman's theorem \cite{New91}.

A particular type of communication problem, namely \emph{decomposable function}, will be central to our communication-to-streaming reductions.
\begin{definition}[Decomposable function] \label{def:decomposable}
    $F:(\mathcal{X}\times \mathcal{Y})^n\to \set{0,1}$ is a decomposable function if it can be written as $F(\vec{x},\vec{y}) = g(h(x_1,y_1),\ldots,h(x_n,y_n))$ for some $g:\set{0,1}^n\to \set{0,1}$ and $h:\mathcal{X}\times \mathcal{Y}\to \set{0,1}$. We also use the notation $F=g\circ^n h$ for the above decomposition.
\end{definition}

\section{Formal Definition of Problems}\label{apx:prob-def}
We provide the formal definitions of the graph problems studied in this work. In each of these problems, the input digraph $G$ has $n$ nodes.
\begin{itemize}
    \item $3\text{-}\CYCLE$: decide whether $G$ contained a directed 3-cycle.
    \item $\REACH$: for a given source $s$ and sink $t$, decide whether there is a directed path from $s$ to $t$.
    \item $\PairReach$: on $n$ pairs of nodes $\set{(u_i,v_i)}_{i=1}^n$, output $w\in\set{0,1}^n$, where $w_i=1$ iff there is a directed path from $u_i$ to $v_i$.
\end{itemize}
The following problems are search problems. For problems where a solution may not exist, a valid algorithm should return $\perp$. With a slight abuse of notation, we sometimes use the same symbol for both the search problem and its decision version, the intended meaning being clear from context.
\begin{itemize}
    \item $\HamCycle$: output a Hamiltonian cycle of $G$ if exists.
    \item $\HamPath$: output a Hamiltonian path of $G$ if exists.
    \item $\OutBranch$: output a spanning out-branching of $G$ rooted at a designated root $r$ if exists.
    \item $\InBranch$: output a spanning in-branching of $G$ rooted at a designated root $r$ if exists.
    \item $k\text{-}\DisjBranch$: output $k$ arc-disjoint spanning out-branchings of $G$ rooted at a mutual designated node $r$ if exist.
    \item $k\text{-}\AConnCert$: output a $k$-arc strong connectivity certificate.
    \item $k\text{-}\NConnCert$: output a $k$-node strong connectivity certificate. In view of \cref{thm:node-vs-arc}, we also use $k\text{-}\ConnCert$ when we do not specify for node-connectivity or arc-connectivity.
    \item $\SCC$: output an SCC decomposition of $G$.
    \item $\MSSS$ : output a minimum spanning strong subgraph of $G$ if exists.
    \item $\StrBridge$: output all strong bridges of $G$.
    \item $2\text{-}\indBranch$: output a pair of $r$-independent spanning out-branchings of $G$ rooted at a mutual designated root $r$ if exist.
    \item $\TranClos$: output all arcs of the transitive closure of $G$. Notice that the output may consist of more arcs than the original graph, so the output graph is assumed to be stored in a separate write-only space.
\end{itemize}
The following problems require the input digraph $G$ to be acyclic:
\begin{itemize}
    \item $\TopoSort$: output a topological ordering of $G$.
    \item $\MCC$: output a minimum chain cover of $G$.
\end{itemize}
The following problem requires the input digraph $G$ to be strongly-connected:
\begin{itemize}
    \item $d\text{-}\DistDom$: output an $O(n/d)$-node distance-$d$ dominating set of $G$.
\end{itemize}

In the following problem, for a parameter $M>1$, the input is an $n$-node weighted digraph $G$ with arc weights between 1 and $M$. In the data stream, an arc is inserted together with its weight.
\begin{itemize}
    \item $\MaxW\OutBranch^{[M]}$: output the spanning out-branching of $G$ rooted at a designated root $r$ (if exists) with the maximum total weight.
\end{itemize}
The problems $\MaxW\InBranch^{[M]}$, $\MinW\OutBranch^{[M]}$, $\MinW\OutBranch^{[M]}$ are defined similarly.

In the following problem, the input is an $n$-node undirected edge-coloured graph $G$. In the data stream, an edge is inserted together with its colour.
\begin{itemize}
    \item $\MaxColorTree$: output a spanning tree of $G$ together with the edge colours, such that the number of colours used is maximized.
\end{itemize}

\subsection{Additional notations for Streaming Problems}
For an optimization problem that outputs an optimal object $C^*$, and a ratio $\rho\in [1,\infty)$, we define the \emph{$\rho$-approximation} to be the problem of returning an object $C$ whose objective value, $OBJ(C)$, is within a multiplicative factor of $\rho$ from $OBJ(C^*)$. Obviously, there is always one side of the range equal to $OBJ(C^*)$; our definition here does not distinguish between a minimization or maximization problem. 
We append the suffix ``$r\apx$'' for the corresponding approximation problem.

We append $n$ as a subscript to a problem whenever we want to emphasize the number of nodes in the digraph.

We append $\alpha$ as a superscript to a problem to the restricted problem that the independence number of the input graph is at most $\alpha$. 
For example, $\HamPath_n$ refers to finding a Hamiltonian path on an $n$-node graph, and $\HamPath_n^\alpha$ refers to finding a Hamiltonian path on an $n$-node graph with independence number at most $\alpha$.

\subsection{2-Satisfiability}
The non-graph problem of $\SAT$ can be modelled as a digraph problem and thus fits within the independence-number parameterization framework. Recall that an $n$-variable $\SAT$ instance is a boolean formula $F$ in the form of $\bigwedge^m_{i=1} (\ell_{1,i}\vee \ell_{2,i})$, where each $\ell_{b,i}$ is a literal of one of the $n$ variables. 
A $\SAT$ instance $F$ is associated with the \emph{implication graph} $G_F$.
$G_F$ consists of $x$ and $\neg x$ for each variable in $F$, giving a total of $2n$ nodes. 
For each clause $(\ell\vee \ell')$ of $F$ on two literals $\ell$ and $\ell'$, construct two arcs $(\neg \ell,\ell')$ and $(\neg \ell', \ell)$.

Naturally, we can extend the notion of independence to $\SAT$ formulas. We define the \emph{stability number} of a $\SAT$ formula $F$ as the independence number of $G_F$.
Following the notations of the other graph streaming problems, we denote: 
\begin{itemize}
    \item $\SAT_n$: output a satisfiable assignment (if one exists) for an $n$-variable $\SAT$ formula.
\end{itemize}
We append the superscript $\alpha$ to indicate the restriction that the stability number is at most $\alpha$.

\section{Streaming Algorithm for $k$-Arc Strong Connectivity Certificates}\label{sec:karc}
In this section, we present a streaming algorithm that computes a $k$-arc strong connectivity certificate directly. As noted in the remark succeeding \cref{thm:node-vs-arc}, our $k$-node strong connectivity certificate algorithm (\cref{algo:kconn}) outperforms the algorithm described here. Thus, \cref{thm:node-vs-arc} implies that a more efficient $k$-arc strong connectivity certificate can be obtained simply by using the $k$-node strong connectivity approach.

The main idea for the arc-strong connectivity certificate algorithm is similar to that of the node-strong connectivity algorithm: adopting the sampling approach \cite{GuhaMT15,AssadiS23} reduces the task of computing a $k$-arc strong connectivity certificate to multiple independent instances of $1$-arc strong connectivity certificates, which can then be executed in parallel. 
Since our goal is to compute an arc-strong connectivity certificate, the most natural modification to \cref{algo:kconn} is to replace the independent node samples with independent arc samples. For the sampling probability parameter $\rho$, we denote $A_i(\rho)$ the $i$-th random subset of $A(G)$ obtained by including each $a\in A(G)$ with probability $\rho$ and independent from other $A_j(\rho)$'s. The algorithm is formalized as follows. 

\begin{figure}[hbtp!]
\begin{algorithm}[H]
    \setlength{\leftskip}{0.75cm}
    \setlength{\parindent}{-0.75cm}
    \caption{$k$-arc Ctrong Connectivity Certificate Algorithm}\label{algo:karcconn}
    \KwIn{A digraph $G=(V,A)$, sampling probability $\rho$, depth parameter $p$}
    Set $r'=O(\rho^{-1}\log n)$\;
    Sample $A_1(\rho),\ldots,A_{r'}(\rho)$ independently\;
    \For{$i=1,\ldots,r'$}{
        Apply {\OneConnCertAlgo} (\cref{algo:1conn}) of depth $p$ on $G_i=(V(G),A_i(\rho))$ to obtain a $1$-strong connectivity certificate $Q_i$\; 
    }
    \KwRet{$Q = \bigcup_{i=1}^{r'} Q_i$}\;
\end{algorithm}
\caption{Randomized algorithm for $k$-arc strong connectivity certificate}
\end{figure}

As before, one can adapt the algorithm into the turnstile model by incorporating the multi-ary search described in \cref{sec:turnstile}. The same argument applies here, so we omit the discussion for the turnstile model.

\begin{lemma}\label{lem:arc-algo}
Let $G=(V,A)$ be an $n$-node digraph with $\ind(G)\leq \alpha$. For $p\in \dbN$ and $k \ge 2$, with sampling probability $\rho=1/k$, \cref{algo:karcconn} computes a $k$-arc strong connectivity certificate of $G$ of size $O(k^2\alpha n\log^2 n)$ with probability $1 - 1/n^{\Omega(1)}$ by a $p$-pass randomized algorithm using:
\begin{enumerate}[label=\textup{(\alph*)}]
    \item $O(k^2 \alpha n^{1+1/p} \log^2 n )$ space in the insertion-only model; and
    \item $O(k^2 \alpha n^{1+O(1/\sqrt{p})} \log^2 n )$ space in the turnstile model.
\end{enumerate}
\end{lemma}

The key technical claim for \cref{algo:karcconn} is similar to \cref{clm:kvc}.
\begin{claim}\label{clm:kac}
    Let $k\geq 2$, $\rho\in (0,1/k]$ and $n$ be sufficiently large. Write $r'=\frac{192\lambda}{\rho}\log n$ for some $\lambda\geq 1$. With probability $1 - n^{-\lambda}$, every arc $(a,b)\in A(G)$ satisfies the following properties:
    \begin{enumerate}[label=\textup{(\alph*)}]
        \item\label{prop.a-kac} If $\lambda_{ab}(G) \ge 2k$, then $\lambda_{ab}(Q) \ge k$;
        \item\label{prop.b-kac} If $\lambda_{ab}(G) < 2k$, then $(a, b) \in A(Q)$.
    \end{enumerate}
\end{claim}
This claim is proven with essentially the same probabilistic calculations as in \cref{clm:R-start,clm:R-middle} but with a much simplified case analysis. We omit the details here.

The crucial difference between the arc-connectivity and node-connectivity cases, which ultimately accounts for the inferior performance of \cref{algo:karcconn}, is that, unlike the induced subgraphs $G[V_i(\rho)]$ sampled in the node-strong connectivity certificate algorithm (\cref{algo:kconn}), the random arc samples in \cref{algo:karcconn} create graphs $G_i$ with potentially higher independence numbers.

Through a careful analysis, we can show that the independence number of $G_i$ remains bounded with high probability, even if it may exceed $\alpha$. This upper bound on the independence number is essential for the efficiency of executing the $1$-strong connectivity certificate algorithm.
\begin{lemma}\label{lem:alphaincrease}
    Let $G$ be an $n$-node \emph{undirected} graph with independence number $\alpha$. For $\rho\in (0,1)$, let $G(\rho)$ be a random subgraph of $G$ such that each arc in $G$ is included in $G(\rho)$ with probability $\rho$ independently. Then, with probability $1-1/n^{\Omega(1)}$, $G(\rho)$ has independence number $O(\alpha\rho^{-1} \log n)$.
\end{lemma}
\begin{proof}
    Let $B$ be a $\beta$-vertex subset of vertices in $G$. We claim that, if $\beta \ge \alpha_{\Delta} \coloneqq \alpha(1+\rho^{-1} \log n)$, then $B$ is not an independent set in $G(\rho)$ with an overwhelming probability.

Let $d_B$ denote the average degree of the induced subgraph $G[B]$. By Tur\'{a}n's Theorem (see the second formulation in \cite{TV07}), every $n$-vertex $m$-edge graph contains an independent set with at least $\frac{n}{1+2m/n}$ vertices. Together with $\ind(G)=\alpha$ gives an upper bound on $\ind(G[B])$:
\[
    \alpha \ge \frac{\beta}{1+d_{B}} \ge \frac{4 \alpha\rho^{-1} \log n + \alpha}{(1+d_B)}.
\]
Thus, $d_B \ge 4\rho^{-1} \log n$. Equivalently, the number $m_B$ of edges in $G[B]$ is at least $2 \beta\rho^{-1} \log n$. The probability that $G[B]$ forms an independent set in $G(\rho)$ is 
\[
    (1-\rho)^{m_B} \le \exp\left(-\rho\cdot \frac{2\beta \log n}{\rho}\right) = \frac{1}{n^{2\beta}}. 
\]
Taking an union bound over all $\beta$-node $B$ for $\beta \ge \alpha_\Delta$, the probability that $G(\rho)$ contains an independent set containing at least $\alpha_{\Delta}$ nodes is at most
\[
    \sum_{\beta = \alpha_\Delta}^n \binom{n}{\beta} \frac{1}{n^{2\beta}} \le \sum_{\beta = \alpha_\Delta}^n \frac{1}{n^{\beta}} \le \frac{2}{n^{\alpha_\Delta}} \le \frac{2}{n^{\alpha}}. 
\]
This completes the proof.
\end{proof}
\cref{lem:alphaincrease} also applies to digraphs by simply ignoring arc orientations. Note that for any two adjacent vertices in a digraph $G$, there can be up to two edges in the corresponding undirected graph $G'$. Consequently, in a stochastic sense, $\ind(G(\rho)) \le \ind(G'(\rho))$.

Piecing \cref{clm:kac,lem:alphaincrease} together, for $\rho=1/k$, \cref{algo:karcconn} runs {\OneConnCertAlgo} on $r'=O(k\log n)$ independent sample arc samples. With high probability, \cref{clm:kac} holds and the independence numbers of all $G_i$'s are bounded by $O(k\alpha\log n)$, so the overall space usage is
\[
    r'\cdot O(k\alpha\log n \cdot n^{1+1/p}) = O(k^2\alpha n^{1+1/p}\log n),
\]
and the certificate size is $r'\cdot O(k\alpha\log n\cdot n) = O(k^2\alpha n\log^2 n)$.


%% file: ref.bib
@inproceedings{AssadiS23,
  author       = {Sepehr Assadi and
                  Vihan Shah},
  title        = {Tight Bounds for Vertex Connectivity in Dynamic Streams},
  booktitle    = {6th Symposium on Simplicity in Algorithms {(SOSA)}},
  pages        = {213--227},
  publisher    = {{SIAM}},
  year         = {2023},
}

@inproceedings{SunW15,
  author       = {Xiaoming Sun and
                  David P. Woodruff},
  title        = {Tight Bounds for Graph Problems in Insertion Streams},
  booktitle    = {Approximation, Randomization, and Combinatorial Optimization. Algorithms
                  and Techniques {(APPROX/RANDOM)}},
  series       = {LIPIcs},
  volume       = {40},
  pages        = {435--448},
  publisher    = {Schloss Dagstuhl - Leibniz-Zentrum f{\"{u}}r Informatik},
  year         = {2015},
}

@inproceedings{ChangFHT20,
  author       = {Yi{-}Jun Chang and
                  Martin Farach{-}Colton and
                  Tsan{-}sheng Hsu and
                  Meng{-}Tsung Tsai},
  title        = {Streaming Complexity of Spanning Tree Computation},
  booktitle    = {37th International Symposium on Theoretical Aspects of Computer Science,
                  {(STACS)}},
  series       = {LIPIcs},
  volume       = {154},
  pages        = {34:1--34:19},
  publisher    = {Schloss Dagstuhl - Leibniz-Zentrum f{\"{u}}r Informatik},
  year         = {2020},
}

@article{CheriyanKT93,
  author       = {Joseph Cheriyan and
                  Ming{-}Yang Kao and
                  Ramakrishna Thurimella},
  title        = {Scan-First Search and Sparse Certificates: An Improved Parallel Algorithms for $k$-Vertex Connectivity},
  journal      = {{SIAM} J. Comput.},
  volume       = {22},
  number       = {1},
  pages        = {157--174},
  year         = {1993},
}

@article{EppsteinGIN97,
  author       = {David Eppstein and
                  Zvi Galil and
                  Giuseppe F. Italiano and
                  Amnon Nissenzweig},
  title        = {Sparsification - a technique for speeding up dynamic graph algorithms},
  journal      = {J. {ACM}},
  volume       = {44},
  number       = {5},
  pages        = {669--696},
  year         = {1997},
}

@inproceedings{GuhaMT15,
  author       = {Sudipto Guha and
                  Andrew McGregor and
                  David Tench},
  title        = {Vertex and Hyperedge Connectivity in Dynamic Graph Streams},
  booktitle    = {Proceedings of the 34th {ACM} Symposium on Principles of Database Systems {(PODS)}},
  pages        = {241--247},
  publisher    = {{ACM}},
  year         = {2015},
}

@book{TV07,
  author       = {Terence Tao and
                  Van H. Vu},
  title        = {Additive combinatorics},
  series       = {Cambridge studies in advanced mathematics},
  volume       = {105},
  publisher    = {Cambridge University Press},
  year         = {2007},
}

@book{Diestel12,
  author       = {Reinhard Diestel},
  title        = {Graph Theory, 4th Edition},
  series       = {Graduate texts in mathematics},
  volume       = {173},
  publisher    = {Springer},
  year         = {2012},
}

@inproceedings{CLT25,
  author       = {Ho-Lin Chen and Peng-Ting Lin and Meng-Tsung Tsai},
  title        = {Parameterized Streaming Algorithms for Topological Sorting},
  booktitle    = {Algorithms and Data Structures - 19th International Symposium {(WADS)}},
  series       = {LIPIcs},
  pages        = {18:1-18:20},
  publisher    = {Schloss Dagstuhl - Leibniz-Zentrum f{\"{u}}r Informatik},
  year         = {2025},
}

@article{GargJS23,
  author       = {Paritosh Garg and
                  Linus Jordan and
                  Ola Svensson},
  title        = {Semi-streaming algorithms for submodular matroid intersection},
  journal      = {Math. Program.},
  volume       = {197},
  number       = {2},
  pages        = {967--990},
  year         = {2023},
}

@inproceedings{CrouchS14,
  author       = {Michael S. Crouch and
                  Daniel M. Stubbs},
  title        = {Improved Streaming Algorithms for Weighted Matching, via Unweighted
                  Matching},
  booktitle    = {Approximation, Randomization, and Combinatorial Optimization. Algorithms
                  and Techniques {(APPROX/RANDOM)}},
  series       = {LIPIcs},
  volume       = {28},
  pages        = {96--104},
  publisher    = {Schloss Dagstuhl - Leibniz-Zentrum f{\"{u}}r Informatik},
  year         = {2014},
}

@inproceedings{CKP+21,
  author       = {Lijie Chen and
                  Gillat Kol and
                  Dmitry Paramonov and
                  Raghuvansh R. Saxena and
                  Zhao Song and
                  Huacheng Yu},
  title        = {Almost optimal super-constant-pass streaming lower bounds for reachability},
  booktitle    = {53rd Annual {ACM} {SIGACT} Symposium on Theory of Computing {(STOC)} },
  pages        = {570--583},
  publisher    = {{ACM}},
  year         = {2021},
}

@incollection{Edmonds73,
  author    = {J. Edmonds},
  title     = {Edge-disjoint Branchings},
  booktitle = {Combinatorial Algorithms},
  pages     = {91--96},
  publisher = {Academic Press},
  address   = {New York},
  year      = {1973}
}

@article{Menger27,
  author  = {K. Menger},
  title   = {Zur allgemeinen Kurventheorie},
  journal = {Fund. Math.},
  volume  = {10},
  pages   = {96--115},
  year    = {1927},
  language = {German}
}

@inproceedings{AhnGM12,
  author       = {Kook Jin Ahn and
                  Sudipto Guha and
                  Andrew McGregor},
  title        = {Analyzing graph structure via linear measurements},
  booktitle    = {Proceedings of the Twenty-Third Annual {ACM-SIAM} Symposium on Discrete
                  Algorithms {(SODA)}},
  pages        = {459--467},
  publisher    = {{SIAM}},
  year         = {2012},
}

@article{McGregor14,
  author    = {Andrew McGregor},
  title     = {Graph stream algorithms: a survey},
  journal   = {{SIGMOD} Rec.},
  volume    = {43},
  number    = {1},
  pages     = {9--20},
  year      = {2014},
}

@inproceedings{ABB+19,
  author       = {Akanksha Agrawal and
                  Arindam Biswas and
                  {\'{E}}douard Bonnet and
                  Nick Brettell and
                  Radu Curticapean and
                  D{\'{a}}niel Marx and
                  Tillmann Miltzow and
                  Venkatesh Raman and
                  Saket Saurabh},
  title        = {Parameterized Streaming Algorithms for {M}in-{O}nes d-{SAT}},
  booktitle    = {39th Annual Conference on Foundations of Software Technology
                  and Theoretical Computer Science {(FSTTCS)}},
  series       = {LIPIcs},
  volume       = {150},
  pages        = {8:1--8:20},
  publisher    = {Schloss Dagstuhl - Leibniz-Zentrum f{\"{u}}r Informatik},
  year         = {2019},
}

@inproceedings{CGM+20,
  author       = {Amit Chakrabarti and
                  Prantar Ghosh and
                  Andrew McGregor and
                  Sofya Vorotnikova},
  title        = {Vertex Ordering Problems in Directed Graph Streams},
  booktitle    = {Proceedings of the 2020 {ACM-SIAM} Symposium on Discrete Algorithms {(SODA)}},
  pages        = {1786--1802},
  publisher    = {{SIAM}},
  year         = {2020},
}

@inproceedings{BJW22,
  author       = {Anubhav Baweja and
                  Justin Jia and
                  David P. Woodruff},
  title        = {An Efficient Semi-Streaming {PTAS} for Tournament Feedback Arc Set
                  with Few Passes},
  booktitle    = {13th Innovations in Theoretical Computer Science Conference {(ITCS)}},
  series       = {LIPIcs},
  volume       = {215},
  pages        = {16:1--16:23},
  publisher    = {Schloss Dagstuhl - Leibniz-Zentrum f{\"{u}}r Informatik},
  year         = {2022},
}

@article{FJ81,
  author       = {Greg N. Frederickson and
                  Joseph F. J{\'{a}}J{\'{a}}},
  title        = {Approximation Algorithms for Several Graph Augmentation Problems},
  journal      = {{SIAM} J. Comput.},
  volume       = {10},
  number       = {2},
  pages        = {270--283},
  year         = {1981},
}

@inproceedings{GK24,
  author       = {Prantar Ghosh and
                  Sahil Kuchlous},
  title        = {New Algorithms and Lower Bounds for Streaming Tournaments},
  booktitle    = {32nd Annual European Symposium on Algorithms {(ESA)}},
  series       = {LIPIcs},
  volume       = {308},
  pages        = {60:1--60:19},
  publisher    = {Schloss Dagstuhl - Leibniz-Zentrum f{\"{u}}r Informatik},
  year         = {2024},
}

@article{DBLP:journals/dmgt/BroersmaL97,
  author       = {Hajo Broersma and
                  Xueliang Li},
  title        = {Spanning trees with many or few colors in edge-colored graphs},
  journal      = {Discuss. Math. Graph Theory},
  volume       = {17},
  number       = {2},
  pages        = {259--269},
  year         = {1997},
  url          = {https://doi.org/10.7151/dmgt.1053},
  doi          = {10.7151/DMGT.1053},
  bibsource    = {dblp computer science bibliography, https://dblp.org}
}

@article{Muthu05,
 author = {Muthukrishnan, S.},
 title = {Data Streams: Algorithms and Applications},
 journal = {Found. Trends Theor. Comput. Sci.},
 issue_date = {August 2005},
 volume = {1},
 number = {2},
 month = aug,
 year = {2005},
 issn = {1551-305X},
 pages = {117--236},
 numpages = {120},
 publisher = {Now Publishers Inc.},
 address = {Hanover, MA, USA},
}

@book{garey2002computers,
  title={Computers and intractability},
  author={Garey, Michael R. and Johnson, David S.},
  volume={29},
  year={2002},
  publisher={W. H. Freeman and Company}
}

@inproceedings{vetta2001approximating,
  author       = {Adrian Vetta},
  editor       = {S. Rao Kosaraju},
  title        = {Approximating the minimum strongly connected subgraph via a matching
                  lower bound},
  booktitle    = {Proceedings of the Twelfth Annual Symposium on Discrete Algorithms {(SODA)}},
  pages        = {417--426},
  publisher    = {{ACM/SIAM}},
  year         = {2001},
}

@inproceedings{bacrach2019hardness,
  title={Hardness of distributed optimization},
  author={Bacrach, Nir and Censor-Hillel, Keren and Dory, Michal and Efron, Yuval and Leitersdorf, Dean and Paz, Ami},
  booktitle={Proceedings of the 2019 ACM Symposium on Principles of Distributed Computing {(PODC)}},
  pages={238--247},
  year={2019}
}

@InProceedings{Ter25,
  author =	{Terao, Tatsuya},
  title =	{{Deterministic $(2/3 - \epsilon)$-Approximation of Matroid Intersection Using Nearly-Linear Independence-Oracle Queries}},
  booktitle =	{19th International Symposium on Algorithms and Data Structures {(WADS)}},
  pages =	{50:1--50:18},
  series =	{Leibniz International Proceedings in Informatics (LIPIcs)},
  year =	{2025},
  volume =	{349},
  publisher =	{Schloss Dagstuhl -- Leibniz-Zentrum f{\"u}r Informatik},
}

@article{GO16,
  title = {Superlinear Lower Bounds for Multipass Graph Processing},
  volume = {76},
  ISSN = {1432-0541},
  url = {http://dx.doi.org/10.1007/s00453-016-0138-7},
  DOI = {10.1007/s00453-016-0138-7},
  number = {3},
  journal = {Algorithmica},
  publisher = {Springer Science and Business Media LLC},
  author = {Guruswami,  Venkatesan and Onak,  Krzysztof},
  year = {2016},
  month = mar,
  pages = {654–683}
}

@inproceedings{BJKS02,
author = {{Bar-Yossef}, Ziv and Jayram, T. S. and Kumar, Ravi and Sivakumar, D.},
title = {An Information Statistics Approach to Data Stream and Communication Complexity},
year = {2002},
isbn = {0769518222},
publisher = {IEEE Computer Society},
booktitle = {Proceedings of the 43rd Symposium on Foundations of Computer Science {(FOCS)}},
pages = {209–218},
numpages = {10},
}

@inproceedings{jayram2003two,
  title={Two applications of information complexity},
  author={Jayram, Thathachar S and Kumar, Ravi and Sivakumar, D},
  booktitle={Proceedings of the thirty-fifth annual ACM symposium on Theory of computing},
  pages={673--682},
  year={2003}
}

@inbook{BOV13,
  title = {How Hard Is Counting Triangles in the Streaming Model?},
  ISBN = {9783642392061},
  ISSN = {1611-3349},
  url = {http://dx.doi.org/10.1007/978-3-642-39206-1_21},
  DOI = {10.1007/978-3-642-39206-1_21},
  booktitle = {Automata,  Languages,  and Programming},
  publisher = {Springer Berlin Heidelberg},
  author = {Braverman,  Vladimir and Ostrovsky,  Rafail and Vilenchik,  Dan},
  year = {2013},
  pages = {244–254}
}

@book{CT05,
  title = {Elements of Information Theory},
  ISBN = {9780471748823},
  url = {http://dx.doi.org/10.1002/047174882X},
  DOI = {10.1002/047174882x},
  publisher = {Wiley},
  author = {Cover,  Thomas M. and Thomas,  Joy A.},
  year = {2005},
  month = apr 
}

@article{Wei15,
author = {Weinstein, O.},
title = {Information Complexity and the Quest for Interactive Compression},
year = {2015},
issue_date = {June 2015},
publisher = {Association for Computing Machinery},
address = {New York, NY, USA},
volume = {46},
number = {2},
issn = {0163-5700},
url = {https://doi.org/10.1145/2789149.2789161},
doi = {10.1145/2789149.2789161},
journal = {SIGACT News},
month = jun,
pages = {41–64},
numpages = {24}
}

@book{KN96,
  title = {Communication Complexity},
  ISBN = {9780511574948},
  url = {http://dx.doi.org/10.1017/CBO9780511574948},
  DOI = {10.1017/cbo9780511574948},
  publisher = {Cambridge University Press},
  author = {Kushilevitz,  Eyal and Nisan,  Noam},
  year = {1996},
  month = dec 
}

@article{Rou16,
author = {Roughgarden, Tim},
title = {Communication Complexity (for Algorithm Designers)},
year = {2016},
issue_date = {May 2016},
publisher = {Now Publishers Inc.},
address = {Hanover, MA, USA},
volume = {11},
number = {3–4},
issn = {1551-305X},
url = {https://doi.org/10.1561/0400000076},
doi = {10.1561/0400000076},
journal = {Found. Trends Theor. Comput. Sci.},
month = may,
pages = {217–404},
numpages = {188}
}

@article{New91,
  title = {Private vs. common random bits in communication complexity},
  volume = {39},
  ISSN = {0020-0190},
  url = {http://dx.doi.org/10.1016/0020-0190(91)90157-D},
  DOI = {10.1016/0020-0190(91)90157-d},
  number = {2},
  journal = {Information Processing Letters},
  publisher = {Elsevier BV},
  author = {Newman,  Ilan},
  year = {1991},
  month = jul,
  pages = {67–71}
}

@article{italiano2012finding,
  title={Finding strong bridges and strong articulation points in linear time},
  author={Italiano, Giuseppe F and Laura, Luigi and Santaroni, Federico},
  journal={Theoretical Computer Science},
  volume={447},
  pages={74--84},
  year={2012},
  publisher={Elsevier}
}

@article{jones1976new,
  title={New problems complete for nondeterministic log space},
  author={Jones, Neil D and Lien, Y Edmund and Laaser, William T},
  journal={Mathematical systems theory},
  volume={10},
  number={1},
  pages={1--17},
  year={1976},
  publisher={Springer}
}

@book{ford1956maximal,
  author    = {Lester R. Ford and Delbert R. Fulkerson},
  title     = {Maximal Flow through a Network},
  series    = {Canadian Journal of Mathematics},
  volume    = {8},
  pages     = {399--404},
  year      = {1956},
  publisher = {Cambridge University Press}
}

@article{even1975network,
  author    = {Shimon Even and Robert Endre Tarjan},
  title     = {Network Flow and Testing Graph Connectivity},
  journal   = {SIAM Journal on Computing},
  volume    = {4},
  number    = {4},
  pages     = {507--518},
  year      = {1975}
}

@article{lipton1979separator,
  title={A separator theorem for planar graphs},
  author={Lipton, Richard J and Tarjan, Robert Endre},
  journal={SIAM Journal on Applied Mathematics},
  volume={36},
  number={2},
  pages={177--189},
  year={1979},
  publisher={SIAM}
}

@inproceedings{leighton1988approximate,
  author    = {F. Thomson Leighton and Satish Rao},
  title     = {An Approximate Max-Flow Min-Cut Theorem for Uniform Multicommodity Flow Problems with Applications to Approximation Algorithms},
  booktitle = {Proceedings of the 29th Annual Symposium on Foundations of Computer Science (FOCS)},
  pages     = {422--431},
  year      = {1988}
}

@article{OvL24,
  title = {Parameterized Complexity of Streaming Diameter and Connectivity Problems},
  volume = {86},
  ISSN = {1432-0541},
  url = {http://dx.doi.org/10.1007/s00453-024-01246-z},
  DOI = {10.1007/s00453-024-01246-z},
  number = {9},
  journal = {Algorithmica},
  publisher = {Springer Science and Business Media LLC},
  author = {Oostveen,  Jelle J. and van Leeuwen,  Erik Jan},
  year = {2024},
  month = jun,
  pages = {2885–2928}
}

@inproceedings{AY19,
author = {Alman, Josh and Yu, Huacheng},
title = {Faster update time for turnstile streaming algorithms},
year = {2020},
publisher = {Society for Industrial and Applied Mathematics},
address = {USA},
booktitle = {Proceedings of the Thirty-First Annual ACM-SIAM Symposium on Discrete Algorithms},
pages = {1803–1813},
numpages = {11},
location = {Salt Lake City, Utah},
series = {SODA '20}
}

@article{CFR10,
author = {Coppersmith, Don and Fleischer, Lisa K. and Rurda, Atri},
title = {Ordering by weighted number of wins gives a good ranking for weighted tournaments},
year = {2010},
issue_date = {June 2010},
publisher = {Association for Computing Machinery},
address = {New York, NY, USA},
volume = {6},
number = {3},
issn = {1549-6325},
url = {https://doi.org/10.1145/1798596.1798608},
journal = {ACM Trans. Algorithms},
month = jul,
articleno = {55},
numpages = {13}
}

@inproceedings{nickelsen2002reachability,
  title={On reachability in graphs with bounded independence number},
  author={Nickelsen, Arfst and Tantau, Till},
  booktitle={International Computing and Combinatorics Conference},
  pages={554--563},
  year={2002},
  organization={Springer}
}

@article{fradkin2015edge,
  title={Edge-disjoint paths in digraphs with bounded independence number},
  author={Fradkin, Alexandra and Seymour, Paul},
  journal={Journal of Combinatorial Theory, Series B},
  volume={110},
  pages={19--46},
  year={2015},
  publisher={Elsevier}
}

@inproceedings{Chitnis2014,
  title = {Parameterized Streaming: Maximal Matching and Vertex Cover},
  url = {http://dx.doi.org/10.1137/1.9781611973730.82},
  DOI = {10.1137/1.9781611973730.82},
  booktitle = {Proceedings of the Twenty-Sixth Annual ACM-SIAM Symposium on Discrete Algorithms},
  publisher = {Society for Industrial and Applied Mathematics},
  author = {Chitnis,  Rajesh and Cormode,  Graham and Hajiaghayi,  Mohammad Taghi and Monemizadeh,  Morteza},
  year = {2014},
  month = dec 
}

@inproceedings{CC19,
  author       = {Rajesh Chitnis and
                  Graham Cormode},
  editor       = {Bart M. P. Jansen and
                  Jan Arne Telle},
  title        = {Towards a Theory of Parameterized Streaming Algorithms},
  booktitle    = {14th International Symposium on Parameterized and Exact Computation,
                  {IPEC} 2019, September 11-13, 2019, Munich, Germany},
  series       = {LIPIcs},
  volume       = {148},
  pages        = {7:1--7:15},
  publisher    = {Schloss Dagstuhl - Leibniz-Zentrum f{\"{u}}r Informatik},
  year         = {2019},
  url          = {https://doi.org/10.4230/LIPIcs.IPEC.2019.7},
  doi          = {10.4230/LIPICS.IPEC.2019.7},
  bibsource    = {dblp computer science bibliography, https://dblp.org}
}

@inbook{Fafianie2014,
  title = {Streaming Kernelization},
  ISBN = {9783662444658},
  ISSN = {1611-3349},
  url = {http://dx.doi.org/10.1007/978-3-662-44465-8_24},
  DOI = {10.1007/978-3-662-44465-8_24},
  booktitle = {Mathematical Foundations of Computer Science 2014},
  publisher = {Springer Berlin Heidelberg},
  author = {Fafianie,  Stefan and Kratsch,  Stefan},
  year = {2014},
  pages = {275–286}
}

@inproceedings{CCE+16,
author = {Chitnis, Rajesh and Cormode, Graham and Esfandiari, Hossein and Hajiaghayi, MohammadTaghi and McGregor, Andrew and Monemizadeh, Morteza and Vorotnikova, Sofya},
title = {Kernelization via sampling with applications to finding matchings and related problems in dynamic graph streams},
year = {2016},
isbn = {9781611974331},
publisher = {Society for Industrial and Applied Mathematics},
address = {USA},
booktitle = {Proceedings of the Twenty-Seventh Annual ACM-SIAM Symposium on Discrete Algorithms},
pages = {1326–1344},
numpages = {19},
location = {Arlington, Virginia},
series = {SODA '16}
}

@article{lokshtanov2018known,
  title={Known algorithms on graphs of bounded treewidth are probably optimal},
  author={Lokshtanov, Daniel and Marx, D{\'a}niel and Saurabh, Saket},
  journal={ACM Transactions on Algorithms (TALG)},
  volume={14},
  number={2},
  pages={1--30},
  year={2018},
  publisher={ACM New York, NY, USA}
}

@inproceedings{BKS02,
  title={Reductions in streaming algorithms, with an application to counting triangles in graphs},
  author={{Bar-Yossef}, Ziv and Kumar, Ravi and Sivakumar, D},
  booktitle={SODA},
  volume={2},
  pages={623--632},
  year={2002}
}

@InProceedings{GS24,
  author =	{Ghosh, Prantar and Shah, Vihan},
  title =	{{New Lower Bounds in Merlin-Arthur Communication and Graph Streaming Verification}},
  booktitle =	{15th Innovations in Theoretical Computer Science Conference (ITCS 2024)},
  pages =	{53:1--53:22},
  series =	{Leibniz International Proceedings in Informatics (LIPIcs)},
  ISBN =	{978-3-95977-309-6},
  ISSN =	{1868-8969},
  year =	{2024},
  volume =	{287},
  editor =	{Guruswami, Venkatesan},
  publisher =	{Schloss Dagstuhl -- Leibniz-Zentrum f{\"u}r Informatik},
  address =	{Dagstuhl, Germany},
  URL =		{https://drops.dagstuhl.de/entities/document/10.4230/LIPIcs.ITCS.2024.53},
  URN =		{urn:nbn:de:0030-drops-195815},
  doi =		{10.4230/LIPIcs.ITCS.2024.53}
}

@article{GM60,
  author    = {T. Gallai and A. N. Milgram},
  title     = {Verallgemeinerung eines graphentheoretischen Satzes von R\'edei},
  journal   = {Acta Sci.\ Math.},
  volume    = {21},
  pages     = {181--186},
  year      = {1960}
}

@inproceedings{parter_planar,
  title={Distributed planar reachability in nearly optimal time},
  author={Parter, Merav},
  booktitle={34th International Symposium on Distributed Computing (DISC 2020)},
  pages={38--1},
  year={2020},
  organization={Schloss Dagstuhl--Leibniz-Zentrum f{\"u}r Informatik}
}

@inproceedings{schudy2008,
  title={Finding strongly connected components in parallel using O (log2 n) reachability queries},
  author={Schudy, Warren},
  booktitle={Proceedings of the twentieth annual symposium on Parallelism in algorithms and architectures},
  pages={146--151},
  year={2008}
}

@inproceedings{ghaffari2015,
  title={Brief announcement: Distributed single-source reachability},
  author={Ghaffari, Mohsen and Udwani, Rajan},
  booktitle={Proceedings of the 2015 ACM Symposium on Principles of Distributed Computing},
  pages={163--165},
  year={2015}
}

@inproceedings{jambulapati2019,
  title={Parallel reachability in almost linear work and square root depth},
  author={Jambulapati, Arun and Liu, Yang P and Sidford, Aaron},
  booktitle={2019 IEEE 60th Annual Symposium on Foundations of Computer Science (FOCS)},
  pages={1664--1686},
  year={2019},
  organization={IEEE}
}

@inproceedings{forster2018faster,
  title={A faster distributed single-source shortest paths algorithm},
  author={Forster, Sebastian and Nanongkai, Danupon},
  booktitle={2018 IEEE 59th Annual Symposium on Foundations of Computer Science (FOCS)},
  pages={686--697},
  year={2018},
  organization={IEEE}
}

@inproceedings{cao2021brief,
  title={Brief announcement: An improved distributed approximate single source shortest paths algorithm},
  author={Cao, Nairen and Fineman, Jeremy T and Russell, Katina},
  booktitle={Proceedings of the 2021 ACM Symposium on Principles of Distributed Computing},
  pages={493--496},
  year={2021}
}

@article{MP80,
  title = {Selection and sorting with limited storage},
  volume = {12},
  ISSN = {0304-3975},
  url = {http://dx.doi.org/10.1016/0304-3975(80)90061-4},
  DOI = {10.1016/0304-3975(80)90061-4},
  number = {3},
  journal = {Theoretical Computer Science},
  publisher = {Elsevier BV},
  author = {Munro,  J.I. and Paterson,  M.S.},
  year = {1980},
  month = nov,
  pages = {315–323}
}

@book{peleg2000distributed,
  title={Distributed computing: a locality-sensitive approach},
  author={Peleg, David},
  year={2000},
  publisher={SIAM}
}

@article{NI92,
  title={A linear-time algorithm for finding a sparse k-connected spanning subgraph of ak-connected graph},
  author={Nagamochi, Hiroshi and Ibaraki, Toshihide},
  journal={Algorithmica},
  volume={7},
  number={1},
  pages={583--596},
  year={1992},
  publisher={Springer}
}

@article{whitty1987,
  title={Vertex-disjoint paths and edge-disjoint branchings in directed graphs},
  author={Whitty, Robin W},
  journal={Journal of Graph Theory},
  volume={11},
  number={3},
  pages={349--358},
  year={1987},
  publisher={Wiley Online Library}
}

@article{huck1995,
  title={Disproof of a conjecture about independent branchings in k-connected directed graphs},
  author={Huck, Andreas},
  journal={Journal of Graph Theory},
  volume={20},
  number={2},
  pages={235--239},
  year={1995},
  publisher={Wiley Online Library}
}

@article{wang2003,
  title={Distributed algorithms for finding the unique minimum distance dominating set in directed split-stars},
  author={Wang, Fu-Hsing and Chang, Jou-Ming and Wang, Yue-Li and Huang, Sun-Jen},
  journal={Journal of Parallel and Distributed Computing},
  volume={63},
  number={4},
  pages={481--487},
  year={2003},
  publisher={Elsevier}
}

@article{joeshi2021,
  title = {A study on distance $k$-domination in digraphs},
  volume = {1770},
  ISSN = {1742-6596},
  url = {http://dx.doi.org/10.1088/1742-6596/1770/1/012073},
  DOI = {10.1088/1742-6596/1770/1/012073},
  number = {1},
  journal = {Journal of Physics: Conference Series},
  publisher = {IOP Publishing},
  author = {Joeshi,  Annie Vetha V. and A. Anto,  Kinsley},
  year = {2021},
  month = mar,
  pages = {012073}
}

@article{meir1975,
  title={Relations between packing and covering numbers of a tree},
  author={Meir, Amram and Moon, John},
  journal={Pacific Journal of Mathematics},
  volume={61},
  number={1},
  pages={225--233},
  year={1975},
  publisher={Mathematical Sciences Publishers}
}
